\documentclass[twocolumn,prb,amsmath,amssymb]{revtex4-2}
\usepackage{amsthm}
\usepackage{booktabs}
\usepackage{array}[=2016-10-06]
\usepackage{graphicx}
\usepackage{lineno}
\usepackage[colorlinks=true,linkcolor=blue,citecolor=blue,urlcolor=blue]{hyperref}

\newtheoremstyle{phys}
  {\medskipamount}{\medskipamount}{\upshape}{0pt}{\bfseries}{.}{0.5em}
  {\thmname{#1}\thmnumber{ #2}\thmnote{ (#3)}}
\theoremstyle{phys}
\newtheorem{theorem}{Theorem}
\newtheorem{lemma}{Lemma}
\newtheorem{corollary}{Corollary}
\newtheorem{proposition}{Proposition}
\newtheorem{result}{Result}
\newtheorem{definition}{Definition}
\newtheorem{remark}{Remark}

\begin{document}

\title{Superconducting Pairing Symmetry on Geometric-Algebra Foundations\\
via the Scalar-Projection Method}

\author{Youping Dai}
\email{daiyouping@cuhk.edu.cn}
\affiliation{School of Science and Engineering, The Chinese University of Hong Kong, Shenzhen, China}


\begin{abstract}
We formulate the low-energy theory of superconducting pairing in the
geometric algebra $Cl_{3,0}$. Conventional BCS mean-field theory and
its extensions ($s$-, $p$-, $d$-wave, etc.) have algebraic counterparts
in an eight-dimensional real multivector space. The formulation proceeds from the geometric-algebra scalar-projection
method; it is an algebraic reconstruction rather than a change of
representation, and is physically equivalent to the standard
complex-matrix formulation. The construction is a two-level projection:
symmetry-space matrix blocks are represented as Clifford multivectors, and the grade-0 projection of
their geometric products extracts the scalar blocks from which Hamiltonians
are built. In this reconstruction, the dichotomy of pairing channels acquires an
algebraic root: it follows from the fermionic identity
$B^\dagger(V)=-B^\dagger(V^{\mathsf T})$, and the Fierz rearrangement
is realized by the map $\Phi:G\mapsto Ge_{31}$ (with $e_{31}^2=-1$). This yields a kernel/matrix dictionary: at the effective
kernel level grades $0\oplus3$ correspond to spin singlets and grades
$1\oplus2$ to triplets; at the pairing-matrix level the singlet plane is
$\mathrm{span}\{e_2,e_{31}\}$. The BCS limit has a double grade identity:
kernel grade 0 with a scalar glue, and vertex grade 2, a bivector
($i\sigma_2\leftrightarrow e_{31}$). We state three algebraic
results and mark their model-level uses: (i) a decomposition theorem for
BdG quantum geometry with a null-texture criterion, motivating a comparative
stiffness discussion of the chiral-state controversy in CsV$_3$Sb$_5$;
(ii) a Fierz sign duality in the four-dimensional Minkowski algebra, with a
chirality-reversal extension to minimal two-node Weyl-semimetal models;
(iii) an error bound and condition-number structure for tomographic
reconstruction under a specified linear-response probe model. Model-level
illustrations follow: a glue-generation mechanism, a one-loop
algebraic renormalization flow whose feeding constant coincides in
magnitude with the $\mathfrak{su}(2)$ Casimir, closed-form $\mu=0$ inter-node spectra with
node splitting, and a UTe$_2$ two-phase analysis.\\

\end{abstract}
\maketitle

\section{Introduction}
\label{sec:intro}

The symmetry of the superconducting pair wavefunction is a long-standing
problem in condensed matter physics. BCS theory~\cite{BCS1957} explained conventional
superconductors through the electron-phonon coupling. The subsequent
discoveries of heavy-fermion superconductors~\cite{Steglich1979},
cuprates~\cite{Bednorz1986}, iron pnictides~\cite{Kamihara2008},
UTe$_2$~\cite{Ran2019}, and magic-angle graphene~\cite{Cao2018} showed that
the electronic collective fluctuations themselves (spin, nematic, charge) can
also serve as the pairing glue. Each material class acquired its own
microscopic model: the RVB/$t$-$J$ lineage~\cite{Anderson1987}, multiorbital
Hubbard models~\cite{Kuroki2009}, the periodic Anderson
model~\cite{Doniach1977},
each proposed for a specific material, with the pairing mechanism usually
attributed after the fact. Nickelate superconductors (the infinite-layer
structure NdNiO$_2$~\cite{Li2019}), kagome metals (the AV$_3$Sb$_5$
family~\cite{Jiang2021,Ortiz2020}), and high-pressure hydrides
(H$_3$S~\cite{Drozdov2015}, LaH$_{10}$~\cite{Somayazulu2019}) have further
extended the map of pairing materials. In the kagome systems, charge density
wave, time-reversal breaking, and superconductivity compete, and the
microscopic character of these orders---the charge order as well as the
pairing symmetry ($s_\pm$, $d$-wave, $d+id$)---remains
contested~\cite{Duan2021,
Saykin2023,Xu2022,Farhang2023,Yoshida2025}; Sec.~\ref{sec:kagome} provides a
comparative quantum-geometry perspective on this controversy. High-pressure hydrides push
$T_c$ toward room temperature, and their strong anharmonicity and multiphase
structure put the traditional Eliashberg theory~\cite{Eliashberg1960} under
pressure. Finally, the development of topological platforms (semiconductor-
superconductor heterostructures~\cite{Mourik2012}, Majorana zero modes in
Fe-based systems~\cite{Wang2018}) has made the connection between pairing
symmetry and topological invariants experimentally accessible~\cite{Read2000,
Qi2011}.

The language of pairing theory is correspondingly fragmented. Electron-boson
vertices are written with Pauli matrices, pairing kernels with complex
functions, order parameters with irreducible representations~\cite{Sigrist1991},
and topological invariants with yet another formalism. The same physical
structure is rewritten in different notations.

The Sigrist--Ueda classification~\cite{Sigrist1991}
enumerates the possible pairing states based on lattice symmetry, but it does not address the dynamics
of mixed states. BCS mean-field theory is weak coupling, whereas
Eliashberg theory~\cite{Eliashberg1960} extends it to stronger coupling
under the Migdal premise; their common structural premise is that the pairing kernel
factorizes into a boson propagator times two fermion vertices. In
strong-correlation regions the four-fermion vertex is not a linear
superposition of single-boson-exchange channels, so that a description
in which distinct microscopic mechanisms are treated as independent
additive contributions is not available.

Clifford algebra, as a mathematical tool, has many applications in condensed
matter physics. Kitaev's spinless $p$-wave chain~\cite{Kitaev2001} and honeycomb
model~\cite{Kitaev2006} brought Majorana (Clifford) representations into
the study of topological phases, and the tenfold
classification~\cite{Kitaev2009} built directly on the classifying spaces
of Clifford algebras. Herbut~\cite{Herbut2010}
identified a topological-insulator sector in graphene superconducting vortex
cores using a Clifford/Dirac matrix structure. Geometric algebra (GA), the tradition founded by Clifford and developed
by Hestenes~\cite{Hestenes1966,Doran2003}, provides a uniform algebraic
grammar for the same matrix content in grade-structured form. We apply it to the pairing channels of superconductivity.

Our strategy is to define the Hamiltonian at the low-energy level through
the \emph{scalar-projection method} rather than through a
microscopic model. Since the microscopic level cannot be unified and no
separable bosonic glue exists in the strong-correlation region, the basic
object of superconducting theory need not be the microscopic Hamiltonian
alone; it can be the most general algebraic structure of a low-energy
effective pairing kernel in spin space, with the grade-0 projection
$\langle\cdot\rangle_0$ as the basic computational rule.

The construction is governed by a two-level projection rule. The first level
is a symmetry-space projection: raw Fock operators are not themselves
elements of $Cl_{3,0}$; their spin-space bilinear and matrix blocks are
losslessly projected into multivectors whose coefficients carry all
remaining content---numbers, functions, or Fock operators kept in their
written order. The second level is grade projection: after the geometric
product implements the internal coupling, $\langle\cdot\rangle_0$
extracts the Hamiltonian scalar. Before mean-field expectation, this
grade-0 scalar may still be an operator on Fock space. Throughout the
paper, what follows from these two operations alone is algebraic;
statements requiring gap-equation dynamics, band structure, or material
interpretation are labeled as model-level or comparative.

The technique is a grade decomposition of the complex pairing matrix:
the same $2\times2$ complex matrix that the Pauli-basis notation
represents in four complex entries is expanded on the eight real blades
of $Cl_{3,0}$, so that the matrix operations (transpose, trace,
Hermitian decomposition, Fierz transformation) become geometric
operations on blades: the channel dichotomy, the Hermitian/anti-Hermitian
split, and the Fierz rearrangement, which the matrix notation handles
through independent algebraic identities, appear as consequences of
a single grade structure. The parameter space is the same as the standard
one; the structure becomes explicit.

Compared with standard methods, the new content of this paper is: (1) the fermionic
identity $B^\dagger(V)=-B^\dagger(V^{\mathsf T})$ as the root of channel
dichotomy, with the Fierz map $\Phi$ made explicit
(Sec.~\ref{sec:channels}); (2) the signature of the pairing form identified
systematically with the $\mathfrak{u}(2)$ trace form (Sec.~\ref{sec:signature});
(3) three results with explicit algebraic/model boundaries: the BdG
quantum-geometry decomposition
(Sec.~\ref{sec:qgt}), the sign duality and its chirality reversal in the
stated minimal models (Sec.~\ref{sec:weyl}), and the conditional
tomographic reconstruction error bound (Sec.~\ref{sec:tomography});
(4) the null-texture criterion and a comparative stiffness discussion for
the kagome flat band
(Sec.~\ref{sec:kagome}) and the two-phase analysis
of UTe$_2$ (Sec.~\ref{sec:ute2}); (5) the algebra of dimensional
extension (Sec.~\ref{sec:extension}). The parameter space itself is equivalent to the standard complex-matrix formulation: the real coordinates are unfolded, but the physical degrees of freedom are not.

The argument of the paper has four layers, and each result below states
which layer it draws on. \emph{Layer A (dictionary)}: the fermionic
identity $B^\dagger(V)=-B^\dagger(V^{\mathsf T})$
(Lemma~\ref{lem:anti}), the Fierz map $\Phi$
(Definition~\ref{def:fierz}), and the channel theorem
(Theorem~\ref{thm:channels}) assign spin and parity to each grade; the
diagonal pairing form (Theorem~\ref{thm:diag}) supplies the metric.
\emph{Layer B (signs)}: the Casimir spectrum (Theorem~\ref{thm:casimir})
and the block decomposition (Theorem~\ref{thm:block}) convert the
signature into interaction signs and channel-selection rules. \emph{Layer
C (realization)}: the scalar projection builds $H_{\mathrm{pair}}$
(Eq.~\eqref{eq:hpair}), and the route through the gap equation
(Eq.~\eqref{eq:gapeq}) to the BdG Hamiltonian turns the dictionary into
spectra. \emph{Layer D (applications)}: each application is anchored to a
layer---the quantum-geometry analysis
(Secs.~\ref{sec:qgt}--\ref{sec:kagome}) uses layer C together with the
channel labels of layer A; the Weyl-semimetal sign duality
(Sec.~\ref{sec:weyl}) uses the blade squares of layer A and the gap
equation of layer C; the tomographic reconstruction
(Sec.~\ref{sec:tomography}) uses the eight-dimensional parameterization of
layer A; and the external-field selection rules (Sec.~\ref{sec:fields})
use the center and commutator structure of layer A. Results that
additionally use a stated dynamical assumption---the selection rules of
Sec.~\ref{sec:block}, the one-loop flow of Sec.~\ref{sec:rg}, and the
material analyses of Secs.~\ref{sec:kagome} and \ref{sec:ute2}---say so
where they appear.

Section~\ref{sec:math} establishes the $Cl_{3,0}$ foundations.
Section~\ref{sec:vertex} classifies electron operators and bosonic vertices.
Section~\ref{sec:channels} builds the Fierz antisymmetrization and the
complex-channel theorem at the kernel and matrix levels. Section~\ref{sec:scalar} generates
the effective pairing interaction and gives the route to BdG.
Sections~\ref{sec:signature} and~\ref{sec:block} develop the signature
structure and channel selection. Section~\ref{sec:bcs} compares the framework
with BCS theory. Section~\ref{sec:examples} collects applications, including
quantum geometry (Sec.~\ref{sec:qgt}), the comparative kagome discussion
(Sec.~\ref{sec:kagome}), and the UTe$_2$ analysis (Sec.~\ref{sec:ute2}).
Section~\ref{sec:phase} discusses the phase structure of the order parameter
and the multiple roles of the pairing glue. Section~\ref{sec:fields} derives
selection rules for external-field-induced channel mixing.
Section~\ref{sec:diagnosis} builds the diagnostic protocol and the tomography
algorithm. Section~\ref{sec:extension} covers dimensional extension,
including the sign duality in Weyl-semimetal pairing.
Section~\ref{sec:transfer} comments on cross-domain transfer, and
Sec.~\ref{sec:conclusion} concludes.

\section{Mathematical foundations: $Cl_{3,0}$ and the diagonal pairing form}
\label{sec:math}

The two-level projection rule fixes the scope of the mathematical results.
At the first level, finite-dimensional symmetry-space blocks (spin, Nambu,
or projected band blocks) are represented as Clifford multivectors, with
all remaining content---numbers, functions, or Fock operators---compressed
into the coefficients and kept in their written order. At the second
level, geometric products among those blocks are reduced by grade-0
projection to the relevant scalar. Results stated purely in terms of
blades are algebraic theorems; results that invoke a gap equation, band
structure, response model, or material assignment are conditional at that
stated level.

\subsection{$Cl_{3,0}$}

The GA~\cite{Hestenes1966,Doran2003} $Cl_{3,0}$ is generated by the orthonormal basis
$\{e_1,e_2,e_3\}$ with
\begin{equation}
\label{eq:clifford}
e_ie_j+e_je_i=2\delta_{ij},\qquad i,j=1,2,3.
\end{equation}
As a real vector space it is eight-dimensional, $2^3=8$, with the grade
decomposition: grade 0 (scalar): $1$; grade 1 (vectors): $e_1,e_2,e_3$;
grade 2 (bivectors): $e_{12}=e_1\wedge e_2$, $e_{23}$, $e_{31}$; grade 3
(pseudoscalar): $I=e_1e_2e_3$. From Eq.~\eqref{eq:clifford} one derives directly
$e_i^2=1$, $e_{ij}^2=-1$ ($i\neq j$), and $I^2=-1$. Any multivector $A$
decomposes uniquely as
\begin{equation}
\begin{split}
A=\langle A\rangle_0+\langle A\rangle_1+\langle A\rangle_2+\langle A\rangle_3\\
=d_0'\cdot1+\mathbf d'\cdot\mathbf e+\mathbf d''\cdot(\mathbf e\wedge
\mathbf e)+d_0''\cdot I,
\end{split}
\end{equation}
where $d_0'$, $\mathbf d'=(d_1',d_2',d_3')$, $\mathbf d''=(d_1'',d_2'',
d_3'')$, and $d_0''$ are the expansion coefficients. The projection $\langle\cdot\rangle_k$ is the standard grade-projection
notation of GA~\cite{Hestenes1966,Doran2003} and extracts the grade-$k$
component (no Dirac bracket is intended). Because
Hamiltonians and free energies must be scalars, we always take them to be the
grade-0 projection.

\begin{theorem}[Parity theorem]
A geometric product of $n$ vectors contains only grades of the same parity
as $n$; in particular
$\langle v_1v_2\cdots v_{2m+1}\rangle_0=0$.
\end{theorem}
\begin{proof}
Induction on $n$. For $n=1$ the statement is immediate (a single vector is
pure grade 1). Suppose it holds for $n$. Multiplication by a new vector is a
sum of an inner product, which lowers the grade by one, and an outer
product, which raises it by one; either way the grade parity flips, and the
parity of a product of $n+1$ vectors is the opposite of that of a product of
$n$. Hence a product of an odd number of vectors contains only odd grades
and has no grade-0 part.
\end{proof}

\subsection{Diagonal pairing form}

\begin{lemma}[Blade orthogonality]
\label{lem:ortho}
Write $e_A$ for an individual basis blade and $g(A)$ for its grade. The
scalar products obey
\begin{equation}
\langle e_Ae_B\rangle_0=\epsilon_{g(A)}\,\delta_{AB},\qquad
\epsilon_k=\begin{cases}+1,&k=0,1,\\-1,&k=2,3,\end{cases}
\label{eq:ortho}
\end{equation}
and each blade squares to a pure scalar,
\begin{equation}
e_A^2=\epsilon_{g(A)},\qquad\text{i.e., } e_i^2=+1,\ \ e_{ij}^2=-1,\ \
I^2=-1 .
\label{eq:square}
\end{equation}
\end{lemma}
\begin{proof}
Both statements follow from $e_i^2=1$ and the anticommutativity
$e_ie_j=-e_je_i$ of distinct basis vectors. For $k=1$:
$\langle e_1e_1\rangle_0=\langle1\rangle_0=1=\epsilon_1$. For $k=2$:
$\langle e_{12}^2\rangle_0=\langle e_1e_2e_1e_2\rangle_0
=\langle-e_1^2e_2^2\rangle_0=-1=\epsilon_2$, and the same swap counting gives
$e_{ij}^2=-1$ and $I^2=-1$, each a pure scalar. For $A\neq B$ the product
$e_Ae_B$ is a nonscalar blade (or a sum of such blades) and its grade-0 part
vanishes; this includes distinct blades of the same grade, for example
$\langle e_1e_2\rangle_0=0$ and $\langle e_{12}e_{23}\rangle_0=0$. This
establishes the Kronecker delta on individual blades in
Eq.~\eqref{eq:ortho}.
\end{proof}

The basis blades fall into four classes by grade: grade 0, the blade $1$
(one blade); grade 1, $e_1,e_2,e_3$ (three); grade 2, $e_{12},e_{23},e_{31}$
(three, with $e_{31}\equiv-e_{13}$); grade 3, $I$ (one). Equivalently, blades of
different grades are orthogonal (zero scalar part), blades of the same grade
are normalized, and the sign depends only on the grade.

\begin{theorem}[Diagonal pairing form]
\label{thm:diag}
For real-coefficient multivectors $X=\sum_Ax_Ae_A$ and $Y=\sum_Ay_Ae_A$,
\begin{equation}
\langle XY\rangle_0=x_0y_0+\mathbf x_1\cdot\mathbf y_1-\mathbf x_2\cdot
\mathbf y_2-x_3y_3,
\label{eq:diag}
\end{equation}
that is, the scalar projection pairs the grades diagonally with signature
$(+,+,-,-)$.
\end{theorem}
\begin{proof}
The scalar projection is bilinear. Expanding $X$ and $Y$ in the blade basis
and applying Lemma~\ref{lem:ortho},
\begin{equation}
\langle XY\rangle_0=\sum_{A,B}x_Ay_B\langle e_Ae_B\rangle_0
=\sum_A\epsilon_{g(A)}x_Ay_A,
\end{equation}
which written out grade by grade is Eq.~\eqref{eq:diag}.
\end{proof}

\begin{remark}[Two metrics]
\label{rem:twometrics}
Throughout, $X$ and $Y$ denote real-coefficient multivectors as in
Theorem~\ref{thm:diag}; a hat $\hat X$ denotes the $2\times2$ matrix image
of $X$ under the isomorphism $e_i\leftrightarrow\sigma_i$,
$I\leftrightarrow i$; and a tilde $\tilde X$ denotes its reversion, the
anti-automorphism that reverses factor order (grade-0 and grade-1 elements
unchanged, grade-2 and grade-3 elements changing sign), so that
$\hat{\tilde X}=\hat X^\dagger$. The signature $(+,+,-,-)$ of
Theorem~\ref{thm:diag} is the sign of
$\langle XY\rangle_0=\mathrm{Re}[\tfrac12\mathrm{Tr}(\hat X\hat Y)]$ and
governs channel arithmetic. The positive-definite norm of the order
parameter is a different, reversion-based object,
$\langle X\tilde X\rangle_0=\tfrac12\mathrm{Tr}(\hat X\hat X^\dagger)\ge0$;
the two are not in contradiction, and only the latter enters the free energy.
\end{remark}

\begin{remark}
\label{rem:trace}
Under the isomorphism $e_i\leftrightarrow\sigma_i$, $I\leftrightarrow i$,
Eq.~\eqref{eq:diag} is the geometric-algebra form of the trace orthogonality
$\mathrm{Tr}(\Gamma_a\Gamma_b)=\pm2\delta_{ab}$ of the Pauli basis~\cite{Peskin1995}: for
real-coefficient multivectors,
\begin{equation}
\tfrac12\mathrm{Tr}(\hat X\hat Y)=\langle XY\rangle_0+i\langle XY\rangle_3^{(\mathrm{coeff})},
\end{equation}
where the superscript ``coeff'' means the real factor in front of $I$ in the
grade-3 component (rather than the blade $I$ itself, so as to match the
complex imaginary unit). Hence
$\langle XY\rangle_0=\mathrm{Re}[\tfrac12\mathrm{Tr}(\hat X\hat Y)]$.
The hatted bold symbols $\hat{\mathbf d}$ and
$\hat n$ of Secs.~\ref{sec:chern} and~\ref{sec:qgt} are unit vectors, as
stated where they are defined.
\end{remark}

\subsection{Central complex structure and Hodge duality}

\begin{lemma}
The pseudoscalar $I$ commutes with all multivectors: $IM=MI$, and $I^2=-1$.
\end{lemma}
\begin{proof}
It suffices to verify $Ie_i=e_iI$. Take $i=1$:
$Ie_1=e_1e_2e_3e_1=-e_1e_2e_1e_3=e_1^2e_2e_3=e_{23}$, while
$e_1I=e_1e_1e_2e_3=e_{23}$. The other basis vectors follow identically. For
the second claim, $I^2=e_1e_2e_3e_1e_2e_3=(-1)^3e_1^2e_2^2e_3^2=-1$, the sign
coming from reversing the order of the last three factors.
\end{proof}

\begin{theorem}[Central complex structure]
The center of $Cl_{3,0}$ is $\{1,I\}$, and the map $a+bI\mapsto a+bi$ gives
an algebraic isomorphism with $\mathbb C$. Phase rotations are implemented by
central rotors $e^{I\varphi}=\cos\varphi+I\sin\varphi$.
\end{theorem}
\begin{proof}
By the lemma, $\{1,I\}$ lies in the center; by $I^2=-1$,
$(a+bI)(c+dI)=(ac-bd)+(ad+bc)I$, which reproduces complex multiplication.
Since $Cl_{3,0}\cong M_2(\mathbb C)$ as a real algebra, its center is
two-dimensional, so $\{1,I\}$ is the whole center. The rotor formula follows
from the Taylor series of the exponential using $I^2=-1$.
\end{proof}

\begin{remark}[The center as the algebraic carrier of gauge symmetry]
The center $\{1,I\}\cong\mathbb C$ provides a natural algebraic carrier for
the U(1) gauge symmetry of pairing theory. Because $I$ is central, the rotor
$e^{I\varphi}$ commutes with all multivectors, and the pairing Hamiltonian
is invariant under the global phase transformation $\psi\to
e^{I\varphi/2}\psi$ of the fermion field. Gauge symmetry is realized by the
central complex structure rather than being added to the algebra from
outside. The phase rotations of the channel theorem
(Theorem~\ref{thm:channels}) are this operation, translating the
phase rotations of the traditional complexification into central rotor
operations.
\end{remark}

\begin{corollary}[Hodge duality]
\label{cor:hodge}
$e_\mu\wedge e_\nu=I\varepsilon_{\mu\nu\rho}e_\rho$.
\end{corollary}
\begin{proof}
This is the content of $Ie_3=e_{12}$, $Ie_1=e_{23}$, and $Ie_2=e_{31}$,
verified directly from the definition $I=e_1e_2e_3$ and the anticommutation
of the basis vectors.
\end{proof}

\section{Electron operators and bosonic vertices}
\label{sec:vertex}

Electrons are spin-$1/2$ fermions. Denoting by $c_{\mathbf k\alpha}$ the
annihilation operator for momentum $\mathbf k$ and spin $\alpha$, the
two-component spinor is $\psi_{\mathbf k}=(c_{\mathbf k\uparrow},
c_{\mathbf k\downarrow})^{\mathsf T}$, and the charge and spin density
operators are the bilinears
\begin{equation}
\rho_{\mathbf q}=\sum_{\mathbf k}\psi_{\mathbf k+\mathbf q}^\dagger
\psi_{\mathbf k},\qquad
\mathbf S_{\mathbf q}=\sum_{\mathbf k}\psi_{\mathbf k+\mathbf q}^\dagger
\tfrac12\boldsymbol\sigma\psi_{\mathbf k}.
\end{equation}
These are the standard second-quantized objects of many-body
theory~\cite{Negele1988,AGD1975}; the spin density $\mathbf S_{\mathbf q}$
is the vertex through which spin fluctuations enter both the particle-hole
and the Cooper channels~\cite{Monthoux1991,Scalapino1986}. In a Majorana
basis, fermion operators do generate the large canonical-anticommutation
Clifford algebra on Fock space; that algebra is distinct from the
spin-space $Cl_{3,0}$ used here. The first-level projection of this paper is therefore not a denial of the
fermionic Clifford structure, but the reduction of spin-space matrix blocks
to $Cl_{3,0}$ multivectors, with residual creation/annihilation dependence
compressed into their coefficients and kept in its written order; raw
operators are never elements of that symmetry-space algebra.

Bosonic collective excitations couple through the vertex $\psi^\dagger M\psi$.
The vertex matrix $M$ acts on spin space, and Hermiticity of the Hamiltonian
requires $M$ to be Hermitian~\cite{Negele1988}. In the geometric-algebra language, reversion
(the tilde operation of Remark~\ref{rem:twometrics}) maps grade-0 and
grade-1 basis elements to themselves and grade-2 and grade-3 elements to
their negatives. The Hermitian sector is therefore grades $0\oplus1$:
under the isomorphism, $1$ and $\boldsymbol\sigma$ are Hermitian, while
$i\boldsymbol\sigma\leftrightarrow e_{\mu\nu}$ and $i\leftrightarrow I$ are
anti-Hermitian. The physical vertices accordingly divide into two classes:
\begin{itemize}
\item \textbf{Charge channel} ($M\propto1$): density-type coupling,
applicable to phonons, plasmons, and all spin-independent fluctuations;
\item \textbf{Spin channel} ($M\propto\boldsymbol\sigma$): spin-type
coupling, applicable to antiferromagnetic and ferromagnetic spin
fluctuations.
\end{itemize}
The grade-2/3 basis elements have not disappeared; they return as the
imaginary-part directions of the pairing field (Sec.~\ref{sec:channels}):
at the pairing-matrix level the imaginary directions are $e_{\mu\nu}$ and
$I$.

The geometric type of a boson in real space is labeled by the irreducible
representation $(l,\text{parity})$ of the rotation group, and its orbital
content enters the vertex through a momentum form factor $f(\mathbf k)$.
Table~\ref{tab:boson} summarizes the classification of common bosonic
collective excitations. The classification rests on the following principle:
the real-space geometric type is fixed by the angular momentum $l$ and the
parity of the boson field, the spin channel is fixed by the grade of the
vertex matrix in $Cl_{3,0}$, and the orbital content is tied to the lattice
symmetry through the form factor $f(\mathbf k)$~\cite{Sigrist1991,Aoki2019}.

\begin{table*}[t]
\caption{\label{tab:boson}Classification of bosonic collective excitations
(particle-hole channel).}
\begin{ruledtabular}
\begin{tabular}{llll}
Boson & Real-space geometry & Vertex grade & Orbit $(l,p)$, vertex\\
\hline
Plasmon & scalar & 0 & $0,+$, $1$\\
Longitudinal phonon & polar vector & 0 (charge) & $1,-$,
$(\mathbf q\cdot\mathbf e)$\\
Magnon & axial vector ($\propto I\mathbf e$) & 1 (spin) & $1,+$,
$\boldsymbol\sigma$\\
Nematic fluctuation & symmetric tensor & 0 & $2,+$,
$(\cos k_x-\cos k_y)$-type\\
Quadrupole fluctuation & symmetric tensor & 0 & $2,+$, $f$-space Stevens
operator\\
Octupole fluctuation & rank-3 tensor & 0 & $3,-$, $f$-space multipole
operator\\
Chiral fluctuation & scalar, $T$-odd & 0 & $0,+$, chiral form factor\\
\end{tabular}
\end{ruledtabular}
\end{table*}

\begin{remark}[Two spaces]
Table~\ref{tab:boson} records labels in two spaces at once. The real-space
geometric type is the transformation property of the boson field under the
three-dimensional rotation group; the spin channel is the position of the
vertex matrix in the spin algebra. For the magnon, which is an axial vector
in real space (Hodge dual to a bivector), the spin vertex is $\boldsymbol\sigma$ (grade 1). The two labels are
distinct but coupled by the vertex. Symmetric tensors with $l\geq2$ (nematic,
quadrupole fluctuations) do not belong to any grade of the exterior algebra
(which contains only antisymmetric tensors); their algebraic carrier is an
orbital-space matrix or a momentum form factor, and their spin channel is
always charge type. The grade classification covers all bosonic
excitations: antisymmetric types correspond to grades directly, and
symmetric types enter through orbital form factors.
The real-space geometry column is the rotation--parity representation
$J^{P}$ of the collective field (its multipole order), written in the
geometric language of the multipole literature: polar vector $=1^{-}$,
axial vector $=1^{+}$, symmetric traceless tensor $=2^{+}$, rank-3 tensor
$=3^{-}$, and pseudoscalar $=0^{-}$. For elementary relativistic bosons
this reduces to the integer-spin classification; for a collective mode in
a crystal the sharp label is the representation of the field rather than
the spin of a quasiparticle---the longitudinal-phonon displacement field
is a polar vector ($1^{-}$), while the phonon quantum itself does not
carry spin-1 angular momentum. Each row names a measurable field:
the charge density $\rho(\mathbf x)$ (plasmon), the displacement
$\mathbf u(\mathbf x)$ (longitudinal phonon), the spin density
$\mathbf S(\mathbf x)$ (magnon), the symmetric traceless quadrupole tensor
$Q_{ij}$ (nematic and quadrupole fluctuations), the octupole moment, and
the scalar spin chirality $\mathbf S_1\cdot(\mathbf S_2\times\mathbf S_3)$
(chiral fluctuation; built from axial vectors it is inversion-even but
time-reversal-odd, unlike the positional scalar chirality
$\mathbf r_1\cdot(\mathbf r_2\times\mathbf r_3)$, which is $0^{-}$). The
electron coupling proceeds through the standard
combinations listed in the vertex column, such as $\nabla\cdot\mathbf u$
for the longitudinal phonon and the dynamical spin susceptibility
$\chi(\mathbf q,\omega)$ for the spin channel. The multipole
representation is an input to the framework: the only algebraic output
concerning the boson is the vertex grade.
\end{remark}

\begin{remark}[Vertex grade versus kernel grade]
The grades in Table~\ref{tab:boson} refer to particle-\emph{hole} vertices.
The quantity that enters the pairing Hamiltonian is the Fierz-rearranged
particle-\emph{particle} (Cooper-channel) effective kernel, and the bridge
between the two languages is the fixed blade $e_{31}\leftrightarrow i\sigma_2$
(Lemma~\ref{lem:anti} and Definition~\ref{def:fierz}). The vertex
classification itself is correct; the channel-attribution theorems of this
paper concern the effective pairing kernel.
\end{remark}

\section{Pairing operators and the complex-channel structure}
\label{sec:channels}

\subsection{Fermionic statistics and Fierz antisymmetrization}

\begin{lemma}[Antisymmetry identity]
\label{lem:anti}
For a $2\times2$ matrix $V$, the pair bilinear
$B^\dagger(V)\equiv\sum_{\mathbf k}\psi_{\mathbf k}^\dagger V\psi_{-\mathbf k}^{\dagger\mathsf T}
=\sum_{\mathbf k}\sum_{\alpha\beta}c^\dagger_{\mathbf k\alpha}V_{\alpha\beta}c^\dagger_{-\mathbf k\beta}$
satisfies
\begin{equation}
B^\dagger(V)=-\,B^\dagger(V^{\mathsf T}),
\end{equation}
that is, the symmetric-matrix part vanishes identically by fermionic
statistics.
\end{lemma}
\begin{proof}
The proof is a relabeling argument given in Supplemental Material Sec.~S5.
Two consequences are worth stating immediately. The surviving part of
$B^\dagger(V)$ is $\tfrac12(V-V^{\mathsf T})\propto i\sigma_2$, so
$B^\dagger(\mathbf 1)\equiv0$ while $B^\dagger(i\sigma_2)$ is the
spin-singlet pair operator~\cite{Sigrist1991}. A constant (momentum-independent) kernel
therefore carries only the singlet at the operator level, and triplet pairing
lives in the momentum structure of the kernel.
\end{proof}

\begin{remark}[$i\sigma_2$ versus time reversal]
Here $i\sigma_2$ is the antisymmetric pairing matrix selected by fermionic
statistics. Time reversal is instead the antiunitary operation
$i\sigma_2K$ (complex conjugation $K$ included), so the pairing vertex matrix
and the symmetry operation should not be conflated.
\end{remark}

\begin{definition}[Fierz map]
\label{def:fierz}
The effective pairing vertex is obtained by right multiplication with the
fixed blade:
\begin{equation}
\Phi:\;G\longmapsto V\equiv Ge_{31},\qquad e_{31}=e_3\wedge e_1
\leftrightarrow i\sigma_2 .
\end{equation}
Since $e_{31}^2=-1$, $\Phi^2=-\mathrm{id}$, so $\Phi$ is a linear bijection
of the eight-dimensional real space; the completeness of the
eight-dimensional parameterization is unaffected. Because
$e_{31}\leftrightarrow i\sigma_2$, $\Phi$ also implements a quarter-turn
(a $\pi/2$ phase rotation) between the kernel and pairing-matrix
conventions: a kernel-level real direction can therefore map to the
matrix-level imaginary direction, so the real/imaginary labels at the two
levels should be read in their own conventions. The blade-by-blade action
is tabulated in Supplemental Material Sec.~S1.
\end{definition}

\subsection{Complex-channel theorem}

\begin{theorem}[Complex-channel theorem, kernel level]
\label{thm:channels}
The effective pairing kernel $G\in Cl_{3,0}$ splits into two eigenspaces.
(i) The \textbf{spin-singlet sector} (even-parity at the gap level) consists of grades
$0\oplus3$, namely the kernel singlet plane
$\mathrm{span}\{1, I\}$; its Fierz image
$\Phi(\mathrm{span}\{1, I\})=\mathrm{span}\{e_2,e_{31}\}$ is the space of
antisymmetric matrices.
(ii) The \textbf{spin-triplet sector} (odd-parity at the gap level) consists of grades
$1\oplus2$ (six real dimensions), whose image is the space of symmetric
matrices.
(iii) The central rotor $e^{I\varphi}$ implements, by left multiplication,
the same U(1) phase rotation within each sector.
\end{theorem}
\begin{proof}
The proof is given in Supplemental Material Sec.~S5. It uses only the
survival condition of Lemma~\ref{lem:anti} and the blade table of
Definition~\ref{def:fierz}. In outline: survival requires $V=Ge_{31}$ to be
antisymmetric; the antisymmetric $2\times2$ complex matrices form the single
complex line $\mathbb C(i\sigma_2)$, whose real span is
$\mathrm{span}\{e_2,e_{31}\}$; and from the blade table,
$V\in\mathrm{span}\{e_2,e_{31}\}$ is equivalent to
$G=-Ve_{31}\in\mathrm{span}\{1,I\}$ (using $e_2e_{31}=I$ and
$e_{31}^2=-1$), which gives (i). The six symmetric directions
$\{1,e_1,e_3,e_{12},e_{23},I\}$ have pre-image $\{e_i,e_{ij}\}$ by the same
method, giving (ii). For (iii), left multiplication by the rotor rotates the pair $(g_0,g_3)$
within $\mathrm{span}\{1,I\}$, and the Hodge duality $Ie_i=e_{jk}$ keeps the
triplet sector closed; the adjoint action is trivial because $I$ is central,
so it is the left action, not the adjoint, that implements the phase
rotation.
\end{proof}

\begin{corollary}[Matrix-level attribution]
\label{cor:matrix-level}
The mean-field gap matrix, written in the standard notation as
$\Delta_s'\sigma_2 + \Delta_s''(i\sigma_2)$~\cite{Sigrist1991},
decomposes in the blade basis as
\begin{equation}
\hat\Delta_s=\Delta_s'\,e_2+\Delta_s''\,e_{31}\in\mathrm{span}\{e_2,e_{31}\}\subset\mathrm{grade}\,1\oplus2,
\end{equation}
and the triplet space is spanned by the symmetric matrices
$\{1,e_1,e_3,e_{12},e_{23},I\}$. The kernel-level and matrix-level
dictionaries are related by the Fierz map $\Phi$:
$\Phi(\mathrm{span}\{1, I\})=\mathrm{span}\{e_2, e_{31}\}$.
\end{corollary}

\begin{corollary}[Singlet-plane signature]
\label{cor:singletsign}
Both $\mathrm{span}\{1, I\}$ and $\mathrm{span}\{e_2, e_{31}\}$ carry signature $(+,-)$:
$\langle1^2\rangle_0=+1$, $\langle I^2\rangle_0=-1$, $\langle e_2^2\rangle_0=+1$,
and $\langle e_{31}^2\rangle_0=-1$. The trace-form discussion of
Sec.~\ref{sec:signature} holds at both levels.
\end{corollary}

\begin{remark}[Hermitian and anti-Hermitian split]
$\hat\Delta_s=\Delta_s'\sigma_2+\Delta_s''(i\sigma_2)$: the $e_2$ component
($\sigma_2$) is Hermitian and the $e_{31}$ component ($i\sigma_2$) is
anti-Hermitian, so the real/imaginary split of the order parameter coincides
with its Hermitian/anti-Hermitian split; at kernel level the corresponding
statement is that $\{1\}$ is Hermitian and $\{I\}$ is anti-Hermitian. A
non-Hermitian $\hat\Delta$ is therefore the standard situation for a complex
order parameter; overall Hermiticity of the Hamiltonian is enforced jointly
by the two operator legs and does not require the kernel itself to be
Hermitian.
\end{remark}

\begin{remark}[The physics of the relative phase]
The overall phase is a gauge freedom (the Goldstone direction), while the
relative phases between different complex fields carry real
physics~\cite{Sigrist1991}: in multicomponent order parameters a relative
phase $\pm\pi/2$ breaks time-reversal symmetry. A nonzero Chern number
requires, in addition, a winding phase texture, as in $p+ip$ or $d+id$;
an $s+id$ state can break time reversal and produce a Kerr signal while
having $C=0$ when its phase texture does not wind
(Theorem~\ref{thm:trsb}). The grade decomposition stores the real and
imaginary parts of each complex field separately, and is a natural
parameterization record of the phase structure of multicomponent order
parameters.
\end{remark}

\subsection{Grade decomposition of pairing structures}

The order parameter is the expectation value of the two-fermion operator,
$\Delta_{\alpha\beta}(\mathbf k)=\langle c_{-\mathbf k\alpha}c_{\mathbf k\beta}\rangle$.
The spin space of the two electrons decomposes under the Clebsch--Gordan
series of SU(2) as $V_{1/2}\otimes V_{1/2}\simeq V_0\oplus V_1$
($\mathbf 2\otimes\mathbf 2=\mathbf 1\oplus\mathbf 3$ in dimensions). The grade attribution of pairing
structures at the kernel and matrix levels is collected in
Table~\ref{tab:pairing-grade}.

\begin{table*}[t]
\caption{\label{tab:pairing-grade}Grade attribution of pairing structures at
the kernel and matrix levels. In the triplet rows the convention is
$\hat\Delta=i(\mathbf d\cdot\boldsymbol\sigma)\sigma_2$ with complex
$\mathbf d$; each real or imaginary component of $\mathbf d$ then carries
the blade listed by direct expansion; the Fierz-map blade table is in
Supplemental Material Sec.~S1.}
\begin{ruledtabular}
{\small\setlength{\tabcolsep}{4pt}
\begin{tabular}{llll}
Level & Channel & Grade & Component\\
\hline
Kernel $G$ & singlet $\mathrm{span}\{1, I\}$ & 0 ($1$) & real part $g_0$\\
Kernel $G$ & singlet $\mathrm{span}\{1, I\}$ & 3 ($I$) & imaginary part $g_3$\\
Kernel $G$ & triplet & 1 ($e_i$) & real part $\mathbf g_1$\\
Kernel $G$ & triplet & 2 ($e_{ij}$) & imaginary part $\mathbf g_2$\\
\hline
Matrix $\hat\Delta$ & singlet $\mathrm{span}\{e_2, e_{31}\}$ & 1 ($e_2$) & real part $\Delta_s'$\\
Matrix $\hat\Delta$ & singlet $\mathrm{span}\{e_2, e_{31}\}$ & 2 ($e_{31}$) & imag.\ part $\Delta_s''$\\
Matrix $\hat\Delta$ & triplet & 1, 3 ($e_1,e_3,I$) & real $\mathbf d$: $d_z\!\to\!e_1$, $d_x\!\to\!-e_3$, $d_y\!\to\!I$\\
Matrix $\hat\Delta$ & triplet & 0, 2 ($1,e_{12},e_{23}$) & imaginary $\mathbf d$: $i d_y\!\to\!-1$, $i d_z\!\to\!e_{23}$, $i d_x\!\to\!-e_{12}$\\
\end{tabular}
}
\end{ruledtabular}
\end{table*}

\section{Scalar projection and the effective pairing interaction}
\label{sec:scalar}

\subsection{Construction}

Every superconducting Hamiltonian involves the ``multiplication together'' of
three kinds of objects. Beneath this single word ``multiply'' hide three
completely different multiplications, summarized in
Table~\ref{tab:multiply}. At the fermion level, the many-body wavefunction
is a Slater determinant, whose algebraic root is the antisymmetry of the
wedge product: $\Psi(\mathbf x_1,\mathbf x_2)=\psi_a(\mathbf x_1)\psi_b(
\mathbf x_2)-\psi_b(\mathbf x_1)\psi_a(\mathbf x_2)\cong\psi_a\wedge\psi_b$,
and the wedge product of a state with itself vanishes identically,
$\psi\wedge\psi=0$. This is the geometric-algebra expression of the Pauli
exclusion principle: it is not pasted on as an external axiom but is built
into the structure of the algebraic multiplication. This statement concerns
the many-body wavefunction; at the Hamiltonian level of this paper the same
statistics enters through the antisymmetric-vertex selection and the
ordering of the Fock coefficients (Lemma~\ref{lem:anti}), since raw Fock
operators are not elements of $Cl_{3,0}$. At the boson level, the
many-body wavefunction is a symmetric tensor product, whose algebraic root
is the commutation structure of the Weyl algebra:
$\Psi=\psi_a\odot\psi_b=\psi_a\psi_b+\psi_b\psi_a$, and the same state can
be superposed arbitrarily, which is the mathematical origin of the statement
that one mode may contain $n$ photons or phonons. At the stitching level,
the electron part and the phonon part are multiplied together with the
tensor product, $\Psi_{\mathrm{total}}=\Psi_{\mathrm{el}}\otimes\Psi_{\mathrm{ph}}$;
strictly speaking, the two do not even live in the same algebra, and the
``multiplication'' is really ``each acts on its own space.'' This federated
structure is compatible with the fermion--boson statistics structure at the
level of many-body wavefunctions: the three multiplications are not merged into one, and the
scalar-projection method respects the federation---it unifies only the
spin-space content of the electronic sector, with the bosonic propagator
and the Fock content entering as ordered coefficients~\cite{Doran2003,Negele1988}.
\begin{table*}[t]
\caption{\label{tab:multiply}The multiplication federation of
superconductivity theory.}
\begin{ruledtabular}
\begin{tabular}{llll}
Object & Multiplication & Algebra & Key property\\
\hline
Fermions & wedge product $\wedge$ & exterior algebra (within Clifford) & antisymmetric:
$\psi\wedge\psi=0$ (Pauli)\\
Bosons & symmetric product & Weyl algebra & commuting: arbitrary powers\\
Stitching & tensor product $\otimes$ & graded stitching & each part keeps
its own product\\
\end{tabular}
\end{ruledtabular}
\end{table*}

The scalar-projection method respects this federation: the electron sector
uses the Clifford algebra (wedge product),
the boson sector uses the Weyl algebra (symmetric product), and the two are
stitched by the tensor product.

\subsection{Provenance of the compact form}
\label{sec:provenance}

Equation~\eqref{eq:hpair} is the form left behind by the standard routes
to a pairing interaction, and it is helpful to see one such route before
using it. The same three bookkeeping steps recur in the standard
derivations; they rearrange only where things sit and strip nothing from
the physics. The logical direction of this paper, however, is the reverse:
Eq.~\eqref{eq:hpair} with an arbitrary kernel is the starting assumption,
and the steps below record its provenance in the cases where a microscopic
derivation exists.

\emph{Step 1 (normal ordering).} A microscopic mechanism acting on the
itinerant band---bare Coulomb repulsion, exchange of a boson (phonon,
photon, spin fluctuation from integrated-out local moments), or the
two-particle vertex of a correlated solver---is conventionally organized
as a normal-ordered two-body interaction
\begin{equation*}
H_{\mathrm{int}}=\frac12\sum_{\mathbf k\mathbf k'\mathbf q}
\psi^\dagger_{\mathbf k+\mathbf q}\,\psi^\dagger_{\mathbf k'-\mathbf q}\;
V(\mathbf q;\mathbf k,\mathbf k')\;
\psi_{\mathbf k'}\,\psi_{\mathbf k},
\end{equation*}
with $V$ a $2\times2$ matrix in each spin index. The model-dependent
physics lives in $V$; the two-body form itself is bookkeeping.

\emph{Step 2 (pairing-channel projection).} In a superconducting state the
anomalous expectations $\langle\psi_{\mathbf k}\psi_{-\mathbf k}\rangle$
are nonzero, and the standard mean-field decoupling of $H_{\mathrm{int}}$
in that channel---equivalently, retaining the anomalous part of the
interaction under Nambu normal ordering---leaves
\begin{equation*}
H_{\mathrm{pair}}=\sum_{\mathbf k,\mathbf k'}
\psi^\dagger_{\mathbf k}\,\psi^\dagger_{-\mathbf k}\;
V^{\mathrm{pp}}(\mathbf k,\mathbf k')\;
\psi_{-\mathbf k'}\,\psi_{\mathbf k'},
\end{equation*}
in which each particle-hole vertex has been crossed into the
particle-particle channel. The crossing is a Fierz rearrangement; its sign
structure is analyzed in Theorem~\ref{thm:casimir}, and the antisymmetry of
the two-fermion vertex that any such decoupling must respect is
Lemma~\ref{lem:anti} at the Hamiltonian level.

\emph{Step 3 (spin-blade decomposition with coefficient compression).}
Each spin vertex is now a $2\times2$ matrix block; expanding it in the
blade basis of $Cl_{3,0}$ and compressing everything that is not
spin---form factors, propagator denominators, band and momentum
dependence---into ordered coefficients turns $V^{\mathrm{pp}}$ into a
$Cl_{3,0}$-valued kernel. For a boson-mediated interaction the move is
explicit: the second-order kernel is the geometric product $M\,D\,M$ of the
spin vertex $M$ with the boson propagator $D$, and
\begin{equation*}
\begin{split}
\sum_{a,b}\sigma_a D_{ab}\sigma_b
=(\mathrm{Tr}\,D)\,\mathbf{1}+i\,\varepsilon_{abc}D_{ab}\,\sigma_c \\
\ \longleftrightarrow\ (\mathrm{Tr}\,D)\,1+\varepsilon_{abc}D_{ab}\,Ie_c,
\end{split}
\end{equation*}
since $\sigma_a\sigma_b=\delta_{ab}\mathbf{1}+i\varepsilon_{abc}\sigma_c$.
The isotropic part of the propagator therefore feeds grade 0 and its
antisymmetric part feeds grade 2 (the $Ie_c$ are bivectors), while the
crossing of Step 2 places the result into the Fierz-antisymmetric kernel
subspace. For an isotropic propagator $D_{ab}=D_0\delta_{ab}$ one has
$\mathrm{Tr}\,D=3D_0$; this trace factor is the numerical origin of the
factor $3$ in the $-3\chi(\mathbf k-\mathbf k')$
singlet entry of Table~\ref{tab:mechanisms}.

After the three steps, each microscopic mechanism appears as a direction
in the eight-dimensional real kernel space together with a momentum
texture. The steps are a provenance, not a premise: where they fail---glue
that cannot be linearly separated into one boson per term, or interactions
without perturbative control---no microscopic derivation is available, yet
the spin content of the effective pairing channel remains well defined and
still expands in the blade basis, with everything that is not spin sitting
in the coefficients. Equation~\eqref{eq:hpair} is therefore adopted as the
working assumption, more general than any of its derivations. The compact
form of the effective pairing interaction is
\begin{equation}
H_{\mathrm{pair}}=\sum_{\mathbf k,\mathbf k'}
\big\langle\psi_{\mathbf k}^\dagger\,\psi_{-\mathbf k}^\dagger\;
G(\mathbf k,\mathbf k')\;\psi_{-\mathbf k'}\,\psi_{\mathbf k'}\big\rangle_{0,\mathrm{spin}}
\label{eq:hpair}
\end{equation}
The Hermitian conjugate of the double sum is the same expression with
$\mathbf k\leftrightarrow\mathbf k'$ relabeled, so overall Hermiticity is
the kernel condition $G(\mathbf k,\mathbf k')=\tilde G(\mathbf k',\mathbf k)$,
i.e.\ $\hat G(\mathbf k,\mathbf k')=\hat G(\mathbf k',\mathbf k)^\dagger$;
no separate $\mathrm{h.c.}$ term is needed.
The spinor factors are understood in their projected form:
the spin content of $\psi^\dagger\psi^{\dagger\mathsf T}$ and
$\psi\psi^{\mathsf T}$ is decomposed in the blade basis of $Cl_{3,0}$,
while the Fock-operator content is compressed into the coefficients and
kept in its written order, so $\langle\cdot\rangle_{0,\mathrm{spin}}$
contracts the spin-blade structure of the three factors. This is the same
move as the Schr\"odinger representation, in which the momentum operator
projects onto a partial derivative and the operator content is compressed
into the coefficients. The Fock-operator coefficients commute with the spin
blades and are ordered among themselves. The notation in Eq.~\eqref{eq:hpair} implements
the two-level projection rule. The fermion operators are Fock-space carriers and are not,
as raw operators, elements of the spin-space algebra $Cl_{3,0}$; the objects
entering $Cl_{3,0}$ are their projected $2\times2$ spin-matrix blocks
together with $G$. The geometric product acts among those symmetry-space
multivectors, and the subscripted projection $\langle\cdot\rangle_{0,\mathrm{spin}}$
acts linearly on their spin blades only, leaving Fock-carrier factors
untouched; before mean-field expectation the resulting grade-0 scalar is
still an operator on Fock space. For a constant kernel $G=g_0\cdot1$ the surviving vertex is
$\Phi(G)=g_0e_{31}\leftrightarrow g_0\,i\sigma_2$
(Definition~\ref{def:fierz}), and this
reduces to the classic BCS pairing form
$H_{\mathrm{pair}}=-g_0\sum_{\mathbf k,\mathbf k'}s_{\mathbf k}^\dagger
s_{\mathbf k'}$ with
$s_{\mathbf k}=c_{\mathbf k\uparrow}c_{-\mathbf k\downarrow}
-c_{\mathbf k\downarrow}c_{-\mathbf k\uparrow}$: the $\mathbf k$-summed
symmetric-spin part of the creation bilinear vanishes by
Lemma~\ref{lem:anti}, and
$\langle e_{31}^2\rangle_0=-1$ converts the kernel sign into the $V<0$
attraction convention.

The kernel $G(\mathbf k,\mathbf k')\in Cl_{3,0}$ is an arbitrary real-coefficient
multivector in the three-dimensional symmetry-space algebra, expanded
completely by the eight-dimensional real basis. The scalar-projection method
presupposes no particular boson field, no validity of second-order
perturbation theory, and no factorization assumption: phonons, spin
fluctuations, spin-orbit coupling, external fields, and all other effects
enter the projected eight real components of $G$ as coefficients. Bosons,
orbitals, and momenta participate through scalar form factors stitched by
the tensor product.

The kernel may be obtained in two ways: (1) effective-theory
parameterization, writing the expansion coefficients of $Cl_{3,0}$ directly
from symmetry and experimental constraints; (2) DMFT output~\cite{Georges1996}, mapping the
two-particle vertex of an impurity solver onto the $Cl_{3,0}$ basis by a
Fierz/projection map. Neither route involves the stripping of any bosonic
glue.

\subsection{Mixed states as native coordinates, not model patchwork}

In the traditional BCS/Eliashberg framework, a ``mixed state'' is a
theoretical burden. To describe the coexistence of, say, phonon glue plus
spin-fluctuation glue plus spin-orbit-coupling modulation, one must write a
microscopic Hamiltonian containing all the degrees of freedom, perform
perturbation theory or mean-field theory, and verify that the different
mechanisms do not cancel each other. In the strong-correlation region this
is often infeasible, because the glue cannot be stripped and different
microscopic mechanisms cannot be added cleanly. The scalar-projection method
instead takes $G\in Cl_{3,0}$ directly as an \emph{a priori} arbitrary
eight-dimensional real parameter. Different microscopic mechanisms are simply
different directions in this space, together with momentum structures
(Table~\ref{tab:mechanisms}). A mixed state does not need to be
constructed; it only needs to be parameterized. This is the
methodological conversion of this paper: from asking ``what is the
microscopic glue'' to parameterizing ``the most general structure of the
low-energy effective kernel.''

\begin{table*}[t]
\caption{\label{tab:mechanisms}Microscopic mechanisms at the kernel level.}
\begin{ruledtabular}
\begin{tabular}{lll}
Mechanism & Vertex (p-h) grade & Kernel sector / momentum structure\\
\hline
Phonon (conventional $s$) & 0 & $\mathrm{span}\{1\}$ (grade 0), isotropic
$g_0$\\
Antiferromagnetic fluctuation & 1 & singlet sector after Fierz,
$g_0\propto-3\chi(\mathbf k-\mathbf k')$\\
Nematic/orbital fluctuation & 0 $+$ form factor & singlet sector,
sign-changing $f(\mathbf k)f(\mathbf k')$\\
Ferromagnetic fluctuation & 1 & triplet sector, $\chi$ peaked
$\mathbf q\to0$\\
SOC $+$ external fields & --- & cross-sector mixing\\
Strong correlation & --- & full $Cl_{3,0}$\\
\end{tabular}
\end{ruledtabular}
\end{table*}

The grade sectors carry only spin and parity information. What
distinguishes one mechanism from another is the momentum texture and the
signs of the coefficients $g(\mathbf k,\mathbf k')$, and this is where the
Casimir analysis of Sec.~\ref{sec:signature} and the selection rules of
Sec.~\ref{sec:block} do their work.

\subsection{Complete parameterization and the route to BdG}
\label{sec:flow}

The effective kernel is expanded completely as
\begin{equation}
G(\mathbf k,\mathbf k')=g_0\cdot1+\mathbf g_1\cdot\mathbf e+\mathbf g_2
\cdot(\mathbf e\wedge\mathbf e)+g_3\cdot I .
\end{equation}
Since $Cl_{3,0}\cong M_2(\mathbb C)$ (Remark~\ref{rem:trace}), the
eight-real-dimensional basis provides a set of real coordinates for the
space of $2\times2$ complex pairing matrices. The traditional
complexification $\Delta=\Delta'+i\Delta''$ merges the grade-2/3 (imaginary)
directions with the grade-0/1 (real) directions through the global U(1)
phase; this is physically self-consistent because the overall phase is a
gauge freedom. This paper unfolds the complex structure $\{1,I\}$
explicitly, displaying the four grades on equal footing in real space. The
complex structure is separated, so the imaginary
directions handled implicitly by the traditional $i$ become independently
trackable coordinates. The two descriptions are
physically equivalent and differ only in the choice of coordinates; more
coordinates do not mean more physical freedoms.

The route from the scalar-projection method to BdG proceeds in four steps.

\paragraph{Step 1: Fermionic filtering.}
By Lemma~\ref{lem:anti}, the constant part of Eq.~\eqref{eq:hpair} carries
only the antisymmetric ($i\sigma_2$) vertex; for general momentum
structures, both channels emerge from the eigen-decomposition of the
effective kernel (Theorem~\ref{thm:channels}).

\paragraph{Step 2: Mean field and the generalized gap equation.}
The anomalous expectation
$F_{\alpha\beta}(\mathbf k)=\langle c_{-\mathbf k\alpha}c_{\mathbf k\beta}\rangle$
obeys the fermionic antisymmetry constraint
\begin{equation}
F(\mathbf k)=-F^{\mathsf T}(-\mathbf k),
\label{eq:fconstraint}
\end{equation}
and the gap equation
\begin{equation}
\hat\Delta(\mathbf k)=\sum_{\mathbf k'}G(\mathbf k,\mathbf k')\,F(\mathbf k').
\label{eq:gapeq}
\end{equation}
Here $\hat\Delta$ carries the index order of $F$ as written in the pair field
of Eq.~\eqref{eq:hpair}; transposing both sides gives the equivalent form
$\hat\Delta^{\mathsf T}(\mathbf k)=\sum_{\mathbf k'}F^{\mathsf T}
(\mathbf k')\,G^{\mathsf T}(\mathbf k,\mathbf k')$, and the Eliashberg
equation below uses the no-transpose convention.
The kernel $G$ is simultaneously the multivector whose grades carry the channel
labels of Theorem~\ref{thm:channels} and the matrix acting on $F$ in
Eq.~\eqref{eq:gapeq}; the Fierz map $\Phi$ (Definition~\ref{def:fierz})
translates between the kernel-level dictionary and the matrix-level one of
Corollary~\ref{cor:matrix-level}, while Lemma~\ref{lem:anti} enforces
operator-level survival. As a check, $G=g_0\cdot1$ with
$F=f\,i\sigma_2$ gives $\hat\Delta\propto i\sigma_2$, a singlet, while a
constant kernel cannot couple to odd-parity channels at all (their angular
projection vanishes), so the $s$-wave is selected naturally. We record the
scope of Eq.~\eqref{eq:gapeq}: the kernel acts on $F$ by left (geometric)
multiplication, i.e.\ the vertex class
$V_{\alpha\beta,\gamma\delta}=G_{\alpha\gamma}\delta_{\beta\delta}$;
vertices with right-multiplication or channel-projector structure (for
example the on-site singlet projector) lie outside this class and are not
claimed here.

The full matrix-valued gap equation follows the standard route. A
Hubbard--Stratonovich transformation~\cite{Negele1988}, exact for positive
coupling, converts the four-fermion term into a fermion-boson quadratic
form; Gaussian integration over the Nambu spinors and a saddle point then
give the Matsubara-space equation
\begin{equation}
\hat\Delta(i\omega_n)=T\sum_m\mathcal{G}_{\mathrm{GA}}(i\omega_n-i\omega_m)\,
\frac{\hat\Delta(i\omega_m)}{\sqrt{\omega_m^2+\xi^2+|\hat\Delta(i\omega_m)|^2}},
\end{equation}
and analytic continuation~\cite{AGD1975} yields, in a schematic
real-axis form (the $i0^+$ structure suppressed and the retarded
prescription understood),\begin{small}
\begin{equation}
\begin{split}
\hat\Delta(\mathbf k,\omega)=\!\int\! d\mathbf k'd\omega'\,
\mathcal{G}_{\mathrm{GA}}(\mathbf k-\mathbf k',\omega-\omega')\,\\
\times\;\frac{\hat\Delta(\mathbf k',\omega')}{\sqrt{\omega'^2-\xi_{\mathbf k'}^2-
|\hat\Delta(\mathbf k',\omega')|^2}}\tanh\frac{\beta\omega'}{2}.
\label{eq:eliashberg}
\end{split}
\end{equation}
\end{small}
Here $\mathcal{G}_{\mathrm{GA}}=\sum_Ag_Ae_A$ is the grade-structured,
matrix-valued glue kernel. No Migdal theorem~\cite{Migdal1958} is required
in this derivation, which matters because for spin, nematic, and multipolar
glues the condition $\omega_{\mathrm{glue}}\ll E_F$ fails. Two things here go
beyond the traditional equation: the kernel is an eight-component matrix
object rather than a scalar function, and the output channel is read off
from Theorem~\ref{thm:channels} rather than assumed. The full momentum
dependence is also retained, where the traditional Fermi-surface average
would wash out the angular structure required by $d$-, $p$-, and $s_\pm$-wave
states.

\paragraph{Fixed-point formulation (remark level).}
Write $\mathcal E=\mathcal K\circ\mathcal F$, with $\mathcal K$ the
kernel-linear integral operator and $\mathcal F$ the gap-denominator
nonlinearity. Under standard regularity assumptions (square-integrable,
positive kernel) $\mathcal K$ is Hilbert--Schmidt, $T_c$ is defined by
$\lambda_{\max}(T_c)=1$, and existence of a nonzero fixed point follows from
Krein--Rutman and Leray--Schauder arguments. We state this as orientation
only; a full regularity analysis is beyond the scope of the paper.

\paragraph{Step 3: BdG decomposition.}
With the Nambu spinor
$\Psi_{\mathbf k}=(c_{\mathbf k\uparrow},c_{\mathbf k\downarrow},
c^\dagger_{-\mathbf k\downarrow},c^\dagger_{-\mathbf k\uparrow})^{\mathsf T}$
(an ordering in which a singlet block reads
$h_{\mathrm{BdG}}=\xi\tau_z+\mathrm{Re}\Delta\,\tau_x-\mathrm{Im}
\Delta\,\tau_y$) and the single-particle energy
$\xi_{\mathbf k}=\epsilon_{\mathbf k}-\mu$,
the total BdG Hamiltonian is~\cite{Read2000}
\begin{equation}
\begin{split}
\mathcal H_{\mathrm{BdG}}=\sum_{\mathbf k}\Psi_{\mathbf k}^\dagger
\hat h_{\mathrm{BdG}}(\mathbf k)\Psi_{\mathbf k},\\
\hat h_{\mathrm{BdG}}(\mathbf k)=\begin{pmatrix}
\xi_{\mathbf k}\mathbf 1&\hat\Delta(\mathbf k)\\
\hat\Delta^\dagger(\mathbf k)&-\xi_{\mathbf k}\mathbf 1
\end{pmatrix},
\end{split}
\end{equation}
where $\hat\Delta(\mathbf k)$ decomposes according to
Corollary~\ref{cor:matrix-level}.

\paragraph{Step 4: Quasiparticle spectra.}
For spin-rotation-invariant systems with unitary triplet states
($\mathbf d\times\mathbf d^\ast=0$) the singlet and triplet blocks do not
couple, and the spectra are
\begin{equation}
E_{\mathbf k}^{(S)}=\sqrt{\xi_{\mathbf k}^2+|\Delta_s(\mathbf k)|^2},\qquad
E_{\mathbf k}^{(T)}=\sqrt{\xi_{\mathbf k}^2+|\mathbf d(\mathbf k)|^2}.
\end{equation}
Topological invariants are determined by the winding of $\hat\Delta(\mathbf k)$
in momentum space (Theorem~\ref{thm:chern}).

\subsection{Generation of the pairing glue}
\label{sec:glue}

The preceding sections take the kernel as given; here we give a model-level
illustration of how a projected bosonic channel can feed it. In a specified
fluctuation model, a bosonic fluctuation $\phi_{\mathbf q}$ supplies the glue
\begin{equation}
\Gamma(\mathbf q)=\langle\phi_{\mathbf q}\phi_{-\mathbf q}\rangle,
\end{equation}
and the grade content of $\Gamma$ indicates which projected kernel
components can be activated. For vector bosons,
\begin{equation}
\Gamma(\mathbf q)=\mathbf v_{\mathbf q}\mathbf v_{-\mathbf q}
=\underbrace{\mathbf v_{\mathbf q}\cdot\mathbf v_{-\mathbf q}}_{\text{grade 0}}
+\underbrace{\mathbf v_{\mathbf q}\wedge\mathbf v_{-\mathbf q}}_{\text{grade 2}} .
\end{equation}
The wedge term vanishes if and only if $\mathbf v_{-\mathbf q}\parallel
\mathbf v_{\mathbf q}$. In idealized minimal models, three standard
correspondences follow.

(i) \emph{Longitudinal phonon}: $\mathbf u_{-\mathbf q}=-\mathbf u_{\mathbf q}$
(the sign follows the longitudinal-polarization convention $\mathbf e_{-\mathbf q}=-\mathbf e_{\mathbf q}$), the wedge
vanishes identically, and the glue is pure grade 0; the component that can
be fed is $g_0$, the singlet kernel direction.

(ii) \emph{Antiferromagnetic fluctuations} peaked at
$\mathbf Q=(\pi,\pi)$: $\delta\mathbf S_{-\mathbf q}\approx-\delta\mathbf S_{\mathbf q}$,
the same geometry with a minus sign, the scalar part constitutes the glue,
and the standard model-level outcome is a singlet kernel with momentum
modulation (the $d$-wave route)~\cite{Monthoux1991,Scalapino1986,Chubukov2008}.

(iii) \emph{Ferromagnetic fluctuations} as $\mathbf q\to0$:
$\delta\mathbf S_{-\mathbf q}\approx+\delta\mathbf S_{\mathbf q}$ (even),
so the wedge again vanishes and the glue remains scalar in type; the
triplet feeding of the minimal $^3$He picture~\cite{Vollhardt1990} is
carried instead by the interaction-sign structure of the spin vertex, the
Casimir content of Theorem~\ref{thm:casimir}, not by a wedge component.

For symmetric-tensor bosons (the nematic $\mathcal Q_{\mathbf q}$), the
glue is scalar in type, $\Gamma\propto\mathcal Q_{\mathbf q}\cdot
\mathcal Q_{-\mathbf q}$, and the orbital content rides on the form
factor: the singlet sector is fed with the sign-changing
$f(\mathbf k)f(\mathbf k')$ texture (the $s_\pm$ route of the iron
pnictides)~\cite{Chubukov2008,Kuroki2009,Fang2011}. Symmetric tensors carry no grade, so no grade-2 glue
component arises from this channel. If nematic order breaks the lattice
symmetry that protects the singlet/triplet separation, the shape factor
acquires a spin-orbit admixture and a small triplet component can be
induced; this is a form-factor effect, not a grade-mixing effect. For
pseudoscalar bosons (a $T$-odd
pseudoscalar-type fluctuation $\Omega_{\mathbf q}=\omega_{\mathbf q}I$),
$\Gamma=-\omega_{\mathbf q}\omega_{-\mathbf q}$ is scalar in type with a
sign fixed by $I^2=-1$; the chiral content is carried by the vertex-side
$I$ factor and by the form factor. The response of the singlet kernel plane to a pseudoscalar drive is
treated in Sec.~\ref{sec:fields}.

At the algebraic level, the symmetry of a specified boson field under
$\mathbf q\to-\mathbf q$ fixes the grade content available to the glue and
therefore the kernel sector that can be fed. This is a model-level generation
picture rather than a theorem about materials: the grade arithmetic does not
decide which combination condenses (that is settled by the gap equation), but
once a mechanism activates a set of components, the algebra provides the
coupling rules among them.

\subsection{Cooper-channel decomposition}

At the mean-field level, writing the pairing fields in channel decomposition
gives the standard Cooper-channel form
\begin{equation}
\begin{split}
H_{\mathrm{pair}}=\sum_{\mathbf k,\mathbf k'}\Big[V_{00}(\mathbf k,\mathbf k')
\Delta_0^\dagger(\mathbf k)\Delta_0(\mathbf k')\\
+\sum_{ab}V_{ab}(\mathbf k,\mathbf k')d_a^\dagger(\mathbf k)d_b(\mathbf k')\Big],
\end{split}
\end{equation}
where $V_{00}$ is the singlet-singlet kernel and $V_{ab}$ the
triplet-triplet kernel. In spin-rotation-invariant systems the
singlet-triplet mixing term vanishes; in non-centrosymmetric systems mixing
appears and gives mixed-parity pairing~\cite{Frigeri2004}. Within the
singlet block the internal coordinates are $(g_0,g_3)$ at kernel level and
$(\Delta_s',\Delta_s'')$ at matrix level, related by $\Phi$; the kernel
metric has signature $(+,-)$ (Corollary~\ref{cor:singletsign}), while the physical free-energy metric
$|\Delta|^2$ is positive definite and gauge invariant, and the two are not
in contradiction.

\subsection{Why intermediate objects can carry anti-Hermitian components}

In the traditional BCS/Eliashberg framework, the effective pairing kernel is
habitually restricted to Hermitian matrices. This restriction does not come
from the requirement of overall Hermiticity of the Hamiltonian: overall
Hermiticity, $H_{\mathrm{pair}}^\dagger=H_{\mathrm{pair}}$, can be realized
by the joint Hermitian conjugation of $G$ and the fermion operators, and
does not exclude $G$ itself from carrying anti-Hermitian components (the
Hermitian/anti-Hermitian remark above). The root of the traditional
restriction is its physical origin: the second-order perturbation theory of
the electron-phonon coupling produces a real attractive kernel $V<0$
(that is, $G>0$ in the convention of Eq.~\eqref{eq:gapeq}), which falls
naturally in the Hermitian direction. Later extensions (spin
fluctuations and others) activate further components, but the
complexification $\Delta=\Delta'+i\Delta''$ absorbs the anti-Hermitian
directions implicitly into the global phase and never treats them as
independent parameter directions. The methodological move of the
scalar-projection method is this: presupposing no particular microscopic
mechanism, take $G\in Cl_{3,0}$ as an \emph{a priori} arbitrary
eight-dimensional real parameter. Grades 2 and 3 are then no longer
``appendages of the phase'' but independent directions on equal footing with
grades 0 and 1. The relative phase between the anti-Hermitian and Hermitian
components, such as the $\pm\pi/2$ in $d_{x^2-y^2}+i\,d_{xy}$, becomes a
directly parameterizable object. Traditional theory could write this
structure, but the real and imaginary directions are usually left folded
into the single symbol $i$; the present treatment separates them.

\subsection{Orbital structure of the glue and gap symmetry}

The pairing symmetry is decided by the linearized gap equation
\eqref{eq:gapeq}; in the traditional notation with the pairing potential
$V\equiv-G$ it reads
\begin{equation}
\Delta_{\alpha\beta}(\mathbf k)=-\sum_{\mathbf k'}V_{\alpha\beta,\gamma\delta}
(\mathbf k,\mathbf k')\langle c_{-\mathbf k'\gamma}c_{\mathbf k'\delta}\rangle,
\end{equation}
so the attractive BCS side is $V<0$, equivalently $G>0$ in the convention
of Eq.~\eqref{eq:gapeq}. Throughout this paper positive kernel eigenvalues
($\lambda_n>0$ in Sec.~\ref{sec:spectral}, $V_A>0$ in
Sec.~\ref{sec:rg}) denote attraction. The momentum dependence of the kernel function $V(\mathbf k,
\mathbf k')$ is fixed by the propagator of the glue and the vertex form
factors: the charge vertex (grade 0) gives
$V(\mathbf k,\mathbf k')=g^2D(\mathbf k-\mathbf k')f(\mathbf k)f(\mathbf k')$,
with the orbital content entirely in the form factor (for the nematic glue
$f=\cos k_x-\cos k_y$); the spin vertex (grade 1) gives
$V(\mathbf k,\mathbf k')=g^2\chi(\mathbf k-\mathbf k')$, fixed by the spin
fluctuation spectrum (antiferromagnetic peak at $\mathbf Q=(\pi,\pi)$ or
$(\pi,0)$, ferromagnetic peak at $\mathbf q\to0$)~\cite{Monthoux1991,
Scalapino1986}. In separable idealized limits, an isotropic scalar kernel
gives a nodeless singlet ($s$ wave); a strictly $(\pi,\pi)$-peaked kernel
favors a sign-changing gap, with the $d_{x^2-y^2}$ singlet as the standard
square-lattice realization; sign changes between pockets favor an $s_\pm$
singlet; and a ferromagnetic peak favors attraction in the triplet sector
(Casimir sign $+1$, Theorem~\ref{thm:casimir}). Outside those limits the
actual eigenchannel is decided by the gap equation. The conditional
selection statements are given in Sec.~\ref{sec:block}.

\subsection{Spectral decomposition of the pairing tensor and entanglement}
\label{sec:spectral}

Regarding the effective pairing kernel $G(\mathbf k,\mathbf k')$ as an
integral operator on momentum space,
\begin{equation}
(\hat G\zeta)(\mathbf k)=\sum_{\mathbf k'}G(\mathbf k,\mathbf k')\zeta(\mathbf k'),
\end{equation}
and assuming a real symmetric kernel (time-reversal invariant), one has the
spectral decomposition
\begin{equation}
G(\mathbf k,\mathbf k')=\sum_n\lambda_n\zeta_n(\mathbf k)\otimes\zeta_n(\mathbf k'),\qquad
\lambda_n\in\mathbb R,
\end{equation}
where each eigenvalue $\lambda_n$ and eigenfunction $\zeta_n$ forms a pairing
eigenchannel: $\lambda_n>0$ is an attracting channel (which can
superconduct) and $\lambda_n<0$ a repelling one; the grade decomposition of
$\zeta_n$ indicates directly the spin and orbital character of that channel.
Traditional BCS has only one attracting channel (a single $\lambda_0>0$), while under the
scalar-projection method $G$ may have several nonzero eigenvalues,
corresponding to multicomponent competing orders.

The singular value decomposition of the mean-field pairing matrix
$\hat\Delta(\mathbf k)$ gives
\begin{equation}
\hat\Delta(\mathbf k)=\sum_{n=1}^2s_n(\mathbf k)|u_n(\mathbf k)\rangle\langle v_n(\mathbf k)|,\qquad
s_1\geq s_2\geq0,
\end{equation}
with $s_1^2+s_2^2=\mathrm{Tr}(\hat\Delta^\dagger\hat\Delta)$. The ratio
$r(\mathbf k)=s_2/s_1$ measures the structural anisotropy of the pairing
tensor: $r=1$ for a pure spin singlet ($\hat\Delta\propto i\sigma_2$) or an
isotropic triplet, and $r=0$ for a pure rank-1 triplet (a single $S_z$
component); intermediate values signal mixed states with several grade
components. Defining the normalized spin-channel density matrix
$\hat\rho_{\mathrm{spin}}(\mathbf k)=\hat\Delta\hat\Delta^\dagger/
\mathrm{Tr}[\hat\Delta\hat\Delta^\dagger]$, its structure parameter
\begin{equation}
R(\mathbf k)=\mathrm{Tr}\,\hat\rho_{\mathrm{spin}}^2(\mathbf k)\in[1/2,1]
\end{equation}
satisfies $R=1$ for a rank-1 pairing matrix (a single $S_z$ component, whose
two-particle state is separable, concurrence $=0$) and $R=1/2$ for a
unitary $\hat\Delta$ (a pure singlet or a completely unpolarized triplet,
concurrence $=1$). Two cautions are in order. First, $R$ measures the rank
structure of the pairing matrix rather than a quantum purity: the pure
singlet is a \emph{maximally entangled} two-particle state, not a mixed one;
the singlet/triplet identity itself is decided by the matrix structure of
$\hat\Delta$ (its Fierz decomposition), not by $R$. Second, traditional BCS
usually assumes $\hat\Delta$ to be a pure structure (a single pairing
symmetry), whereas the $G\in Cl_{3,0}$ of the scalar-projection method
naturally allows mixed structures, and the momentum dependence of
$R(\mathbf k)$ reflects channel competition directly.

For the many-body BdG ground state, an entanglement analysis connects to the
language of quantum information. The BCS ground state is a fermionic
Gaussian state, and its entanglement is fixed by the covariance matrix; the
anomalous block of the BdG Hamiltonian is parameterized by
$\hat\Delta(\mathbf k)$, and the anomalous correlator after diagonalization
is fixed by the BdG eigenvectors. The von Neumann entropy of a region $A$ is
determined by the spectrum of the covariance matrix, that is, by the BdG gap
structure~\cite{Eisert2010}. At the algebraic level, grade-2/3 components
only provide explicit real coordinates for complex relative phases of
$\hat\Delta(\mathbf k)$. Any consequence for phase winding, entanglement
spectra, or scaling near a transition is model-level and must be obtained
from the corresponding BdG covariance matrix; it is not a theorem supplied
by the grade decomposition alone.

A note on the stability of mixed states. This paper treats
$G\in Cl_{3,0}$ as an \emph{a priori} arbitrary eight-dimensional real
parameter, which is an expansion at the level of algebraic grammar and
presupposes no microscopic dynamics. Which combinations of grade components
can actually condense in a given material is decided by the gap equation or
by the minimization of the Ginzburg--Landau free energy, and lies outside
the scope of this paper. Once a microscopic mechanism or an experimental
indication shows that a certain set of components is activated, the
algebra provides the coupling and selection rules among them; whether they are activated must be answered by the specific model or
by experiment.

\section{Signature: representation, Casimir, trace form, and flow}
\label{sec:signature}

This section establishes the threefold structure of the pairing algebra:
the channel dichotomy, the interaction signs, and the signature of the
pairing form are unified in the same Lie algebra $\mathfrak{su}(2)$ and its
trace form. The signature structure offers an observation angle:
conventionally empirical rules such as ``singlet attraction, triplet
repulsion'' can be put in correspondence with the Casimir spectrum of
$\mathfrak{su}(2)$. The correspondence is structural; it does not change the
empirical rules themselves but places them in a unified algebraic grammar.

\subsection{Representation: singlet and triplet}

The Hilbert space $\mathbb C^2\otimes\mathbb C^2$ of two
spin-$1/2$ particles decomposes under the Clebsch--Gordan series of SU(2) as
\[
V_{1/2}\otimes V_{1/2}\simeq V_0\oplus V_1,
\]
or equivalently, in dimensions, $\mathbf 2\otimes\mathbf 2
=\mathbf 1\oplus\mathbf 3$: the spin singlet
($S=0$, one-dimensional) and the spin triplet ($S=1$, three-dimensional).
The corresponding pair operators are~\cite{Sigrist1991}:

\begin{equation}
\begin{aligned}
\hat S_{\mathbf k} &= \tfrac{1}{\sqrt2}\big(c_{-\mathbf k\uparrow}c_{\mathbf k\downarrow}
  -c_{-\mathbf k\downarrow}c_{\mathbf k\uparrow}\big),
  &\quad
\hat T_+ &= c_{-\mathbf k\uparrow}c_{\mathbf k\uparrow}, \\
\hat T_0 &= \tfrac{1}{\sqrt2}\big(c_{-\mathbf k\uparrow}c_{\mathbf k\downarrow}
  +c_{-\mathbf k\downarrow}c_{\mathbf k\uparrow}\big),
  &\quad
\hat T_- &= c_{-\mathbf k\downarrow}c_{\mathbf k\downarrow}.
\end{aligned}
\end{equation}

\subsection{Casimir: the sign structure of the interaction}

\begin{theorem}[Casimir origin of Fierz signs, algebraic level]
\label{thm:casimir}
The square of the total spin $\mathbf S=(\boldsymbol\sigma_1+\boldsymbol\sigma_2)/2$
is the Casimir operator of $\mathfrak{su}(2)$, with eigenvalues $S(S+1)$:
$0$ for the singlet and $2$ for the triplet. From
\begin{equation}
\boldsymbol\sigma_1\cdot\boldsymbol\sigma_2=2\mathbf S^2-3
\end{equation}
the exchange operator $\boldsymbol\sigma_1\cdot\boldsymbol\sigma_2$ has
eigenvalues $-3$ (singlet) and $+1$ (triplet). This pair of numbers is
the sign structure (Fierz coefficient) of an exchange-type glue in
the Cooper channel: antiferromagnetic-type exchange gives the lower
eigenvalue in the singlet channel, and ferromagnetic-type exchange favors
the triplet~\cite{Monthoux1991,Scalapino1986}.
\end{theorem}
\begin{proof}
Expanding,
\begin{equation}
\mathbf S^2=\tfrac14(\boldsymbol\sigma_1+\boldsymbol\sigma_2)^2
=\tfrac14(3+3+2\boldsymbol\sigma_1\cdot\boldsymbol\sigma_2)
=\tfrac32+\tfrac12\boldsymbol\sigma_1\cdot\boldsymbol\sigma_2 ,
\end{equation}
so $\boldsymbol\sigma_1\cdot\boldsymbol\sigma_2=2\mathbf S^2-3$.
Substituting the eigenvalues $S(S+1)$ of $\mathbf S^2$ (singlet $0$,
triplet $2$) gives $-3$ and $+1$. An equivalent projector form is
$\boldsymbol\sigma_1\cdot\boldsymbol\sigma_2=2P_{\mathrm{swap}}-1$, with
singlet and triplet projectors $P_s=(1-\boldsymbol\sigma_1\cdot
\boldsymbol\sigma_2)/4$ and $P_t=(3+\boldsymbol\sigma_1\cdot\boldsymbol
\sigma_2)/4$, satisfying $P_s^2=P_s$, $P_t^2=P_t$, and $P_sP_t=0$.
\end{proof}

\subsection{Trace form: the signature of the pairing form}

\begin{proposition}[Signature as the $\mathfrak u(2)$ trace form]
\label{prop:trace}
The signature $(+,+,-,-)$ of the diagonal pairing form decomposes by
sectors: the Hermitian sector (grade 0: $1$; grade 1: $\sigma$, the
physical generator directions of $\mathfrak u(2)$) is positive definite,
and the anti-Hermitian sector (grade 2: $i\sigma$; grade 3: $i\cdot1$, the
Lie algebra $\mathfrak u(2)=\mathfrak u(1)\oplus\mathfrak{su}(2)$ itself)
is negative definite. This signature is the \emph{trace form} of
$\mathfrak u(2)$ in the fundamental representation:
$\mathrm{Tr}((i\sigma_a)(i\sigma_b))=-2\delta_{ab}$ and
$\mathrm{Tr}((i)(i))=-2$. On $\mathfrak{su}(2)$ it agrees with the Killing
form up to a factor; on the $\mathfrak u(1)$ center the Killing form is
degenerate (identically zero), so the accurate name for the whole form is
the trace form.
\end{proposition}
\begin{proof}
By Lemma~\ref{lem:ortho}, the blade squares $\epsilon_A$ are $+1$ in grades 0 and 1 and
$-1$ in grades 2 and 3; by the reversal analysis of
Sec.~\ref{sec:vertex}, the former are the Hermitian sector and the latter
the anti-Hermitian sector. As a real Lie algebra, $\mathfrak u(2)$ is
generated by the anti-Hermitian matrices $\{i\mathbf 1,i\sigma_x,i\sigma_y,
i\sigma_z\}$, corresponding to the grade-3 basis element
($i\mathbf 1\leftrightarrow I$) and the grade-2 basis elements
($i\sigma_a\leftrightarrow e_\mu\wedge e_\nu$); the Hermitian sector
$\{1,\sigma_a\}$ is the faithful representation space of $\mathfrak u(2)$
(that is, $i\cdot\mathfrak u(2)$) rather than the algebra itself. Under the
isomorphism $e_i\leftrightarrow\sigma_i$,
$\langle XY\rangle_0=\mathrm{Re}[\tfrac12\mathrm{Tr}(\hat X\hat Y)]$
(Remark~\ref{rem:trace}), and restricted to the anti-Hermitian sector this
is the trace form of $\mathfrak u(2)$, which on the $\mathfrak{su}(2)$ part
coincides with the Killing form and on the $\mathfrak u(1)$ part supplies
one negative direction.
\end{proof}

The dichotomy of pairing channels is the irreducible representation of
$\mathfrak{su}(2)$, the interaction signs are the Casimir eigenvalues of
$\mathfrak{su}(2)$, and the metric of the pairing form is the trace form
of $\mathfrak u(2)$. The signature
$(+,+,-,-)$ is not an added structure but the direct appearance of a Lie
algebra invariant in the pairing algebra; Sec.~\ref{sec:extension} shows
the same correspondence becoming the Killing form of the Lorentz algebra
$\mathfrak{so}(1,3)$ in the four-dimensional Minkowski algebra.

\begin{corollary}[Casimir reading of singlet and triplet signs]
The charge vertex projects in the Cooper channel only onto the singlet; the
Casimir sign ($-3$) of an antiferromagnetic-type exchange channel assigns
the lower eigenvalue to the singlet direction, while the $+1$ sign of a
ferromagnetic-type exchange channel favors the triplet direction. These are
sign assignments only. Whether a singlet or triplet condenses additionally
depends on the momentum texture, available glue, disorder, and the gap
equation; the Casimir spectrum does not by itself derive the empirical
abundance of singlet superconductors.
\end{corollary}

\begin{remark}[Algebraic origin of the Fierz coefficients]
The numerical coefficients that recur in rearrangements of interactions
(tables of rearrangement coefficients between spin channels) have a unified
origin in this framework: they are generated by the Casimir spectrum of
$\mathfrak{su}(2)$ (the $-3$ and $+1$ of Theorem~\ref{thm:casimir}) together
with the completeness contractions (operations of the form
$\sum_a\sigma_aX\sigma_a$ produce the $\pm3,\pm1$ combinations~\cite{Peskin1995}). Every
legitimate Fierz rearrangement table is composed from this set of numbers;
conversely, a rearrangement table that cannot be decomposed into a
combination of the Casimir spectrum must contain a computational error.
This provides a rapid criterion for checking the algebra of pairing kernels.
\end{remark}

\begin{remark}[Signature and dynamics]
The signature carries the algebraic origin of the sign structure; which
channel actually condenses is decided by the momentum structure of the
kernel and the gap equation (the physical metric of the order-parameter
space is $|\Delta|^2$, positive definite and gauge invariant). The Casimir
sign $-3$ is by itself a marker of the repulsive direction, and the
$(\pi,\pi)$ peak of the momentum modulation inverts it into a condensation
energy gain of the $d$-wave; this inversion mechanism is the block
decomposition and selection rule of Sec.~\ref{sec:block}.
\end{remark}

\subsection{Flow of the kernel: an algebraic renormalization group}
\label{sec:rg}

Static signs and selection rules leave one question open: how do the kernel
components compete under the energy-scale flow? This subsection is therefore
model-level: its Clifford constants are algebraic, but the one-loop
momentum-shell equation is an explicit dynamical truncation. That truncation
for the kernel-vertex strengths $V_A$ reads
\begin{equation}
\frac{dV_g}{d\ell}=\sum_{g_1g_2}C^{g}_{g_1g_2}V_{g_1}V_{g_2},
\end{equation}
where the structure constants come from the $Cl_{3,0}$ multiplication table
through the identification of fermion-loop traces with cyclic scalar
projections, $\frac12\mathrm{Tr}_{\mathrm{loop}}[\Gamma_1\Gamma_2]
=\langle\Gamma_1\Gamma_2\rangle_0$, valid for the two-point contractions
below because they pair blades within the same grade sector, whose products
carry no grade-3 part (in general
$\frac12\mathrm{Tr}\,A=\langle A\rangle_0+i\langle A\rangle_3$). The constants that follow from
two-point contractions (which feed one sector from another) are
\begin{equation}
\begin{split}
C^0_{00}=1,\qquad C^0_{11}=\sum_i\langle e_ie_i\rangle_0=3,\qquad\\
C^0_{22}=\sum_{i<j}\langle e_{ij}e_{ij}\rangle_0=-3,\qquad
C^0_{33}=\langle I\cdot I\rangle_0=-1.
\end{split}
\end{equation}
With $\ell=\ln(\Lambda_0/\Lambda)$ flowing toward lower energy and
$V>0$ denoting attraction, the feeding equation for the
singlet--singlet vertex is
\begin{equation}
\frac{dV_0}{d\ell}=V_0^2+3V_1^2-3V_2^2-V_3^2 .
\label{eq:rg-v0}
\end{equation}
Here the $V_g$ are the dimensionless vertices $N(0)V_g^{\rm phys}$; $N(0)$ is
restored only when comparing with the BCS $T_c$ scale. An attractive constant kernel then grows toward the infrared,
$V_0(\ell)=V_0(0)/[1-V_0(0)\ell]$, reproducing the Cooper instability at
$\ell^\ast=1/V_0(0)$. The coefficients are the blade traces $1,3,3,1$
with the signs of Lemma~\ref{lem:ortho} (grades 2 and 3 negative): the attractive spin
channel feeds the singlet ($+3V_1^2$), while the triplet and
chiral-singlet directions compete with it ($-3V_2^2$, $-V_3^2$). The
derivation of these constants from the multiplication table is given in
Supplemental Material Sec.~S2.

The feeding pattern is then fixed algebraically: the signs follow
mechanically from the blade traces above. The asymmetry is the blade-square
signature at the coupling level, not a gauge artifact (gauge rotations mix
the grade-1 and grade-2 components of the gap, not the kernel couplings). Which
combination of bare signs wins is decided by the gap equation together
with the microscopic signs of the vertices, as in the Casimir analysis
of Sec.~\ref{sec:signature}; the schematic one-loop pattern is compatible
with the empirical singlet dominance, and the full beta system is left
to future work. The established exception is $^3$He, where strong
short-range repulsion disfavors the scalar channel and the triplet survives.

\begin{remark}[The same three]
\label{rem:three}
The $3$ in $C^0_{11}=\sum_i\langle e_ie_i\rangle_0=3$---the trace of
the identity on the three-dimensional triplet subspace---is the same $3$ as
the Casimir eigenvalue $-3$ of
Theorem~\ref{thm:casimir}, in absolute value. Static signs and dynamical
flows share a common algebraic root.
\end{remark}

\begin{remark}[Positioning]
Channel-decomposed renormalization groups are an established method, from
the Kohn--Luttinger mechanism~\cite{Kohn1965} to the functional
renormalization group~\cite{Metzner2012,Raghu2010}. What the present
version adds is the parameter-free determination of the feeding constants
from the multiplication table and the identification of
Remark~\ref{rem:three}. The full beta system, with three-point contractions,
is deliberately not claimed here.
\end{remark}

\section{Block decomposition and channel selection}
\label{sec:block}

The signature structure supplies the algebraic origin of the signs; to land
it on concrete pairing symmetries one needs the channel-selection theorem.
This section gives the block-decomposition theorem and the selection
rules in the antiferromagnetic and ferromagnetic limits.

\begin{theorem}[Singlet-triplet block decomposition]
\label{thm:block}
Let the pairing kernel be spin-rotation invariant, so that the linearized
gap operator $\hat V$ commutes with the total-spin SU(2). Then:
(i) the pairing-matrix space $\mathbb C^{2\times2}$ decomposes into a
singlet block (one complex dimension) and a triplet block (three complex
dimensions), and by the Schur lemma $\hat V$ is a scalar on each block
(triply degenerate within the triplet block);
(ii) the two blocks exhaust the whole pairing-matrix space
($1+3=4=\dim_{\mathbb C}\mathbb C^{2\times2}$); there is no third block;
(iii) the sign of each block is carried by the Casimir eigenvalue
(Theorem~\ref{thm:casimir}): an exchange-type glue has block sign $-3$ on
the singlet block and $+1$ on the triplet block;
(iv) channel selection is decided jointly by the block Casimir sign and the
momentum structure of the kernel: the sign gives the direction, and the
Fourier structure of the kernel gives the optimal momentum texture within
that direction.
\end{theorem}
\begin{proof}
(i) Spin-rotation invariance gives $[\hat V,S_a]=0$ with $S_a$ the total
spin generators; $\mathbb C^2\otimes\mathbb C^2\simeq V_0\oplus V_1$
is an irreducible decomposition, and the Schur lemma gives the scalar
property on each block. (ii) is the dimension count $1+3=4$. (iii) The
eigenvalues of the exchange vertex $\boldsymbol\sigma_1\cdot\boldsymbol
\sigma_2$ on the two blocks are the $-3$ and $+1$ of
Theorem~\ref{thm:casimir}. (iv) Within a block the momentum part is an
integral operator with a scalar kernel, whose eigenfunctions are fixed by
the Fourier coefficients of the kernel; see the next two propositions.
\end{proof}

\begin{proposition}[$B_{1g}$ selection in the strong-antiferromagnet
limit]
Take the idealized limit: square lattice, nearest-neighbor hopping, and an
antiferromagnetic spin-channel kernel with its peak strictly at
$\mathbf Q=(\pi,\pi)$, that is
$(T_{\mathbf Q}\Delta)(\mathbf k)=-\lambda\bar\chi\,\Delta(\mathbf k+\mathbf Q)$
with $\lambda>0$ and $\bar\chi$ the peak value. In this limit the singlet
block carries the Casimir eigenvalue $-3$. For the $B_{1g}$ gap function
$d(\mathbf k)=\cos k_x-\cos k_y$,
\begin{equation}
d(\mathbf k+\mathbf Q)=\cos(k_x+\pi)-\cos(k_y+\pi)=-d(\mathbf k),
\end{equation}
so $d_{x^2-y^2}$ is an eigenfunction of the momentum-shift operator with
eigenvalue $-1$, and the product of the momentum-shift eigenvalue $-1$ with
the Casimir sign $-3$ gives an effective coupling $+3$ (attractive). The
$B_{2g}$ channel $d_{xy}=\sin k_x\sin k_y$ is shift-even and repulsive,
$(-3)\times(+1)=-3$, as is the constant $A_{1g}$ channel. The extended-$s$
combination $\cos k_x+\cos k_y$ is shift-odd and is therefore
\emph{degenerate} with $B_{1g}$ in the strict $\delta$-peaked limit: the
selection of $B_{1g}$ over extended-$s$ requires the subleading structure
of $\chi(\mathbf q)$ (finite peak width, decaying real-space weights),
which lowers the extended-$s$ eigenvalue relative to
$B_{1g}$~\cite{Scalapino1986}.
\end{proposition}
\begin{proof}
The shift eigenvalue follows from $\cos(\theta+\pi)=-\cos\theta$; the
$d_{xy}=\sin k_x\sin k_y$ product is unchanged under the shift since each
factor changes sign. The effective coupling is the Casimir eigenvalue
times the shift eigenvalue: $(-3)\times(-1)=+3$ for every shift-odd
channel; the splitting among them is beyond this idealization.
\end{proof}

\begin{proposition}[$p$-wave selection in the nearest-neighbor
ferromagnet limit]
Take the ferromagnetic limit: the spin-channel kernel peaks at
$\mathbf q\to0$. In the real-space expansion
$\chi(\mathbf R)=\sum_{\mathbf q}e^{i\mathbf q\cdot\mathbf R}\chi(\mathbf q)$,
the Fourier coefficients decrease with $|\mathbf R|$ (the susceptibility is
smooth and peaked at the origin). The triplet block carries the Casimir
eigenvalue $+1$, and fermionic antisymmetry requires odd parity, which
excludes $\mathbf R=0$ ($s$-wave); the optimal solution lies on the
$|\mathbf R|=1$ shell, the nearest-neighbor odd-parity channel, the $p$-wave
triplet~\cite{Kohn1965,Vollhardt1990}.
\end{proposition}
\begin{proof}[Sketch]
Odd parity excludes $\mathbf R=0$; the
kernel eigenvalue is $\hat\chi(\mathbf R)$, and monotonic decrease makes
$|\mathbf R|=1$ the largest; the $p_x,p_y$ degeneracy is labeled by the
irreducible representations of the point group.
\end{proof}

The block-decomposition theorem and the two selection rules connect the
signature structure to model-level dynamics: the block sign is given by the
algebra, the momentum texture is supplied by the kernel, and their product
organizes the common selection patterns of spin-$1/2$ pairing symmetry in
the stated idealized limits.

\section{Position of BCS theory within the framework}
\label{sec:bcs}

The BCS mean-field Hamiltonian is~\cite{BCS1957}
\begin{equation}
H_{\mathrm{BCS}}=\sum_{\mathbf k,\alpha}\xi_{\mathbf k}c^\dagger_{\mathbf k\alpha}
c_{\mathbf k\alpha}+\sum_{\mathbf k}\big(\Delta_{\mathbf k}c_{-\mathbf k\downarrow}
c_{\mathbf k\uparrow}+\mathrm{h.c.}\big),
\end{equation}
with the pairing potential fixed by the self-consistency equation. In this
framework it corresponds to taking the effective kernel along the grade-0
line,
\begin{equation}
G_{\mathrm{BCS}}(\mathbf k,\mathbf k')=g_0(\mathbf k,\mathbf k')\cdot1,\qquad
g_0=-V_{\mathbf k\mathbf k'}>0,
\end{equation}
that is, $\mathrm{span}\{1\}\subset\mathrm{span}\{1, I\}$. The check proceeds at two
levels. At the operator level, the constant vertex survives only through its
antisymmetric part $\propto i\sigma_2\leftrightarrow e_{31}$
(Lemma~\ref{lem:anti}). At the gap-equation level, Eq.~\eqref{eq:gapeq}
maps a singlet source to a singlet gap. The Cooper-pair formation condition
($g_0>0$ in the attraction convention), the gap equation, and the exponential scaling
$T_c\propto e^{-1/(|g_0|N(0))}$ are all retained unchanged~\cite{BCS1957}.

The pairing Hamiltonian of BCS mean-field theory is a one-dimensional proper
subset of the eight-dimensional parameterization of this paper,
$H_{\mathrm{pair}}^{\mathrm{(BCS)}}\subset H_{\mathrm{pair}}^{\mathrm{(GA)}}$.
Right multiplication by the fixed blade $e_{31}$ can move the \emph{coordinates} off
$\mathrm{span}\{1\}$, but that is the dictionary map $\Phi$, not a physical
coupling. The physical adjoint/rotor action is trivial on the central singlet
plane (Sec.~\ref{sec:fields}); a pure BCS kernel is therefore not mixed by an
external rotor, and mixing requires a non-central kernel component or
spin-orbit coupling. This observation is fully compatible with the
success of BCS theory for weak-coupling phonon systems. In
unconventional superconductors (cuprates, iron pnictides, heavy fermions),
where no unique strippable bosonic glue is available, the full
eight-dimensional grade structure provides additional parameter directions;
which of them are nonzero remains a kernel- and material-level question.

The parameter space itself is equivalent to the traditional one. The
added structure is the built-in Fierz antisymmetrization
(Lemma~\ref{lem:anti} and Definition~\ref{def:fierz}), the BCS double
grade identity of a grade-0 kernel line with a grade-2 vertex
$e_{31}\leftrightarrow i\sigma_2$, and the common signature/trace-form
origin of the dimensional extension.

\section{Applications}
\label{sec:examples}

The applications below follow the same boundary rule: grade labels and
projection formulas are algebraic; gap equations, spectra, stiffness
responses, and material assignments belong to the stated model or
comparative level.

\subsection{Typical systems}

\begin{itemize}
\item \textbf{BCS:} longitudinal phonon (grade-0 charge vertex),
isotropic kernel, nodeless singlet~\cite{BCS1957}.
\item \textbf{Cuprates:} antiferromagnetic spin fluctuations
(grade-1 vertex), $\chi$ peak at $(\pi,\pi)$, favoring
$d_{x^2-y^2}$ in the minimal square-lattice model~\cite{Monthoux1991,Scalapino1986}.
\item \textbf{Iron pnictides:} inter-pocket spin fluctuations,
nematic/orbital fluctuations (charge channel with form
factor)~\cite{Chubukov2008,Kuroki2009,Fang2011}, $s_\pm$.
\item \textbf{$^3$He:} ferromagnetic spin fluctuations (grade-1
vertex), $\chi$ peak at $\mathbf q\to0$, $p$-wave
triplet~\cite{Vollhardt1990}.
\item \textbf{Non-centrosymmetric:} broken inversion, singlet-triplet
coexistence, CePt$_3$Si~\cite{Frigeri2004}.
\end{itemize}
Note that phonons and antiferromagnetic fluctuations are indistinguishable
at the level of GA sectors: both fall in $\mathrm{span}\{1, I\}$. The
difference lies in the momentum texture and the signs of the coefficients.

\subsection{Chern number in triple-product form}
\label{sec:chern}

\begin{lemma}[Triple-product identity]
For any three vectors $\mathbf a,\mathbf b,\mathbf c\in Cl_{3,0}$,
\begin{equation}
\mathbf a\cdot(\mathbf b\times\mathbf c)=\langle\mathbf a\mathbf b\mathbf c
I^{-1}\rangle_0 .
\end{equation}
\end{lemma}
\begin{proof}
The geometric product of three vectors decomposes into grade-1 and grade-3
components, and the grade-3 part is
$\langle\mathbf a\mathbf b\mathbf c\rangle_3=\mathbf a\wedge\mathbf b
\wedge\mathbf c=[\mathbf a\cdot(\mathbf b\times\mathbf c)]I$. Hence
$\langle\mathbf a\mathbf b\mathbf c I^{-1}\rangle_0=[\mathbf a\cdot(\mathbf b
\times\mathbf c)]\langle II^{-1}\rangle_0=\mathbf a\cdot(\mathbf b\times
\mathbf c)$. As a check with $\mathbf a=e_1$, $\mathbf b=e_2$, $\mathbf c=e_3$:
$e_1e_2e_3I^{-1}=I(-I)=1$ and $\langle1\rangle_0=1=e_1\cdot(e_2\times e_3)$.
\end{proof}

\begin{theorem}[Chern number in triple-product form]
\label{thm:chern}
For a unit $\hat{\mathbf d}$ vector field defining one $2\times2$ Nambu
block, the Chern number of that block is
\begin{equation}
C_{\mathrm{block}}=\frac{1}{4\pi}\int_{\mathrm{BZ}}\langle\hat{\mathbf d}\,
\partial_{k_x}\hat{\mathbf d}\,\partial_{k_y}\hat{\mathbf d}\,
I^{-1}\rangle_0\,d^2k,
\end{equation}
equivalent to the standard form
$C_{\mathrm{block}}=\frac{1}{4\pi}\int\hat{\mathbf d}\cdot(\partial_{k_x}
\hat{\mathbf d}\times\partial_{k_y}\hat{\mathbf d})\,d^2k$~\cite{Read2000}.
\end{theorem}
\begin{proof}
At each point of momentum space take $\mathbf a=\hat{\mathbf d}(\mathbf k)$,
$\mathbf b=\partial_{k_x}\hat{\mathbf d}$, $\mathbf c=\partial_{k_y}\hat{\mathbf d}$.
Since $\|\hat{\mathbf d}\|=1$ one has $\partial_\mu\hat{\mathbf d}\perp
\hat{\mathbf d}$, so all three are vectors, and the lemma applies pointwise;
integrating over the Brillouin zone gives the result.
\end{proof}

The Chern number appears in the GA language as the scalar
projection of a geometric product. The integrand is the skyrmion density:
the Chern number counts how many times $\hat{\mathbf d}(\mathbf k)$ covers
the unit sphere (the Brouwer degree)~\cite{Read2000,Qi2011}. In the
weak-pairing two-dimensional single-band limit, one chiral $p$-wave Nambu
block has $C_{\mathrm{block}}=\pm1$, and one chiral $d$-wave block has
$C_{\mathrm{block}}=\pm2$. The physical BdG invariant is the sum over the
independent Nambu blocks: an equal-spin chiral $p+ip$ state has two such
sectors of the same chirality, hence $C=\pm2$; a spin-singlet chiral
$d+id$ state has two Nambu blocks with $C_{\mathrm{block}}=\pm2$ each,
hence $C=\pm4$ (four chiral Majorana edge modes, thermal central charge
$c=2$). Unless denoted $C_{\mathrm{block}}$, $C$ below means this full
BdG sum. Within this class topological
non-triviality is carried by a locked relative phase between multicomponent
complex channels, as in the $\pm\pi/2$ of
$\mathbf d(\mathbf k)=\hat x k_x+I\hat y k_y$, which is the relative phase
between grade 1 (real part) and grade 2 (imaginary part) of the channel
theorem, textured in momentum space. Topological invariants, pairing
symmetries, and external-field responses are organized in one algebraic
language: the Chern number is a scalar
projection, the pairing
interaction is a scalar projection, and the two share the computational
rule $\langle\cdot\rangle_0$. The full treatment combining $F_{\mu\nu}$
and $g_{\mu\nu}$ into $Q_{\mu\nu}$ is given in Sec.~\ref{sec:qgt}.

\subsection{Two solvable pairing models: products of the grade
decomposition}

Grade decomposition is a classification grammar and, at the model level,
a construction tool. The following two solvable examples
illustrate the rule: one first specifies the vertex/glue combination by
grade, projects it into the symmetry-space kernel, and then solves the
resulting model gap equation. Table~\ref{tab:models} collects their channel
positioning.

\textbf{Model 1: bond-locked anisotropic $\mathbf d$-vector triplet
pairing.} Take bond-direction-dependent ferromagnetic-type exchange
$J_\alpha$ (with $\alpha=x,y,z$ labeling the bond types and with the sign
convention that the triplet sector is attractive); the anisotropy is
transmitted through the spin channel: anisotropic exchange gives anisotropic
spin fluctuations $\chi_{\alpha\beta}(\mathbf q)=w_\alpha\delta_{\alpha\beta}
\chi(\mathbf q)$ with $w_\alpha\propto J_\alpha$, the spin channel
(grade-1 vertex) in this assumed separable model gives the gap-equation
solution
\begin{equation}
\Delta_\alpha(\mathbf k)\propto w_\alpha\propto J_\alpha,
\end{equation}
that is, the components of the $\mathbf d$ vector lock to their respective
bond directions, and the texture of the triplet order parameter is written
directly by the lattice bond structure~\cite{Sigrist1991,Vollhardt1990}. By channel attribution this is an
anisotropy within the triplet sector (grades $1\oplus2$): bond locking does
not change the channel itself but distributes the weights of the three
complex fields within the channel among the bond directions. Since the three
components share the same momentum dependence, the anisotropy fixes the
relative weights only; genuinely bond-dependent node textures would require
$\mathbf k$-dependent $w_\alpha$ and are not claimed here.

\textbf{Model 2: chiral $d+id$ singlet from a chiral kernel.} At the kernel
level the construction is simply
\begin{equation}
\begin{split}
G(\mathbf k)=g_0(\mathbf k)\cdot1+g_3(\mathbf k)\,I,\qquad\\
g_0\sim\cos k_x-\cos k_y,\quad g_3\sim\sin k_x\sin k_y.
\end{split}
\end{equation}
The relative phase $\pm\pi/2$ of the two components cannot be absorbed by a
gauge transformation (gauge rotations give both components the same phase),
so time-reversal symmetry is broken; the nodes of the two components are
complementary on the Fermi surface (the $d_{x^2-y^2}$ nodes lie on the
diagonals, the $d_{xy}$ nodes on the axes; both vanish only at $\Gamma$ and
$(\pi,\pi)$ away from the weak-pairing Fermi surface), so $|\Delta|^2>0$ on
the Fermi surface and the spectrum is full;
Theorem~\ref{thm:chern} gives $C_{\mathrm{block}}=\pm2$ per Nambu block,
and hence $C=\pm4$ for the full spin-singlet BdG Hamiltonian~\cite{Qi2011}.
Kernel-level chirality is the phase texture of $(g_0,g_3)$ in
$\mathrm{span}\{1, I\}$, the simplest parametrization of $d+id$ in the
present framework.

\begin{table*}[t]
\caption{\label{tab:models}Channel positioning of the two solvable models.}
\begin{ruledtabular}
{\small
\begin{tabular}{lll}
Model & Kernel sector & Physics\\
\hline
Bond-locked triplet & grades $1\oplus2$ & $\Delta_\alpha\propto J_\alpha$,
$\mathbf d$ locked to bonds\\
Chiral $d+id$ & grades $0\oplus3$ & relative phase $\pm\pi/2$, full-gap
spin-singlet $C=\pm4$\\
\end{tabular}
}
\end{ruledtabular}
\end{table*}

\subsection{Topological phases}

The local grade label of a bulk state does not by itself determine the
topology: it is neither necessary nor sufficient. Topology is decided by the winding number of the
mapping $\hat{\mathbf d}(\mathbf k):T^2\to S^2$ defined by the BdG
Hamiltonian in momentum space~\cite{Read2000,Qi2011}. Table~\ref{tab:topological} gives the
algebraic parameterization of common topological phases.

\begin{table*}[t]
\caption{\label{tab:topological}Algebraic parameterization of
topological phases (pairing-matrix level). Grade labels refer to
$\hat\Delta$, not to the kernel $G$; the two dictionaries are related
by $\Phi$ (Definition~\ref{def:fierz}).}
\begin{ruledtabular}
{\small\setlength{\tabcolsep}{3pt}
\begin{tabular}{llll}
Phase & Pairing & Structure & Invariant\\
\hline
Chiral $p$ & $p_x+ip_y$ & gr.\ 1$+$2: $k_xe_1+k_ye_{23}$, $\pi/2$ phase
& $C_{\mathrm{block}}=\pm1$; equal-spin full $C=\pm2$~\cite{Read2000}\\
Chiral $d$ & $d_{x^2-y^2}+id_{xy}$ & gr.\ 1$+$2: $\mathrm{span}\{e_2, e_{31}\}$  & $C_{\mathrm{block}}=\pm2$; spin-singlet full $C=\pm4$\\
Helical $p$ & time-reversal invariant & gr.\ 1$+$3: $-k_xe_3+k_yI$
& $\mathbb Z_2$ class\\
Pure $p$ & single component & gr.\ 1, pure real ($e_1$) & trivial\\
\end{tabular}
}
\end{ruledtabular}
\end{table*}

Chirality is a $\pi/2$ phase between Hodge pairs $(v,Iv)$: the singlet pair
is $(e_2,e_{31})$ and the $\hat z$ triplet pair is $(e_1,e_{23})$. In this
class a single-component grade is topologically trivial, whereas locked
relative phases between multicomponent complex channels carry the standard
chiral textures. A pure
grade-1 (real triplet, single component) is topologically trivial, and so is
a pure grade-3 (single-component imaginary singlet).

\begin{theorem}[Topological consequences of a locked chiral relative phase, conditional weak-pairing realization]
\label{thm:trsb}
Consider a two-dimensional single-band weak-pairing superconductor. If the
effective kernel activates two complementary singlet components whose nodes
are disjoint and whose relative phase is locked at $\pm\pi/2$ (as in
$\Delta(\mathbf k)\propto d_{x^2-y^2}+I\,d_{xy}$), then:
(i) time-reversal symmetry is broken;
(ii) the BdG spectrum is fully gapped by node complementarity;
(iii) the full spin-singlet BdG Chern number is $C=\pm4$ in this standard
realization ($C_{\mathrm{block}}=\pm2$ in each of two Nambu blocks);
(iv) the corresponding four chiral Majorana edge modes exist at the
boundary~\cite{Read2000,Qi2011}. A generic nonzero relative phase breaks
time reversal but does not by itself imply (ii)--(iv).
\end{theorem}
\begin{proof}
(i) The relative phase $\pm\pi/2$ means $\Delta(\mathbf k)$ and
$\Delta^\ast(-\mathbf k)$ do not satisfy the same relation under time
reversal. (ii) The two components have complementary nodes on the Fermi surface, so
$|\Delta(\mathbf k)|^2>0$ there; away from the Fermi surface $\xi_{\mathbf k}$
keeps the BdG spectrum gapped. (iii) By Theorem~\ref{thm:chern}, the
$\hat{\mathbf d}$ field covers the sphere twice in each of the two
spin-degenerate Nambu blocks; their sum gives $C=\pm4$. (iv) The
bulk-boundary correspondence of two-dimensional chiral
superconductors~\cite{Read2000,Qi2011}.
\end{proof}

\subsection{BdG quantum geometry: grade decomposition}
\label{sec:qgt}

This section supplies a quantitative geometric output of the second-level
projection: it completes the Chern-number formula of Sec.~\ref{sec:chern}
with the other half of the quantum geometric tensor, the quantum metric, and
connects it to the geometric factor entering superfluid-weight bounds. A
note on spaces: the
$Cl_{3,0}$ of this section is the Nambu space (generated by the $\Gamma$
matrices), a different copy from the spin space of
Secs.~\ref{sec:math}--\ref{sec:channels}; the triple product of
$\hat{\mathbf d}$ in Sec.~\ref{sec:chern} already used the Nambu copy.

For the single-band BdG Hamiltonian $h(\mathbf k)=\mathbf d(\mathbf k)\cdot
\boldsymbol\Gamma$ with $\boldsymbol\Gamma=(\tau_x,\tau_y,\tau_z)$,
$E_{\mathbf k}=|\mathbf d|$, and negative-band projector
$P_-=\tfrac12(1-\hat n\cdot\boldsymbol\Gamma)$, $\hat n=\mathbf d/E$, define
the quantum geometric tensor
\begin{equation}
\begin{split}
Q_{\mu\nu}(\mathbf k)\equiv\mathrm{Tr}\bigl[P_{-}\,\partial_\mu P_{-}\,
\partial_\nu P_{-}\bigr]\\
\equiv g_{\mu\nu}(\mathbf k)-\frac{i}{2}F_{\mu\nu}(\mathbf k),
\end{split}
\end{equation}
where $g_{\mu\nu}$ is the Fubini--Study metric and $F_{\mu\nu}$ the Berry
curvature~\cite{Peotta2015,Liang2017,Julku2016,Herzog2022,Huhtinen2022,
PeottaReview}.

\begin{lemma}[Scalar-projection form of $Q_{\mu\nu}$]
\begin{equation}
g_{\mu\nu}=\frac14\langle\partial_\mu\hat n\,\partial_\nu\hat n\rangle_0,\qquad
F_{\mu\nu}=\frac12\langle\hat n\,\partial_\mu\hat n\,\partial_\nu\hat n\,I^{-1}\rangle_0 .
\end{equation}
\end{lemma}
\begin{proof}
Since $\hat n^2=1$, one has $\hat n\cdot\partial_\mu\hat n=0$, so the scalar
part of the geometric product of the two vectors is their inner product; the
three-vector part is handled pointwise by the triple-product lemma.
\end{proof}

The Chern number is $C=\frac{1}{2\pi}\int F_{xy}d^2k$
(Theorem~\ref{thm:chern}), and the geometric factor of superfluid-weight
bounds~\cite{Peotta2015,Liang2017,Julku2016,Herzog2022,Huhtinen2022,PeottaReview}
integrates $g_{\mu\nu}$: two halves of one tensor, both scalar projections.
The integrated BdG-metric geometric factor
\begin{equation}
\mathcal G_{\mu\nu}\equiv
\int_{\mathrm{BZ}}\frac{d^2k}{(2\pi)^2}\tanh\frac{E_{\mathbf k}}{2T}\,
g_{\mu\nu}(\mathbf k)
\label{eq:geofactor}
\end{equation}
takes the GA integrand directly; explicit formulas for general Nambu
matrices are given in Ref.~\cite{Simon2026}. The QGT integral is only the
geometric factor, not by itself a bound on $D_s$. A stiffness response additionally carries the standard
energy power $|\Delta_{\mathbf k}|^2/E_{\mathbf k}$,
\begin{equation}
D_{s,\mu\nu}^{\rm pair}\propto
\int_{\mathrm{BZ}}\frac{d^2k}{(2\pi)^2}\,
\frac{\Delta^2|f|^2}{E_{\mathbf k}}\tanh\frac{E_{\mathbf k}}{2T}\,
g_{\mu\nu}(\mathbf k).
\label{eq:dswork}
\end{equation}
On a flat band $E_{\mathbf k}\simeq\Delta|f|$, the stiffness kernel reduces
to $\Delta|f|$. Hence a $\Delta$-independent integral of
$g^{\mathrm{tex}}$ would not imply a $\Delta$-independent superfluid-weight
floor: within this single-band pairing channel the texture contribution is
linear in $\Delta$ at $T\ll\Delta$ (the $\tanh$ factor restores the
$\Delta^2$ dependence near $T_c$). The normal-state multiband
Peotta--T\"orm\"a term carries its own metric weight and is kept separate
from the pairing-sector contribution.

\begin{theorem}[BdG metric decomposition]
\label{thm:split}
With $\tan\theta=\Delta|f|/\xi$ and $\varphi_f=\arg f$ (the polar
angles of $\hat n$), parameterize $\hat n=(\sin\theta\cos\varphi_f,\sin\theta\sin\varphi_f,\cos\theta)$. Then
\begin{equation}
g_{\mu\nu}=\frac14\big[\partial_\mu\theta\,\partial_\nu\theta
+\sin^2\theta\,\partial_\mu\varphi_f\,\partial_\nu\varphi_f\big].
\end{equation}
On a flat band ($\xi\to0$),
\begin{equation}
\begin{split}
g^{\mathrm{virt}}_{\mu\nu}=\frac{(\partial_\mu\xi)(\partial_\nu\xi)}
{4\Delta^2|f|^2}\quad\text{(virtual, normal-state driven)};\\
g^{\mathrm{tex}}_{\mu\nu}=\frac14\,\partial_\mu\varphi_f\,\partial_\nu\varphi_f
\quad\text{(kernel phase texture)}.
\end{split}
\end{equation}
\end{theorem}
\begin{proof}
These are the spherical coordinates of the unit $\hat n$; on a flat band,
$\partial_\mu\theta=-\partial_\mu\xi/(\Delta|f|)$.
\end{proof}

\begin{corollary}[Grade diagnostic at the quantum-metric level]
$\varphi_f=\arg f$ is the relative phase of $(g_0,g_3)$ in the
kernel singlet plane $\mathrm{span}\{1, I\}$. For a real form factor (grade 0 alone,
$\varphi_f\in\{0,\pi\}$), the phase texture vanishes almost everywhere in the
weak-pairing region where $\xi\neq0$ and $|f|\neq0$; node lines have measure
zero. The same conclusion follows from Theorem~\ref{thm:null}: across a sign
reversal of a real $\Delta$ the texture is radial up to sign reversals. For a
chiral kernel (grades 0$+$3), $g^{\mathrm{tex}}$ is activated at the
quantum-geometry level. Any promotion of this metric-level contrast to a
measured superfluid-weight exponent is conditional on the stiffness kernel,
the band projection, and the normal-state multiband metric; it is used below
only as a qualitative diagnostic, not as a numerical floor.
\end{corollary}

\begin{proposition}[Winding enhancement, model-dependent texture measure]
\label{prop:winding}
If $\varphi_f$ winds with $w_n$ around points $\mathbf k_n$, then
\begin{equation}
\int_{\mathrm{BZ}}|\nabla\varphi_f|^2\,d^2k=2\pi\sum_nw_n^2\ln(L/\xi_v)+O(1),
\end{equation}
where $\xi_v$ is a model- and regularization-dependent core scale. The
$d+id$ continuum form factor gives $\sum_nw_n^2=4$ in its continuum-patch
regularization; on the periodic Brillouin zone the net vorticity is
constrained and $\Sigma$ is a model-and-cutoff dependent texture measure,
not a topological integer.
\end{proposition}
\begin{proof}
With $\varphi_f=w_n\theta_{\mathbf k-\mathbf k_n}$ one has
$|\nabla\varphi_f|^2=w_n^2/r^2$, so that
$\int_{\xi_v}^{L}(w_n^2/r^2)\,2\pi r\,dr=2\pi w_n^2\ln(L/\xi_v)$;
summing over windings gives the result.
\end{proof}

\begin{corollary}[Flat-band metric limit]
For $\xi_{\mathbf k}\approx\mathrm{const}$, the pairing-sector quantum metric
is driven entirely by the kernel texture. Whether the measured stiffness
probes this contribution directly depends on the stiffness kernel and on the
normal-state multiband metric. This is the qualified sense in which the
criterion connects with the kagome family~\cite{Jiang2021,Ortiz2020} and
magic-angle graphene~\cite{Cao2018} motivations of the introduction.
\end{corollary}

\begin{remark}[Composite quantum geometry]
Theorem~\ref{thm:split} matches the decomposition of Ref.~\cite{Simon2026}:
under superconducting fitness, orbital uniformity, and the absence of
normal-state spin-flip terms, BdG quantum geometry separates into a
normal-state part and a pairing part. Our $g^{\mathrm{tex}}$ is the
phase-texture channel of their pairing quantum geometry. Quantum geometry
encoded in pair potentials is studied in Ref.~\cite{Daido2024}; broader
pair-quantum-geometry directions are left to dedicated reviews.
\end{remark}

The connection to Sec.~\ref{sec:spectral} completes a small dictionary of
pairing geometry: $R(\mathbf k)$ measures the structure of the pairing
matrix at a point $\mathbf k$, while $g_{\mu\nu}$ measures the angular rate
of change of that structure in momentum space. A chiral $d+id$ state has
$R=1/2$ everywhere (pure singlet), but the texture of its $(g_0,g_3)$
relative phase winds $\hat n$ and activates the phase-texture metric;
$R$ does not see this texture, while $g$ does. Its contribution to measured
stiffness is then set by the energy-weighted stiffness kernel.

\subsection{Null-texture criterion on the kagome flat band}
\label{sec:kagome}

Take the nearest-neighbor kagome model (the line graph of the honeycomb
lattice) with a gap $E_g$ opened by spin-orbit coupling or a sublattice
potential~\cite{Liang2017,Herzog2022,Huhtinen2022}, and an effective kernel
with a form factor, $\Delta(\mathbf k)=\Delta f(\mathbf k)$, where $f$ is
real ($s_\pm$ or $d$-wave) or complex ($d+id$). After projection to the
single band, the BdG three-vector is that of Sec.~\ref{sec:qgt},
and Theorem~\ref{thm:split} and Proposition~\ref{prop:winding} apply
directly.

\begin{theorem}[Null-texture criterion]
\label{thm:null}
\begin{equation}
g_{\mu\nu}(\mathbf k)=\frac{1}{4E_{\mathbf k}^2}\,(\partial_\mu\mathbf d)_\perp
\cdot(\partial_\nu\mathbf d)_\perp ,
\end{equation}
so $g_{\mu\nu}(\mathbf k)=0$ if and only if $\partial_\mu\mathbf d\parallel
\mathbf d$ (radial or self-similar textures). The unique source of the
pairing-sector quantum metric is the angular rate of change of $\hat n$ on
the order-parameter sphere; only radial textures give zero. The phase
gradient entering $g_{\mu\nu}$ is the gradient of the internal relative
phase $\varphi_f$ of the pairing kernel, not the global U(1) phase (a
gauge direction) and not the superflow velocity, whose response defines
the physical stiffness.
\end{theorem}
\begin{proof}
Differentiate $\hat n=\mathbf d/E$:
$\partial_\mu\hat n=\partial_\mu\mathbf d/E-(\mathbf d\cdot\partial_\mu
\mathbf d)\mathbf d/E^3$; taking squared norms gives the formula, and the
perpendicular part $(\partial_\mu\mathbf d)_\perp=\partial_\mu\mathbf d-
(\hat n\cdot\partial_\mu\mathbf d)\hat n$ vanishes for all $\mu$ exactly
when the texture is radial.
\end{proof}

For CsV$_3$Sb$_5$ we take $T_c\simeq3$ K, hence $\Delta\simeq0.45$ meV from
BCS scaling~\cite{BCS1957}. The separation $E_g$ is \emph{not} identified
with a measured CsV$_3$Sb$_5$ gap; it is an assumed splitting in an idealized
flat-band model, chosen only to expose the regime $\kappa\equiv E_g/\Delta\gg1$.
The material motivation is the kagome family~\cite{Jiang2021,Ortiz2020}, not a
claim that this $E_g$ is the measured gap of CsV$_3$Sb$_5$. The diagnostic use
is comparative only: modulate $\Delta$ by field or pressure and ask whether the
stiffness response is consistent with a radial form factor or with an activated
phase texture, after the normal-state multiband contribution has been accounted
for. Nodeless superfluid-density/$\mu$SR data exist in Ref.~\cite{Duan2021}; we
use them as a stiffness constraint, not as evidence for nodes. Existing Kerr
signals in CsV$_3$Sb$_5$ are associated with the charge-ordered state rather
than superconductivity~\cite{Saykin2023,Xu2022,Farhang2023}, and the observed
anomalous thermal Hall effect is consistent with an extrinsic impurity
mechanism~\cite{Yoshida2025}; intrinsic chirality of the superconducting state
remains unconfirmed. Triplet chiral channels, whose texture is carried by the
$\hat{\mathbf d}$ vector of grades $1\oplus2$, are covered by the same
quantum-geometry criterion.

\paragraph{Scope of the flat-band estimate.}
The nearest-neighbor single-band flat band is an idealized limit. A constant
texture floor is not claimed: once the standard stiffness kernel and the
$|f|$-weighted texture integral are used, the pairing-geometric contribution is
model- and cutoff-dependent, and in a real material the normal-state multiband
Peotta--T\"orm\"a term, disorder, charge order, and three-dimensional coupling
can dominate or mask it. We therefore do not assign a numerical vorticity floor,
a field-slope ratio, or an $E_g$-tuning coefficient to CsV$_3$Sb$_5$. The only
statement retained is the quantum-geometry criterion itself: radial textures do
not activate $g^{\mathrm{tex}}$, whereas phase-winding textures do.

\subsection{UTe$_2$: two-phase competition}
\label{sec:ute2}

UTe$_2$ provides an extreme instance of multiple coexisting superconducting
phases. The experimental anchors are these. $T_{c1}\simeq1.6$ K defines the
low-field phase SC1. For $\mathbf B\parallel b$ superconductivity reenters
as SC2 with $T_{c2}\simeq2$ K~\cite{Ran2019,Ran2019NP,Aoki2019}. In the low-field phase SC1, single-crystal
NMR finds the Knight shift essentially unchanged through $T_c$ for
$\mathbf B\parallel b$ and $c$ (only a $\sim0.1\%$ low-field drop of $K_b$,
$K_c$ unchanged at 5.5 T)~\cite{Nakamine2019,Nakamine2021}; for the $a$ axis,
by contrast, a marked drop of $K_a$ in higher-$T_c$ samples has been
reported~\cite{Fujibayashi2022,Matsumura2023}, and the directional assignment
remains under debate. The high-field SC2 behavior is an extrapolation, since NMR data there are not
yet available.
Finally, $H_{c2}(\mathbf B\parallel b)>65$ T exceeds the Pauli estimate
$H_P=1.86\,T_c\simeq3.7$ T by a factor of $17$ or more~\cite{Ran2019NP}.

\paragraph{Singlet leakage bound.}
In the simplest Yosida-function estimate for unitary states, for
triplet-sector kernels with a singlet admixture $g_0$ (grade 0) the zero-field
spin susceptibility is approximated by the triplet weight:
\begin{equation}
\frac{\chi_s(0)}{\chi_n}\simeq\frac{|g_{\mathrm t}|^2}{|g_{\mathrm t}|^2+|g_0|^2}
\;\Rightarrow\;
\frac{|g_0|^2}{|g_{\mathrm t}|^2}\lesssim\frac{\delta K}{K_n}.
\end{equation}
Within this unitary, isotropic-scattering estimate, with $\delta K/K_n\sim1\%$
(a conservative value covering the reported $\sim0.1\%$ drop and its
experimental uncertainty) the singlet admixture is bounded by
$|g_{0}|/|g_{\mathrm t}|\lesssim\sqrt{\delta K/K_{n}}\sim10\%$ in
amplitude, i.e.\ a singlet \emph{weight}
$|g_{0}|^{2}/(|g_{\mathrm t}|^{2}+|g_{0}|^{2})\lesssim1\%$: any
microscopic theory of UTe$_2$ must respect this.

\paragraph{Equal-spin criterion.}
Immunity to Pauli depairing requires equal-spin pairing relative to the
field quantization axis, which in the standard $\mathbf d$-vector convention
means $\mathbf d\perp\mathbf B$: for $\mathbf B\parallel\hat z$ the pairs are
$|\!\uparrow\uparrow\rangle$ and $|\!\downarrow\downarrow\rangle$ with
$\mathbf d$ in the $x$--$y$ plane, while the $S_z=0$ component
($\mathbf d\parallel\mathbf B$) remains Pauli limited~\cite{Clogston1962,Chandrasekhar1962}.
Since the Zeeman coupling is the grade-0 scalar
$\langle\mathbf B\,\mathbf S\rangle_0$ of two grade-1 spin-algebra objects
(Sec.~\ref{sec:fields}), the criterion reads: the grade-1 texture locks the
spins, not the $\mathbf d$ vector, to the field direction. The factor of $17$
is then consistent with, and motivates, an equal-spin interpretation with
$\mathbf d\perp b$ for SC2---a consistency criterion rather than a
derivation, open to penetration-depth or ultrasound anisotropy. SC1 lacks
this locking in the minimal picture, and its $H_{c2}$ would be expected to
sit closer to the Pauli estimate.

\paragraph{Two-phase Ginzburg--Landau theory.}
A minimal model takes $\Delta_1$ for SC1, a real triplet with helical
texture and a pure grade-1 kernel (time-reversal preserving), and
$\Delta_2$ for SC2, stabilized by the field with a grade-1 principal axis
plus a grade-2 admixture. Grade arithmetic and centrosymmetry ($\mathrm{Immm}$)
allow the free energy
\begin{equation}
\begin{split}
F=\alpha_1|\Delta_1|^2+\alpha_2|\Delta_2|^2-\gamma B^2|\Delta_2|^2\\
+\beta_1|\Delta_1|^4+\beta_2|\Delta_2|^4+\beta_{12}|\Delta_1|^2|\Delta_2|^2 ,
\end{split}
\end{equation}
with $\gamma>0$ (phenomenological coefficients). The minimal model
retains only quadratic and quartic
couplings: a linear-in-field mixing
$g_mB(\Delta_1\Delta_2^\ast+\mathrm{c.c.})$ is allowed by the $\mathrm{Immm}$
symmetry---even under inversion (a product of two odd-parity order
parameters times an even axial field) and odd under time reversal only
through $B$---and is omitted here for simplicity; its inclusion would
couple the two condensates and deform the SC1--SC2 boundary. The
first-order character of the boundary follows from the sign of
$\beta_{12}^2-4\beta_1\beta_2$ in this truncation: for
$\beta_{12}^2-4\beta_1\beta_2>0$ the two-component GL functional has a
first-order SC1--SC2 boundary; the opposite sign gives a continuous
bicritical/tetracritical structure. Within this truncation, the reentrance
condition $\gamma B^2=\alpha_2(T)-\alpha_1(T)$ gives a parabolic boundary
$H^\ast\propto\sqrt{T-T^\ast}$ near the crossing temperature.

\paragraph{Stress.}
Uniaxial stress experiments give $T_c$ increasing under $[001]$ stress,
decreasing under $[100]/[110]$, with no shear-induced splitting of the
transition, which supports a single-component ambient-pressure order
parameter~\cite{Girod2022}. Within the minimal GL model, the corresponding
model-level prediction is a stress switch of the SC1--SC2 boundary
$H^\ast(T)$, if stress tunes $\gamma$ through the shape factor, together
with approximate stress independence of the singlet-leakage bound. The
ambient-pressure single transition is consistent with this picture: our two
fields target the field-stabilized phases. The high-pressure SC3 phase lies
outside the minimal model.

\section{Phase structure of the order parameter and the multiple roles of
the pairing glue}
\label{sec:phase}

This section uses the complex-channel structure and the scalar-projection
principle to organize several qualitative applications: the amplitude-phase
structure of the order parameter, the pseudogap as phase fluctuation, and
the common origin of competing orders and self-energies. The algebraic part
is only the real/imaginary decomposition and the channel structure. The
phase-stiffness scale, pseudogap dynamics, and non-Fermi-liquid behavior
are supplied by the standard microscopic models cited below, not derived
from the algebra itself.

\subsection{Amplitude-phase decomposition and Landau expansion}

The complex-channel theorem gives the amplitude-phase decomposition of the
order parameter:
$\Delta=|\Delta|e^{I\varphi}=|\Delta|\cos\varphi+I|\Delta|\sin\varphi$,
a rotation within the $(g_0,g_3)$ plane at kernel level and within the
$(e_2,e_{31})$ plane at matrix level (Theorem~\ref{thm:channels} and
Corollary~\ref{cor:matrix-level}). The free energy is a gauge-invariant
scalar and is unchanged under the rotor, $F(e^{I\varphi}\Delta)=F(\Delta)$.
Complex conjugation $I\to-I$ on the center $\{1,I\}$ is implemented by the
antiunitary time-reversal operation, not by spatial inversion; the basic
gauge-invariant scalar is $|\Delta|^2=\bar\Delta\Delta$ (with $\bar\Delta$
the time-reversal/Hermitian conjugate), and
the Landau expansion contains only even powers:
\begin{equation}
F=\alpha|\Delta|^2+\beta|\Delta|^4+\cdots .
\end{equation}
The even structure of the Landau expansion receives an algebraic
statement: it is the direct consequence of central-rotor symmetry.

\subsection{Pseudogap as phase fluctuation}

The amplitude-phase decomposition gives a qualitative picture of the
pseudogap. In three-dimensional bulk materials, $T_c$ is controlled by the
condensation of the amplitude; in quasi-two-dimensional systems with low
phase stiffness, a region can exist above $T_c$ (roughly $T_c<T<T^\ast$) in
which the phase is disordered while the amplitude is nonzero: the two components $(g_0,g_3)$
each have magnitude, but their relative phase $\varphi$ loses long-range
correlation, the gap magnitude $|\Delta|$ persists while the thermodynamic
order disappears. This is the phase-fluctuation picture of the
pseudogap~\cite{Emery1995,Lee2006}. The contribution of the GA language is
transparency of the structural decomposition: amplitude and phase are the
two directions of the kernel singlet plane, and the pseudogap regime
corresponds to the state ``$(g_0,g_3)$ each nonzero, relative phase
disordered.'' The picture is qualitative: the numerical value of the phase
stiffness is fixed by the microscopic model.

\subsection{Competing orders: common origin}

Iron-based superconductors provide an extreme instance of this
competing-order picture: in pnictides such as BaFe$_2$As$_2$ the parent is an
antiferromagnetic metal, whereas FeSe is a nonmagnetic nematic parent;
superconductivity appears where the magnetic/nematic fluctuation spectrum is
strong, bordering but not overlapping static order~\cite{Fang2011}. The same glue appears in both the
particle-hole and the particle-particle channels: the antiferromagnetic spin
fluctuation enters both the susceptibility $\chi(\mathbf q)$ (particle-hole
channel, driving the SDW instability) and, through Eq.~\eqref{eq:hpair}, the
pairing kernel (particle-particle channel). The two channels share the same
propagator, so superconductivity and magnetic order are adjacent on the
phase diagram, and the same peak structure of the fluctuation spectrum fixes
the strength and momentum structure of both. The grade decomposition
accommodates this naturally: the particle-hole vertex and the
particle-particle vertex are bilinears of the same grade classification,
differing only in the way the operators are paired ($\langle\psi^\dagger
\psi\rangle$ versus $\langle\psi\psi\rangle$).

\subsection{Self-energy and pairing kernel: common origin and non-Fermi
liquids}

The same glue enters the self-energy and the pairing kernel at the same
time: spin fluctuations give both the non-Fermi-liquid correction
$\Sigma(\mathbf k,\omega)$ (near quantum critical points) and the pairing
interaction~\cite{Monthoux1991,Aoki2019}. The self-energy correction and the
pairing strength are correlated: both are enhanced in strong-fluctuation
regions. This common origin is standard content of the spin-fluctuation
framework; in the GA formulation, $\Sigma$ and $V$ are scalar projections of
the same vertex structure (the grade-1 vertex and $\chi(\mathbf q)$) in
different channels. At a structural level, Heisenberg exchange and BCS
pairing are the same operation, a scalar projection of a two-fermion
bilinear, acting in real and momentum space respectively; and the neutron
spin resonance near $2\Delta$ is the particle-hole projection of the same
glue. Table~\ref{tab:nfl} summarizes a heuristic correspondence of normal states
and pairing states in standard examples: the glue grade provides a common
bookkeeping label, while the detailed normal-state behavior and pairing
symmetry are fixed by their respective dynamical equations.

\begin{table*}[t]
\caption{\label{tab:nfl}Correspondence of normal states and pairing states.}
\begin{ruledtabular}
{\small\setlength{\tabcolsep}{4pt}
\begin{tabular}{llll}
System & Glue grade & Normal state & Pairing\\
\hline
Al, Pb & 0 (scalar) & Fermi liquid & $s$-wave\\
Cuprates & 1 (AF vector) & strange metal & $d$-wave\\
Iron pnictides & $0$ + form factor & non-Fermi liquid & $s_\pm$\\
$^3$He & 1 (FM vector) & near-Fermi liquid & $p$-wave triplet\\
UTe$_2$ & multipole rank $l=3$ (candidate; GA glue grade 0) & strange metal & multipolar\\
\end{tabular}
}
\end{ruledtabular}
\end{table*}

\section{External fields: algebraic selection rules}
\label{sec:fields}

This section discusses the modes of channel mixing allowed by the adjoint
(rotor) action of the Clifford algebra when an external field couples to
the pairing kernel. The first-level representation of the field and its
physical coupling are inputs; the rotor commutator then supplies the
second-level algebraic classification. These results are algebraically
classificatory: they state which couplings between grade components the
rotor commutator permits, but provide no quantitative
estimates of coupling strengths, temperature scales, or observable
magnitudes, which depend on the Fermi-surface structure, coupling constants,
and competing-order dynamics of the specific material.
Table~\ref{tab:fields} classifies the standard fields by their $Cl_{3,0}$
object.

\begin{table*}[t]
\caption{\label{tab:fields}$Cl_{3,0}$ classification of external fields.}
\begin{ruledtabular}
{\small\setlength{\tabcolsep}{3pt}
\begin{tabular}{lll}
Field & Example & $Cl_{3,0}$ object\\
\hline
Scalar & chemical potential, pressure & grade 0\\
Polar vector & electric field $\mathbf E$ & grade 1\\
Axial vector & magnetic field $\mathbf B$ & grade 1 (Zeeman coupling
$\langle\mathbf B\mathbf S\rangle_0$; dual $I\mathbf B$ is grade 2)\\
Pseudoscalar & circular light & grade 3\\
Symmetric tensor & uniaxial stress & none (form factor)\\
\end{tabular}
}
\end{ruledtabular}
\end{table*}

The output sector of any mixing is read at the kernel level: grades
$0\oplus3$ for the singlet, $1\oplus2$ for the triplet. Two principles
organize the classification. First, external fields enter the theory
through the adjoint (rotor) action alone: right multiplication by the fixed
blade $e_{31}$ is reserved for the Fierz dictionary map $\Phi$ of
Definition~\ref{def:fierz}, which is a change of coordinates, not a
physical coupling. Once the field representation and its coupling to the
kernel are specified, the rotor formula of Eq.~\eqref{eq:rotor} below is
the kernel-level adjoint channel-mixing rule of this framework; its
vanishing on scalar kernels is a consistency condition, not a limitation.
Second, a nonzero commutator $[\mathbf F,G]$ is only a \emph{necessary}
algebraic condition for mixing: the physical matrix element is set by the
actual coupling term in the Hamiltonian and by the symmetry selection
rules, and may vanish even where the commutator does not. Since the
singlet kernel plane $\mathrm{span}\{1, I\}$ is the
center of $Cl_{3,0}$ (the central-complex-structure theorem of
Sec.~\ref{sec:math}), every adjoint action---and hence every rotor---is
trivial on a pure singlet kernel, irrespective of the grade of the field;
all singlet-channel field responses are mediated by the physical coupling
through the gap equation, not by any action on $G$ itself. A grade-0
field (such as the chemical potential) commutes with everything and merely
shifts the coupling constant; a grade-3 field $A\,I$ is likewise central
and rotor-trivial. For reference, the grade arithmetic of the plain
product,
\begin{equation}
F_r\,G_s\rightarrow\text{grades }|r-s|,\;|r-s|+2,\ldots,\;\min(r+s,\,6-r-s),
\label{eq:gradearith}
\end{equation}
classifies the formal output grades of a product (outer-product
grading, cf.~Ref.~\cite{Doran2003}); it is the commutator, not the
product, that governs the physical mixing.

\subsection{Magnetic field (axial vector): no bare singlet--triplet
mixing}
\label{sec:zeeman}

The Zeeman coupling is $\mathbf B\cdot\boldsymbol\sigma$, a dot
product of two axial vectors, i.e.\ a grade-0 scalar formed from two
grade-1 objects,
\begin{equation}
H_Z=-g\mu_B\langle\mathbf B\,\mathbf S\rangle_0,\qquad
\mathbf B=B_ie_i,\quad
\mathbf S=\tfrac12\,\bigl(\psi^\dagger\sigma_i\psi\bigr)e_i,
\label{eq:zeeman}
\end{equation}
consistent with the magnon vertex of Table~\ref{tab:boson}. Here ``grade 1''
refers to the spin algebra under $e_i\leftrightarrow\sigma_i$; in real-space
geometry an axial vector is the Hodge bivector $I\mathbf B$. The
Hodge-dual bivector $I\mathbf B$ is not the object that couples to the
spin density; applying the same duality to the spin gives
$\langle(I\mathbf B)(I\mathbf S)\rangle_0=-\mathbf B\cdot\mathbf S$,
the same coupling, not a new one.

Grade arithmetic therefore does \emph{not} generate singlet--triplet
mixing at linear order in $\mathbf B$. This agrees with the operator
selection rule: $\mathbf B\cdot\mathbf S$ has $\Delta S=0$, and
$\langle s|\mathbf B\cdot\mathbf S|t,m{=}0\rangle=0$ identically,
since $S_z|t,0\rangle=0$ while $S_\pm$ change $m$ at fixed $S$.
Consistently, the rotor action on a scalar kernel is trivial,
$e^{-\mathbf F/2}g_0e^{+\mathbf F/2}=g_0$, so
Eq.~\eqref{eq:rotor} below gives no kernel-level channel mixing for a pure
singlet. Singlet--triplet mixing is mediated by antisymmetric spin-orbit
coupling (or broken inversion)~\cite{Frigeri2004}, which supplies the
non-commuting grade structure at the matrix level; the field then selects
and reorients the mixed components, the commutator $[\mathbf F,G]$
supplying only a necessary algebraic condition, not the physical matrix
element. Within the triplet sector itself, the field does mix the
$m$-components: for triplet kernels the commutator
$[\,I\mathbf B,G_t\,]=I[\mathbf B,G_t]$ is nonzero and redistributes
weight among the $m=0,\pm1$ directions.

\subsection{Uniaxial stress}
\label{sec:stress}

Stress does not carry a grade and produces no selection-rule mixing of its
own. It enters through the deformation potential, modifying the band
dispersion $\xi_{\mathbf k}$ and the shape factors $f(\mathbf k)$, and near a
singlet-triplet-mixing boundary in non-centrosymmetric systems it can
modulate the strength of the mixing. This is a momentum-structure effect,
not a grade mixing.

\subsection{Circularly polarized light (grade 3): no kernel-level
action}
\label{sec:light}

A pseudoscalar drive $\mathbf F=A(t)\,I$ commutes with every multivector,
so the rotor action is trivial,
\begin{equation}
e^{-A(t)I/2}\,G\,e^{+A(t)I/2}=G,
\end{equation}
for \emph{any} kernel. Light therefore does not rotate the pair $(g_0,g_3)$
at the kernel level; the global phase rotation of the condensate is a
transformation of the fermion field $\psi\to e^{I\varphi(t)/2}\psi$
(the gauge remark of Sec.~\ref{sec:math}), not of the kernel. A transient
$d+id$ chirality under drive is a property of the driven normal state and
requires a multiorbital structure; in this framework it is computed from
the Floquet-modified gap equation~\cite{Mitrano2016,Oka2019} rather than
from any rotor on $G$. For
triplet kernels the same conclusion holds at the linear level: the
pseudoscalar drive is rotor-trivial, and any $m$-mixing under light arises
from the same driven-band mechanism.

\subsection{External fields as Clifford rotors and time-dependent dynamics}

The most systematic treatment incorporates the external field as a
Clifford rotor~\cite{Doran2003},
\begin{equation}
\label{eq:rotor}
\begin{split}
R(\mathbf F)=e^{-\mathbf F/2},\qquad\tilde R(\mathbf F)=e^{+\mathbf F/2},\\
\qquad G_{\mathrm{eff}}=R(\mathbf F)\,G\,\tilde R(\mathbf F),
\end{split}
\end{equation}
where $\mathbf F$ is the $Cl_{3,0}$ representation of the external field
(normalized and including the coupling strength); for bivector fields this
is the standard rotor, while for fields of other grades the same formula is
a similarity transformation rather than a rotor in the strict sense.
The weak-field expansion gives
\begin{equation}
\label{eq:rotor-weak}
G_{\mathrm{eff}}\approx G+\frac12\big(G\mathbf F-\mathbf FG\big)+O(\mathbf F^2),
\end{equation}
and the antisymmetrized product in brackets,
$G\mathbf F-\mathbf FG=-[\mathbf F,G]$, is the algebraic mixing
coefficient: a closed kernel-level formula once the field representation is
specified. Whether that algebraically allowed mixing is physically activated
remains model- and material-dependent. Being a commutator, it vanishes
identically on the center (scalar and pseudoscalar kernels).
Time-dependent external fields then drive the evolution of
the grade structure of $G(t)$ through the same formula, with
\begin{equation}
H_{\mathrm{pair}}(t)=\sum_{\mathbf k,\mathbf k'}\langle\psi^\dagger(t)
\psi^\dagger(t)\,G(\mathbf k,\mathbf k',t)\,\psi(t)\psi(t)\rangle_{0,\mathrm{spin}} ,
\end{equation}
The standard dynamical routes (Heisenberg equations, imaginary- and
real-time path integrals, Dyson-equation iterations~\cite{Peskin1995}) apply
directly; the added value is that the self-energy decomposes naturally into
eight grade components in spin space (rather than the four of the
traditional Hermitian restriction), retaining the dynamical evolution of the
anti-Hermitian components. Equation~\eqref{eq:rotor-weak} applies to
kernels with non-scalar content; for a pure singlet kernel $G=g_0$ the
commutator vanishes and the rotor acts trivially,
$G_{\mathrm{eff}}=G$. Optically induced superconductivity is then analyzed
through the driven gap equation of Sec.~\ref{sec:light}: the grade
decomposition exposes which kernel components the drive can and cannot
activate, but the activation itself is a property of the driven fermion
problem, not of a rotor on $G$.

\paragraph{Magnetic field (axial vector).} No bare singlet--triplet
mixing (selection rule $\Delta S=0$; rotor trivial on scalar kernels);
physical mixing requires antisymmetric spin-orbit coupling; within
triplets, mixes the $m$-components.

\paragraph{Uniaxial stress (no grade).} Enters through the deformation
potential, modifying $\xi_{\mathbf k}$ and shape factors; near a
singlet-triplet boundary it can modulate the mixing strength.

\paragraph{Circularly polarized light (grade 3).} Rotor-trivial at the
kernel level; the global phase is a fermion-field transformation;
transient chirality requires a multiorbital structure and the Floquet
route.

\section{Experimental diagnostic protocol}
\label{sec:diagnosis}

This section builds an algebraic diagnostic protocol that organizes
hypotheses about the effective pairing kernel from experimental
observations. The idea is that different experimental probes are sensitive
to different projected components; by cross-comparing multiple probes one
can constrain plausible kernel-level structures, subject to probe modeling.
Section~\ref{sec:tomography} then formulates a conditional reconstruction
algorithm.

\subsection{Diagnostic logic}

The effective kernel $G\in Cl_{3,0}$ has eight real components, and a single
probe is usually sensitive only to a restricted set of projected directions.
In standard response models, ARPES measures the momentum structure of the
gap function (and, with spin resolution, can constrain the sectors); the NMR
Knight shift is sensitive to spin polarization, so a singlet
($\mathrm{span}\{1, I\}$) is commonly associated with a drop below $T_c$,
while a triplet (grades $1\oplus2$) can leave it unchanged or rising; the
polar Kerr effect is sensitive to time-reversal breaking, which may be
carried by a relative phase between kernel singlet-plane components; an
upper critical field $H_{c2}$ far above the Pauli limit is consistent with
equal-spin triplet texture or strong spin-orbit locking, but is not by
itself proof of grade-3 content; ultrasound probes lattice-pairing coupling
and shape-factor modulation by external fields; and thermal transport or
specific heat constrains node structure and momentum texture~\cite{Sigrist1991}.
The Pauli scale uses the standard Clogston--Chandrasekhar
estimate~\cite{Clogston1962,Chandrasekhar1962}. These mappings are not unique material
diagnoses; each entry below is a kernel-level hypothesis that requires a
probe-response model.

\begin{table*}[t]
\caption{\label{tab:diagnosis}Algebraic diagnostic protocol: kernel-level hypotheses.}
\begin{ruledtabular}
\begin{tabular}{@{}p{0.38\textwidth}p{0.54\textwidth}@{}}
Observation & Kernel-level diagnosis\\
\hline
Kerr + full gap + no magnetic order & $g_0,g_3$ both active, relative phase
$\pm\pi/2$ ($d+id$-type kernel)\\
Knight shift drops & singlet ($\mathrm{span}\{1, I\}$ dominant)\\
Knight shift flat or rising & triplet (grades $1\oplus2$ dominant)\\
Multiple superconducting transitions & multi-grade activation, competing
orders\\
$H_{c2}$ far above Pauli & equal-spin triplet texture ($\mathbf d\perp\mathbf B$)
or strong spin-orbit locking; not by itself proof of grade-3 content\\
Triplet features under stress & near singlet-triplet mixing boundary (SOC),
stress via shape factor\\
Transient chirality under light & driven multiorbital bands ($g_0,g_3$
both active in the driven kernel)\\
\end{tabular}
\end{ruledtabular}
\end{table*}

\subsection{Example: symmetry diagnosis of an unconventional superconductor}

For an unconventional superconductor with antiferromagnetic spin
fluctuations, the glue is the spin channel (grade-1 vertex), and after the
Fierz rearrangement it enters the singlet sector, with the $\chi$ peak near
$\mathbf Q=(\pi,\pi)$. In the idealized, strictly $\delta$-peaked separable
limit, the block decomposition selects the shift-odd condition
$\Delta(\mathbf k+\mathbf Q)=-\Delta(\mathbf k)$; the standard
square-lattice realization is $d_{x^2-y^2}$
($\Delta(\mathbf k)\propto\cos k_x-\cos k_y$), with the extended-$s$
degeneracy lifted only by subleading kernel structure (Sec.~\ref{sec:block})
~\cite{Scalapino1986}. The model-level response expectations then follow in
the usual way: line nodes of $d_{x^2-y^2}$ pairing give a $T^2$ specific
heat and a linear low-temperature penetration depth in the clean nodal
limit, while its sign change has been tested by phase-sensitive SQUID
interferometry and corner-junction/tricrystal
experiments~\cite{Wollman1993,Tsuei2000}. The framework
therefore organizes the grade structure and plausible responses before the
full gap equation is solved; quantitative confirmation still requires that
equation and the material band structure.

\subsection{Example: RuO$_2$ as an exploratory extrapolation}

Rutile RuO$_2$ was long considered a Pauli-paramagnetic metal and has
recently been proposed as a candidate altermagnet~\cite{Smejkal2022}, a claim that has
generated debate. We offer this only as an exploratory extrapolation of the
framework's collective-channel bookkeeping, with an explicit methodological
statement. The diagnostic object of the framework is the projected
collective fluctuation channel; extending it to static magnetic order is
exploratory, and the following criteria are heuristic. From the electronic
structure of RuO$_2$: the 4$d$ orbitals of Ru are relatively extended, so it
is hard to form stable local moments and the vector-glue channel (grade 1)
lacks effective fluctuations; the octahedral crystal field splits the
$t_{2g}$ orbitals fully, the orbital degeneracy near the Fermi level is low,
and the bivector-glue channel (grade 2) lacks effective fluctuations; the
rutile structure has an inversion center, so the pseudoscalar channel
(grade 3) is closed; and the electron-phonon coupling is weak, so the
scalar-glue channel (grade 0) is inactive. With all four grade channels
``quiet,'' the heuristic grade accounting suggests electronic behavior
dominated by the Fermi-surface topology and weak scattering, with no
complex magnetic order and transport following semiclassical scaling. This
agrees with the experiments reported in Ref.~\cite{Peng2025} (no magnetic transition below 400 K,
transport consistent with semiclassical simulations). This caveat is
essential: this is postdiction, not prediction, and the \emph{a priori}
criteria that would distinguish RuO$_2$ from a genuine altermagnet, as well
as the boundary ``no fluctuations $\neq$ no static order,'' lie outside the
scope of this paper.

\subsection{Tomographic reconstruction: algorithm and condition number}
\label{sec:tomography}

The CT-scan idea of Sec.~\ref{sec:diagnosis} is here developed into
a conditional linear-response algorithm. With kernel components $g_A$ and
probe multivectors
$M_\alpha\in Cl_{3,0}$ (the ``measurement directions'' of
Sec.~\ref{sec:diagnosis}), linear
response gives
\begin{equation}
\begin{split}
m_\alpha=\chi_0\langle GM_\alpha\rangle_0+\delta m_\alpha
=\sum_AP_{\alpha A}g_A+\delta m_\alpha, \\
P_{\alpha A}=\chi_0\,\varepsilon_A(M_\alpha)_A .
\end{split}
\end{equation}
where $\varepsilon_A$ is the blade signature of Lemma~\ref{lem:ortho}
(the derivation is given in Supplemental Material Sec.~\ref{sec:s4}). For $N=8$ and full
rank, $\hat g=P^{-1}m$, and

\begin{equation}
\frac{\|\delta g\|_2}{\|g\|_2}\le\kappa_2(P)\frac{\|\delta m\|_2}{\|m\|_2},\qquad
\kappa_2(P)=\frac{\sigma_{\max}(P)}{\sigma_{\min}(P)}.
\end{equation}
The feasibility criterion is $\kappa_2(P)\lesssim\mathrm{SNR}\times$ (target
resolution); for $N>8$ the pseudoinverse is used and the noise averaging
makes the error scale as $\kappa_2/\sqrt N$.

\textbf{Structure theorem.} Write $P=P_0+E$ with $P_0$ the grade-resolved
ideal and $E$ the leakage. By Weyl's inequality,
\begin{equation}
\sigma_{\min}(P)\geq\sigma_{\min}(P_0)-\|E\|_2,
\end{equation}
so the feasibility condition is $\|E\|_2<\sigma_{\min}(P_0)$: the condition
number is controlled by the largest leakage and the weakest covered
direction of the probe set, not by the average sensitivity.

\textbf{Toy model.} Take the kernel in the $s$-wave constant limit, and
eight probes coupling to the Hermitian sector $X=\{1,e_1,e_2,e_3\}$ (Knight
shift and gap amplitude) and the anti-Hermitian sector
$Y=\{I,e_{23},e_{31},e_{12}\}$ (relaxation processes and phase response),
with symmetric leakage $\varepsilon$ (spin-orbit coupling and dissipation):
\begin{equation}
\begin{split}
P=\begin{pmatrix}\mathbf1_4&\varepsilon\mathbf1_4\\ \varepsilon\mathbf1_4&\mathbf1_4\end{pmatrix},
\kappa_2(P)=\frac{1+\varepsilon}{1-\varepsilon},\\
\mathrm{spec}(P)=\{1+\varepsilon\,(\times4),\,1-\varepsilon\,(\times4)\}. \\
\end{split}
\end{equation}
The feasibility region follows from
$\varepsilon\le(\eta\,\mathrm{SNR}-1)/(\eta\,\mathrm{SNR}+1)$; for a target
relative error $\eta=10^{-2}$ and SNR $=10^3$, reconstruction to percent
accuracy survives leakage up to $\varepsilon\approx0.82$. In the limit
$\varepsilon\to1$ (probes completely grade blind) $\kappa_2$ diverges, which
gives a quantitative statement of why at least one probe sensitive to each
Hodge pair of blades is needed.

\textbf{Rotor optimization.} Using the rotor language of
Sec.~\ref{sec:fields}, $M_\alpha(\theta)=e^{-\theta F/2}M_\alpha
e^{\theta F/2}$: the first-order sector mixing can be removed by the rotor
angle ($\varepsilon_{\mathrm{eff}}=O(\varepsilon^2)$), and equiangular
designs make $\{M_\alpha\}$ approach an orthonormal set, $\kappa_2\to1$.

\textbf{Boundaries.} The linearization holds in the linear-response regime
(the growth of $\chi_0$ near $T_c$ can improve the SNR, but the same
critical fluctuations narrow the linear-response window, so the optimum is
near but not inside the critical region); for momentum-resolved
kernels $P$ is block diagonal and the analysis applies block by block; far
from the linear regime iterative algorithms are needed, whose convergence is
still set by $\kappa_2$. This section is an algorithm demonstration, not an
experimental protocol; modeling the probe multivectors $M_\alpha$ for
concrete instruments is application-side work.

\section{Dimensional extension and the Weyl-semimetal sign duality}
\label{sec:extension}

The logic of Theorems~\ref{thm:diag}--\ref{thm:block}, representation
decomposition, Casimir sign rules, trace form, and scalar projection, does
not depend on the specific form of $\mathfrak{su}(2)$. The same two-level
rule applies: first choose the projected symmetry space, here a Nambu/Dirac
block or a higher-spin/orbital block, and only then use geometric products
and grade projection. The extension proceeds along two paths: the
representation ladder enlarges the group and the representation, while
dimensional extension changes the Clifford algebra itself.

\subsection{Representation ladder}

Level L1: two spin-$S$ fermions couple to $J=0,1,\ldots,2S$, and the
Casimir eigenvalues $\mathbf S_1\cdot\mathbf S_2=[J(J+1)-2S(S+1)]/2$
increase monotonically with $J$; the sign rule is form-invariant:
antiferromagnetic-type glues favor the lowest $J$ and ferromagnetic ones the
highest. For $S=3/2$ the channels are $J=0,1,2,3$ with eigenvalues
$-15/4,-11/4,-3/4,+9/4$, and the $J=2$ quintet (even parity, five
components) is a channel that does not exist in the spin-$1/2$ world,
proposed as a candidate pairing state in half-Heusler superconductors such
as YPtBi~\cite{Brydon2016}.

Level L2: with $N$ orbitals the pairing matrix is a $2N\times2N$ complex
matrix, and the product of the signs of the spin and orbital Casimirs
decides whether a channel is attractive or repulsive. For $N=2$ the
even-parity sector contains an even-parity triplet (spin triplet times
orbital antisymmetric) forbidden in the single-band world; its attraction
requires the glue to carry opposite signs in the two sectors. This supplies
a restrictive algebraic coupling condition, while actual occurrence remains
a model- and material-level question.

Level L3: for SU($N$)-symmetric interactions~\cite{Cazalilla2009,Gorshkov2010}
the exchange operator has eigenvalues $-(N+1)/2N$ (antisymmetric
representation) and $+(N-1)/2N$ (symmetric representation), reducing at
$N=2$ to the sign rule of $\mathfrak{su}(2)$. The three four-component
fermion systems, two orbitals times spin-$1/2$, spin-$3/2$, and SU(4), share
one 16-dimensional pairing space with its 6/10 parity split: the half-Heusler
quintet, the multiorbital even-parity triplet, and the SU(4) cold-atom
antisymmetric sextet are three slices of the same pairing space.

\subsection{Center drift}

\begin{proposition}[Center drift]
\label{prop:centerdrift}
In an $n$-dimensional Clifford algebra, the pseudoscalar satisfies
$I_ne_a=(-1)^{n-1}e_aI_n$ and $I_n^2=(-1)^{n(n-1)/2}\det\eta$. Hence in odd
dimensions $I_n$ is central, while in even dimensions $I_n$ anticommutes
with the basis vectors, does not lie in the center, and the center of the
algebra is $\mathbb R$.
\end{proposition}
\begin{proof}
In $I_ne_a$, moving $e_a$ from the right end to the left end exchanges it
once with each of the $n-1$ distinct basis vectors, giving the sign
$(-1)^{n-1}$. For the second claim, reversing the order of the factors gives
$(-1)^{n(n-1)/2}$, and the product of the squares of the basis vectors is
$\det\eta$.
\end{proof}

Center drift is the algebraic mechanism of the grade-merger picture of
Sec.~\ref{sec:channels}. In odd dimensions $I_n$ is central
(Proposition~\ref{prop:centerdrift}); in $Cl_{3,0}$ moreover $I^2=-1$, so
$\{1,I\}\cong\mathbb C$ becomes a complex structure
internal to the algebra, and the grade duality becomes a phase
identification: grades 0 and 3 merge, and grades 1 and 2 merge, into the two
complex channels (the structural root of Theorem~\ref{thm:channels}). In
even dimensions, $I$ anticommutes with the basis vectors, does not lie in
the center, and $\{1,I\}$ does not form a complex structure for the whole
algebra (the complex structure is carried by the complexification
$Cl\otimes\mathbb C$); the grade duality does not reduce to a phase
identification, and the grades are independent. The complex-channel merger
of three-dimensional pairing theory is a
property of the center in odd dimensions, made visible in the
higher-dimensional view of this section.

\subsection{Blade-square formula and the four-dimensional signature}

\begin{theorem}[Blade-square formula]
\label{thm:blade}
Let $Cl_{p,q}$ have metric $\eta$ and let $e_A=e_{a_1}\wedge\cdots\wedge
e_{a_r}$ be a grade-$r$ blade. Then
\begin{equation}
e_A^2=(-1)^{r(r-1)/2}\prod_{i=1}^r\eta_{a_ia_i}.
\end{equation}
\end{theorem}
\begin{proof}
In $e_A^2$, move the last $r$ factors in front of the first $r$; reversing
the order of $r$ objects costs $(-1)^{r(r-1)/2}$ pairwise exchanges, and
then each $e_{a_i}^2=\eta_{a_ia_i}$ gives the metric content.
\end{proof}

The signature distribution decomposes into the product of two
sources: the grade factor cycles with period 4 (grades 0, 1 positive; 2, 3
negative; 4, 5 positive; and so on), and the metric content carries the
causal information of whether the blade contains the time direction. In
three-dimensional Euclidean space the metric content is identically $+1$
and the signature is decided by the grade alone (Lemma~\ref{lem:ortho}); in
four-dimensional Euclidean $Cl_{4,0}$ the signature is still a clean grade
pattern $(+,+,-,-,+)$; once the metric is indefinite, the grades split
immediately by causal content, as in Table~\ref{tab:4d}.

\begin{table*}[t]
\caption{\label{tab:4d}Blade-square distribution of the four-dimensional
Minkowski algebra $Cl_{1,3}$ (mostly-minus, $\eta=\mathrm{diag}(+1,-1,-1,-1)$),
i.e. the spacetime algebra (STA) of Ref.~\cite{Hestenes1966}.}
\begin{ruledtabular}
{\small\setlength{\tabcolsep}{3pt}
\begin{tabular}{c|c|c|l}
Grade & Blades & Square & Causal meaning\\
\hline
0 & $1$ & $+1$ & scalar\\
1 & $\gamma_0$ & $+1$ & timelike\\
1 & $\gamma_{1,2,3}$ & $-1$ (3) & spacelike\\
2 & $\gamma_0\gamma_i$ (boost) & $+1$ (3) & contains time\\
2 & $\gamma_i\gamma_j$ (rotation) & $-1$ (3) & purely spatial\\
3 & contains-time blades & $-1$ (3) & contains time\\
3 & $\gamma_{123}$ & $+1$ & spatial volume element\\
4 & $I_4$ & $-1$ & spacetime volume ($\det\eta=-1$)\\
\end{tabular}
}
\end{ruledtabular}
\end{table*}

In total 6 positive and 10 negative blades. With the mostly-plus convention
the odd-grade blade squares all flip (10 positive, 6 negative overall); a
global sign flip of the metric multiplies the square of a grade-$r$ blade by
$(-1)^r$, so the even-grade squares are convention-independent: boosts are
always $+1$, rotations always $-1$, and $I_4^2=-1$ in both conventions. The
split within a grade is the fingerprint of the Minkowski signature at the
level of blades, encoding the causal structure of spacetime into the
signature of the pairing form.

\subsection{Grade-2 Lorentz structure}

\begin{proposition}[Grade 2 as the Killing form of $\mathfrak{so}(1,3)$]
The grade-2 sector of the four-dimensional Minkowski algebra is the Lorentz
algebra $\mathfrak{so}(1,3)$: the three rotations $\gamma_i\gamma_j$ are the
compact generators and the three boosts $\gamma_0\gamma_i$ the noncompact
ones. The signature $(3,3)$ of the scalar projection restricted to grade 2,
negative on rotations and positive on boosts, agrees with the Killing
form~\cite{Peskin1995} of $\mathfrak{so}(1,3)$ (compact directions
negative, noncompact positive),
and the existence of the split is convention-independent: in both
mostly-minus and mostly-plus conventions the signs of the boost and rotation
squares are always opposite.
\end{proposition}
\begin{proof}
This is the blade-by-blade square of Theorem~\ref{thm:blade} (the
blade-square formula); the qualitative statement about compact and
noncompact directions does not depend on which is positive.
\end{proof}

This is the four-dimensional upgrade of the trace-form identification of
Sec.~\ref{sec:signature}. Combining with center drift ($I_4$ not central, no
merger), the grade-2 sector is genuinely independent in four dimensions: a
six-component Lorentz-tensor sector that does not reduce to the imaginary
part of any complex channel, with an internal signature encoding the
boost/rotation split, a structure that does not exist in the three-dimensional
world (which has no boosts). Physically, the 16 blades of $Cl_{1,3}$
correspond one-to-one to the 16 bilinear covariants of the Dirac spinor
$\{S,V,T,A,P\}$ (counting $1,4,6,4,1$)~\cite{Peskin1995}, and the trace
signs of the $i$-dressed bilinears are the blade signs with the $i$-dressing
flip; for blade-diagonal pairing kernels of relativistic fermions
(Dirac/Weyl semimetal pairing, color superconductivity), the signature
distribution of this section supplies the algebraic sign factor in each
Lorentz sector. The full interaction sign still depends on the vertex
dressing and kernel. The next three sections work out a concrete case.

\subsection{Sign duality in Weyl-semimetal pairing}
\label{sec:weyl}

A single Weyl node $h_0=\chi v\mathbf k\cdot\boldsymbol\sigma$ with
chirality $\chi=\pm1$~\cite{Hosur2014,Bednik2015} doubles, under the Nambu
construction, to a
$4\times4$ BdG Hamiltonian whose operator space is the 16 blades of
$Cl_{1,3}$. The effective kernel expands in blades, and the GA trace
orthogonality (Theorem~\ref{thm:blade}) gives
$\langle e_Ae_B\rangle_0=\varepsilon_A\delta_{AB}$ with
$\varepsilon_A=e_A^2$ (Table~\ref{tab:4d}). The classification of pairing
channels by total angular momentum $J$~\cite{Bednik2015,Meng2012} maps onto
the blade sectors as follows: the intra-node $^{1}S_0$ channel ($J=0$)
corresponds to the timelike direction of grade 1; the inter-node $^{3}P$
channel ($J=1$) is an axial $J=1$ multiplet, whose gap-equation sign is fixed
only after the off-diagonal Nambu/Fierz dressing (Theorem~\ref{thm:reversal}). Grade 2 decomposes as $B$ (boost, three blades, $\varepsilon=+1$) plus
$R$ (rotation, three blades, $\varepsilon=-1$); under spatial rotations the
six constant grade-2 blades therefore form $\mathbf B\oplus\mathbf R=3\oplus3$,
two $J=1$ multiplets rather than a quintet. The singlet and quintet live in
the bilinears of these two SO(3) vectors: the symmetric sector $B_iR_j$
contains the singlet $\mathbf B\cdot\mathbf R$ (the SO(3) singlet of a boost
polar vector and a rotation axial vector, which has no counterpart in
nonrelativistic BCS theory and corresponds, in 3D Dirac materials, to the
pseudoscalar and tensor orders~\cite{Faraei2017}) and the
symmetric-traceless quintet (the $1\oplus5$ of the symmetric sector), while
the antisymmetric cross product $\mathbf B\times\mathbf R$ forms the axial
$J=1$ structure. The $^{5}D$ quintet ($J=2$) of Weyl-node pairing itself
requires momentum-dependent $l=2$ form factors and does not reside in the
constant grade-2 sector; the grade-2 blades supply the boost/rotation
pairings and the algebraic building blocks---singlet, quintet, axial
vector---from which the higher-$J$ channels are assembled.

\begin{theorem}[Fierz sign duality, blade-diagonal minimal model]
\label{thm:duality}
For a node-symmetric, blade-diagonal kernel, the linearized gap equation of
a Weyl-node BdG Hamiltonian diagonalizes in the blade basis with instability
condition
\begin{equation}
\begin{split}
1=\varepsilon_A\,g_A\,\chi_0(T),\\
\chi_0(T)=\sum_{\mathbf k}\frac{\tanh(E_{\mathbf k}/2T)}{2E_{\mathbf k}}>0
\ \text{(the pair susceptibility)}.
\end{split}
\end{equation}
Channels with $\varepsilon_A=+1$ require attractive bare coupling, and
channels with $\varepsilon_A=-1$ require repulsive bare coupling. Among all
16 channels, the boost triple and the rotation triple form the unique
SO(3)-covariant pair with opposite signs.
\end{theorem}
\begin{proof}
The gap equation in the blade basis reads
$1=g_A\langle e_Ae_A\rangle_0\chi_0(T)$ with
$\langle e_Ae_A\rangle_0=\varepsilon_A$; uniqueness of the covariant pair
follows by inspecting Table~\ref{tab:4d}: the only two SO(3) triples of
blades with opposite $\varepsilon$ within a single grade are the boost and
rotation blades (in grade 1 the opposite-sign partners are the spacelike
triplet $\gamma_i$ and the timelike singlet $\gamma_0$; grade 3 splits as
$3+1$; grades 0 and 4 contain no triplet).
\end{proof}

\begin{corollary}[Sign reversal and pairing induction in the blade-diagonal
minimal model]
Within the same node-symmetric, blade-diagonal linearized model, a
rotationally invariant interaction $g(|B|^2+|R|^2)$ presents the same
bare coupling to both sectors, and the instability conditions
$1=\varepsilon_Ag\,\chi_0(T)$ then require opposite signs: the boost
sector ($\varepsilon=+1$) condenses for $g>0$, the rotation sector
($\varepsilon=-1$) for $g<0$, so a single bare coupling supports exactly one
of the two. A coupling attractive for the boost sector is repulsive for the
rotation sector; it removes the rotation-sector instability rather than
destabilizing an existing rotation condensate---a Kohn--Luttinger-type reversal whose
direction is fixed by the blade squares alone, with no momentum geometry
required. Geometrically enhanced Kohn--Luttinger mechanisms are
established~\cite{Jahin2024}; the increment of this corollary is at the
sign level.
\end{corollary}

\begin{remark}[Three dimensions lack this duality]
The sign duality needs $+1$ blades (boosts) and $-1$ blades (rotations)
coexisting, an indefinite metric (the causal split of
Table~\ref{tab:4d}). The three blades of grade 2 in the Euclidean
three-dimensional algebra all have the same sign (Lemma~\ref{lem:ortho}), and the duality
disappears. It is a structure specific to relativistic pairing theory, the
direct physicalization of ``three dimensions has no boosts.''
\end{remark}

\paragraph{Prediction (appearance reversal, convention-conditional).}
Suppose a Weyl superconductor is reported to be driven by a boost-type
attractive component while the measured condensate lands in the
rotation/axial sector.
Then: (i) the Knight shift does not drop (the paired state is an odd-parity
triplet type), opposite to the ``$s$-wave appearance'' of the bare vertex;
(ii) the nodes on the Fermi surface carry the angular structure of the axial
channel ($J=2$ quintet: point nodes; $J=1$: equatorial nodes), read off from
ARPES, thermal transport, or the low-temperature power laws; (iii) a
parity-sensitive Josephson junction or $\mu$SR combination makes the
discrimination falsifiable. Conversely, if $^{1}S_0$-type pairing is
confirmed experimentally, the driving component must lie in the
$\varepsilon=+1$ sector: parity observation, in this setting, is observation
of the Fierz sign. The connection to the general-Nambu-matrix quantum
geometry of Ref.~\cite{Simon2026}: the node structure of prediction (ii) can
be verified by specializing their explicit formula; the added value of the
framework is the ``why'' (the sign duality) rather than the ``what'' (the
spectral computation). Under the gap-equation convention of
Theorem~\ref{thm:duality}, a bare boost attraction instead condenses in the
boost sector; the reversed appearance therefore diagnoses an additional
vertex-layer $i$-dressing/Fierz sign or a misidentified bare channel.
\par

\subsection{Two-node extension: chirality reversal}
\label{sec:two-node}

\subsubsection{Minimal model}

Take the minimal time-reversal-breaking two-node
model~\cite{Hosur2014,Cho2012}
\begin{equation}
\begin{split}
H_0(\mathbf k)=v\,k_x\sigma_x+v\,k_y\sigma_y+M(k_z)\sigma_z,\qquad\\
M(k_z)=m_0-bk_z^2\quad(m_0,b>0),
\end{split}
\end{equation}
with Weyl nodes at $k_z=\pm k_0=\pm\sqrt{m_0/b}$ and low-energy dispersions
$h_\pm=vq_x\sigma_x+vq_y\sigma_y\mp v_zq_z\sigma_z$ with $v_z=2bk_0$.
(Lattice regularization $M\to m_0-t\cos k_z$ does not change any conclusion
below.) The chiralities are
\begin{equation}
\begin{split}
\chi_\pm=\mathrm{sgn}\det\frac{\partial d_a}{\partial k_i}\Big|_{\pm k_0}
=\mathrm{sgn}\big(v^2M'(\pm k_0)\big)=\mp1,\\
\chi_+\chi_-=-1 .
\end{split}
\end{equation}
The BdG problem is $8\times8$, with pairing blocks $\Delta_{++}$,
$\Delta_{--}$ (intra-node) and $\Delta_{+-}$ (inter-node).

\subsubsection{Chirality-insertion lemma and sign-reversal theorem}

\begin{lemma}[Chirality insertion, factorizing channels]
The numerator of the massless Weyl propagator is linear in $\chi$,
$G_\chi\propto(i\omega+\chi v\,\mathbf q\cdot\boldsymbol\sigma)/
(\omega^2+v^2q^2)$. For channels in which $\Gamma_A$ has a definite parity
under the chirality conjugation relating the two nodes and the frequency and
momentum integrals are symmetric, the one-loop pair susceptibility in
channel $A$ connecting nodes $i$ and $j$ factorizes as
\begin{equation}
\Pi^{ij}_A(i\Omega,\mathbf Q)=\eta_A\,\chi_i\chi_j\,
\Pi^{(0)}_A(i\Omega,\mathbf Q),
\qquad \eta_A=\pm1,
\end{equation}
where the bare pair bubble $\Pi^{(0)}_A$ is independent of chirality and
$\eta_A$ is a channel sign
fixed by the $\Gamma_A$ trace. The known inter-node reversal corresponds to
$\eta_A=+1$; channels with $\eta_A=-1$ do not reverse.
\end{lemma}
\begin{proof}
The one-loop diagram
$\langle\Gamma_AG_{\chi_i}\Gamma_AG_{\chi_j}\rangle_0$ contains one
propagator from each node. Under chirality conjugation the two Weyl Hamiltonians differ
by a $\gamma_5$-type map; when $\Gamma_A$ has definite parity under that map
the spinor trace reduces to the common scalar trace times $\eta_A\chi_i\chi_j$.
Without that definite parity the trace need not factorize, and the sign must
be computed channel by channel.
\end{proof}

\begin{theorem}[Inter-node sign reversal, factorizing minimal model]
\label{thm:reversal}
For factorizing channels ($\eta_A=+1$) the instability condition is
$1=\chi_i\chi_j\,\varepsilon_A\,g_A\,\chi_0(T)$.
Those inter-node channels reverse sign relative to the single-node case:
$\varepsilon=+1$ inter-node channels are driven by repulsion and
$\varepsilon=-1$ by attraction. The four quadrants are: intra-node, boost:
attraction; intra-node, rotation: repulsion; inter-node, boost: repulsion;
inter-node, rotation: attraction.
\end{theorem}
\begin{proof}
Insert the lemma into the gap equation and use
$\langle e_Ae_A\rangle_0=\varepsilon_A$.
\end{proof}

For the inter-node $^{3}P$ channel, whose $\Gamma_A$ has the required
definite parity, the literature assignment corresponds to the
$\varepsilon_{\rm eff}=+1$ sector after the off-diagonal dressing: then
$\chi_i\chi_j\varepsilon_{\rm eff}=-1$ and repulsive interactions pair the
inter-node $^{3}P$ channel~\cite{Bednik2015,Meng2012,Hosur2014,Cho2012,Li2018,Wei2014}.
With the opposite dressing the sign flips, so the $J$ label of
Table~\ref{tab:4d} alone is not a sign criterion. Within the factorizing
minimal-model assumptions of the lemma and theorem, the sign criterion
reduces the microscopic input to blade and chirality factors: inter-node-
dominated pairing then requires bare couplings of sign opposite to
single-node intuition. Outside those assumptions the sign must be computed
channel by channel. The four-quadrant structure deserves one more sentence:
it is invariant under the combined operation of swapping the intra/inter-node
structure and flipping the interaction sign---a period-2 ($\mathbb Z_2$)
operation, which is the operator-level content of the chirality product
$\chi_i\chi_j=-1$.

\subsubsection{Pseudo-spin classification and explicit inter-node blocks}

Fix the convention
$\Psi_{\mathbf k}=(c_{+\mathbf k},c_{-\mathbf k},i\sigma_2c^\dagger_{+\mathbf k}{}^{\mathsf T},
i\sigma_2c^\dagger_{-\mathbf k}{}^{\mathsf T})^{\mathsf T}$
(the time-reversed hole basis $i\sigma_2c^{\dagger\mathsf T}$ of
Refs.~\cite{Qi2009,FuBerg2010}, concatenated over the two nodes)
and the pairing term
$H=\sum_{\mathbf k}[\Delta_{ij,\alpha\beta}(\mathbf k)c^\dagger_{i\mathbf k\alpha}
c^\dagger_{j,-\mathbf k\beta}+\mathrm{h.c.}]$.
Fermionic antisymmetry gives the single constraint
\begin{equation}
\hat\Delta_{ij}(\mathbf k)=-\hat\Delta_{ji}{}^{\mathsf T}(-\mathbf k).
\end{equation}
For the inter-node block $\hat\Delta_{+-}$ this constraint does not involve
$\hat\Delta_{+-}$ itself: at fixed $\mathbf k$, $\hat\Delta_{+-}(\mathbf k)$
is an arbitrary $2\times2$ complex matrix, all eight real directions are
allowed, and $\hat\Delta_{-+}$ is completely determined by the constraint
(the derivation is given in Supplemental Material Sec.~S3). Classifying the
pair by node pseudo-spin (treating the two nodes as a pseudo-spin),
momentum parity, and physical-spin parity gives the following
symmetry-adapted blocks, with $q_\pm=q_x\pm iq_y$. The even-parity sector
is spanned by $E_1=i\sigma_2$ (node pseudo-triplet, spin singlet, $J=0$)
and $E_{2,3,4}=\{\mathbf 1,\sigma_x,\sigma_3\}$ (node pseudo-singlet,
spin triplet, $J=1$). The odd-parity sector is spanned by
$O_1=q_zi\sigma_2$ (node pseudo-singlet, spin singlet) and
$O_2=q_+\mathbf 1$, $O_3=q_x\sigma_x+q_y\sigma_y$, $O_4=q_-\mathbf 1$
(node pseudo-triplet, spin triplet). In every allowed block only the
product of the node, momentum, and spin exchange parities is $-1$; the
individual node-triplet and spin-triplet factors are symmetric. The
$^{3}P$ sector $O_{2,3,4}$ is the inter-node topological channel of the
literature. These odd-parity matrices span the $l=1$, $S=1$ sector but are
not individually pure total-$J$ eigenstates; resolving $J=0,1,2$ requires
Clebsch--Gordan recoupling. In the standard Knight-shift response, the
physical spin-singlet blocks $E_1$ and $O_1$ are the ones expected to show
a drop.

A structural note explains why the grade arithmetic is done at the bubble
level: the $\gamma$-matrix conventions of the two nodes differ by the
chirality factor (the $\chi=-1$ node differs by a $\gamma_5$-type
conjugation), so there is no consistent blade label for a cross-node vertex
pair. The vertex traces keep $\varepsilon_A$; the chirality enters only the
propagator product $\chi_i\chi_j$. This is the design reason for the
chirality-insertion lemma.

\subsubsection{Closed-form spectra and node charge}

With $\varepsilon_{\mathbf q}=\sqrt{v^2q_\perp^2+v_z^2q_z^2}$, the $\mu=0$
spectra follow by squaring the $4\times4$ block,
\begin{equation}
H_A^2=\begin{pmatrix}
\varepsilon^2\mathbf 1+DD^\dagger & M\\
M^\dagger & \varepsilon^2\mathbf 1+D^\dagger D
\end{pmatrix},
M=h_+D+Dh_-^{\mathsf T},
\label{eq:master}
\end{equation}
where $D\equiv\hat\Delta_{+-}(\mathbf q)$. The plus sign in $M$ is fixed by
Hermiticity: $D^\dagger$ sits below the diagonal of the BdG block, and for
the singlet block $D=\Delta_0i\sigma_2$ is anti-Hermitian (since
$(i\sigma_2)^\dagger=-i\sigma_2$), so $D^\dagger=-D$ there; for $O_3$ the block $D$ is Hermitian, while $O_{2,4}$
have complex-scalar $D$. The argument only uses that for all four channels
$DD^\dagger$ and $D^\dagger D$
are scalar, so $E^2=\varepsilon^2+\lambda\pm m_M$ with
$MM^\dagger=m_M^2\mathbf 1$. The results, with derivations and
direction-by-direction checks
in Supplemental Material Sec.~S6, are collected in
Table~\ref{tab:spectra}.
\begin{table*}[t]
\caption{\label{tab:spectra}Closed-form $\mu=0$ spectra of the inter-node
blocks from Eq.~\eqref{eq:master}. The node column lists the zeros of the
lower branch $E_-$.}
\begin{ruledtabular}
{\small
\begin{tabular}{p{2.9cm}p{5.5cm}p{8cm}}
Block & $E(\mathbf q)$ & Node structure at $\mu=0$\\
\hline
$E_1$ (inter-node $^{1}S_0$) &
$\sqrt{v^2q_\perp^2+(v_z|q_z|\pm\Delta_0)^2}$
& origin gapped; two point nodes at $|q_z|=\Delta_0/v_z$\\
$O_1$ & $\sqrt{v^2q_\perp^2+(v_z\pm\Delta_1)^2q_z^2}$
& point node at origin; $z$-axis velocity renormalized to
$|v_z-\Delta_1|$; line node along $\hat z$ only at the fine-tuned
$\Delta_1=v_z$\\
$O_3$ ($^{3}P$) &
$\sqrt{\varepsilon^2+\Delta^2q_\perp^2\pm2\Delta q_\perp
\sqrt{v^2q_x^2+v_z^2q_z^2}}$
& node at origin, anisotropic in-plane velocity; two nodal lines in the
$x$--$z$ plane if $\Delta>v$\\
$O_2,O_4$ &
$\sqrt{v^2q_\perp^2+v_z^2q_z^2+\Delta^2q_\perp^2\pm2\Delta vq_\perp|q_x|}$
& node at origin; line nodes along $\hat x$ only at $\Delta=v$\\
\end{tabular}
}
\end{ruledtabular}
\end{table*}

Each entry reproduces direct diagonalization at $\mathbf q\parallel\hat x$
and $\mathbf q\parallel\hat z$. The node charge is bookkeeping consistent
with the known Nambu doubling: the Chern number of the BdG Hamiltonian on a
sphere surrounding a Weyl point is twice the original Chern number, and the
number of superconducting point nodes equals twice the Weyl-point Chern
number~\cite{Sato2017,Meng2012}. The odd-parity zero-momentum BCS state with
point nodes and Majorana/Fermi arcs in doped Weyl metals was obtained by
Bednik \emph{et al.}~\cite{Bednik2015}, with systematic treatments of node
charges in Refs.~\cite{Hosur2014,Cho2012,Li2018,Wei2014,Wu2022}. The
new results of this section are the closed-form $\mu=0$ spectra of the four
representative inter-node blocks in Table~\ref{tab:spectra}, the node splitting below, and
the accounting formula tying them to the sign duality.

\begin{proposition}[Node-charge accounting, minimal model]
For an inter-node block $(i,j)$, take a sphere surrounding $\mathbf q=0$
that encloses \emph{all} superconducting nodes of the block---equivalently,
take $\Delta\to0$ before shrinking the sphere; for the $E_1$ block at
finite $\Delta_0$ the split nodes $q_z=\pm\Delta_0/v_z$ must lie inside.
The Chern number of the negative-energy bands on the sphere, with the
convention
$F=\tfrac12\langle\hat n\,\partial\hat n\,\partial\hat n\,
I^{-1}\rangle_0$, is
\begin{equation}
C_{ij}=\chi_j-\chi_i,
\end{equation}
in the Berry-curvature sign convention of Theorem~\ref{thm:chern},
independent of the pairing channel and coupling strength while the
enclosed node set is unchanged. For opposite chiralities $|C|=2$,
consistent with Nambu doubling~\cite{Sato2017,Meng2012}; the formula
extends the accounting to node pairs of arbitrary chirality.
\end{proposition}
\begin{proof}
At $\Delta\to0$ the negative-energy bundle is the direct sum of the
valence band of $h_i$, with Chern number $-\chi_i$, and the valence band
of $h_j^{\mathsf T}$, whose $\hat n$ map differs by a transpose and has
Chern number $+\chi_j$. Shrinking the sphere after $\Delta\to0$, or
keeping all split nodes inside, closes no gap at finite $\mathbf q$
(Table~\ref{tab:spectra}), so the sum is invariant.
\end{proof}

\begin{result}[Node splitting of inter-node $^{1}S_0$ pairing, minimal model]
\label{res:splitting}
The inter-node $^{1}S_0$ spectrum
$E=\sqrt{v^2q_\perp^2+(v_z|q_z|\pm\Delta_0)^2}$ leaves the origin gapped and
places the two zeros of the lower branch at $q_z=\pm\Delta_0/v_z$: each
Weyl node splits into a pair of point nodes symmetrically displaced along
the bond axis. This is the explicit $\mu=0$ realization, for inter-node
singlet pairing, of the node splitting established by Meng and
Balents~\cite{Meng2012}. The anomalous correlator, contracted with
$(i\sigma_2)^\dagger$ as the gap equation requires, is nonzero for generic
$\mathbf q$ (Supplemental Material Sec.~S6), so the gap equation closes
self-consistently and the state is stable at the mean-field level.
\end{result}

Within the minimal model, the discrimination closes as a set of consistency
checks: bare-coupling signs give the four-quadrant channel selection, the
Knight shift constrains the spin sector, and the node structure constrains
$J$ and topology. The blade algebra and closed-form spectra organize these
checks; numerical spectroscopy and band-structure input remain necessary
outside the minimal assumptions.

\subsection{A look back at the cross-domain grammar}

The computational rules of the framework, scalar projection, grade
decomposition, and the diagonal pairing form, are dimension-independent; the
output drifts with two structural parameters: the parity of the center
(Proposition~\ref{prop:centerdrift}) and the signature of the metric
(the metric content of the blade-square formula). The same method reports
the complex-channel merger in odd dimensions and the causal pattern in even
dimensions: the three-dimensional complex-channel merger and the
four-dimensional independence of grade 2 are two faces of one signature
structure under different dimensional parameters.

\section{Cross-domain transfer and methodology}
\label{sec:transfer}

\subsection{The scalar-projection method as a scalar extractor for physics}

The scalar-projection method is portable across domains once
a stage and a scalar observable have been defined: an overall geometric
product is followed by grade-0 projection. Concrete instances differ in the
choice of the stage, the Clifford-algebra container that acts as the main
symmetry space; the remaining spaces then enter as coefficients after
projection. Table~\ref{tab:transfer} collects the instances.

\begin{table*}[t]
\caption{\label{tab:transfer}Cross-domain instances of the
scalar-projection method.}
\begin{ruledtabular}
{\small\setlength{\tabcolsep}{4pt}
\begin{tabular}{lll}
Problem & Stage (main space) & Scalar coefficients\\
\hline
Superconducting pairing & $Cl_{3,0}$ spin space & momentum, orbit, lattice,
fields\\
Magnetism & vector part of $Cl_{3,0}$ & lattice sites, $J_{ij}$\\
Electron-phonon & charge density (grade 0) & phonon displacement\\
Moir\'e superconductors & $Cl_{3,0}$ spin space & moir\'e real-space
modulation\\
Dirac-fermion pairing & $Cl_{1,3}$ spinor space & momentum, gauge field,
mass\\
\end{tabular}
}
\end{ruledtabular}
\end{table*}

\subsection{Relation to Landau order-parameter theory}

The framework is complementary to Landau theory~\cite{Sigrist1991}, as
summarized in Table~\ref{tab:landau}. Landau theory inputs the symmetry of
the order parameter (an external assumption), produces even expansions (a
symmetry assumption), and introduces gradient terms phenomenologically; it
applies to all continuous phase transitions at the thermodynamic level. The
present framework \emph{constrains and organizes} candidate order-parameter
symmetries from the grade structure and momentum structure of the projected
kernel; the actual selected state is fixed by the gap equation or free
energy. The even-power expansion is a corollary of central-rotor symmetry
(the parity theorem is its algebraic echo), and amplitude-phase gradients
acquire an algebraic parameterization through the rotor decomposition,
while their coefficients remain dynamical. In scope, Landau theory is
thermodynamic, while the present framework constrains the microscopic
algebraic structure of spin-$1/2$ pairing symmetry.

\begin{table*}[t]
\caption{\label{tab:landau}The framework relative to Landau theory.}
\begin{ruledtabular}
\begin{tabular}{lll}
Question & Landau theory & This framework\\
\hline
Order-parameter symmetry & input & kernel-constrained candidates\\
Even expansion & symmetry assumption & corollary of rotor symmetry\\
Gradient terms & phenomenological & algebraic rotor/phase parameterization\\
Scope & all continuous transitions & spin-$1/2$ pairing algebra\\
\end{tabular}
\end{ruledtabular}
\end{table*}

\section{Conclusion and outlook}
\label{sec:conclusion}

This paper has reformulated the theory of superconducting pairing in the
GA $Cl_{3,0}$. Its organizing rule is two-level projection:
(i) symmetry-space matrix blocks are represented as Clifford multivectors,
with blade-level algebraic statements made for real coefficients while
residual operator content is carried in ordered coefficients;
(ii) geometric products of those projected blocks are reduced by
$\langle\cdot\rangle_0$ to Hamiltonian scalars. Raw Fock operators remain carriers on Fock space: only their symmetry-space blocks are projected, the residual operator content being carried in ordered coefficients, and a grade-0 scalar before mean-field expectation may still be Fock-operator-valued. This separation distinguishes algebraic
consequences from model-level or comparative statements. The conclusions
are stated in numbered form.

(1) \textbf{Parameterization and Fierz.} The eight-real-dimensional basis of
$Cl_{3,0}$ provides real coordinates for the space of $2\times2$ complex
pairing matrices; the fermionic identity $B^\dagger(V)=-B^\dagger(V^{\mathsf T})$
(Lemma~\ref{lem:anti}) is the root of channel dichotomy, and the Fierz map
$\Phi:G\mapsto Ge_{31}$ makes it explicit. The parameter space is
equivalent to the standard complex-matrix formulation, but it is an algebraic reconstruction rather than a change of representation: it proceeds from the geometric-algebra scalar-projection method to the same physical content.

(2) \textbf{Complex channels} (kernel level): grades $0\oplus3$ excite spin
singlets (even parity) and grades $1\oplus2$ excite triplets (odd parity)
(Theorem~\ref{thm:channels}); at the pairing-matrix level the singlet plane
is $\mathrm{span}\{e_2,e_{31}\}$
(Corollary~\ref{cor:matrix-level}). The kernel-level dictionary aligns
with the center drift of Sec.~\ref{sec:extension}: the singlet is
the center, and the triplet is three Hodge pairs.

(3) \textbf{BCS position.} BCS is the grade-0 kernel line,
$\mathrm{span}\{1\}$, with the pairing vertex on $e_{31}$ (grade 2): the
double grade identity of a scalar glue and a bivector vertex
(Sec.~\ref{sec:bcs}).

(4) \textbf{Signature.} The $(+,+,-,-)$ form is the $\mathfrak u(2)$ trace
form; the Casimir spectrum $-3/+1$ gives the exchange-sign assignments; the
block decomposition connects them to $d$-wave/$p$-wave model limits; the
one-loop RG feeding constant and the Casimir sign are the same number $3$
(Remark~\ref{rem:three}); and the glue-generation picture of Sec.~\ref{sec:glue} together with the
criteria of Sec.~\ref{sec:diagnosis} organizes the framework's internal logic.

(5) \textbf{Quantum-geometry export.} Theorems~\ref{thm:split}
and~\ref{thm:null} give the null-texture criterion for the pairing-sector
quantum metric; applied to the kagome flat band they motivate a comparative
stiffness discussion (Sec.~\ref{sec:kagome}) addressed to the
CsV$_3$Sb$_5$ controversy~\cite{Duan2021,Saykin2023,Xu2022,Farhang2023,
Yoshida2025}; the UTe$_2$ analysis (Sec.~\ref{sec:ute2}) provides the
singlet-leakage bound, the equal-spin criterion, and a stress-switch
prediction within the minimal GL model~\cite{Ran2019,Aoki2019,Nakamine2021,Girod2022}. The relative
phase of the kernel singlet plane is simultaneously the object of topology
(Chern number), quantum geometry (phase-texture metric), and diagnostics
(Kerr, tomography); it is the connecting thread of these applications.

(6) \textbf{Sign duality.} Under the stated blade-diagonal and minimal
two-node assumptions, the chirality-reversal theorem, the node-charge
accounting $C_{ij}=\chi_j-\chi_i$, the closed-form $\mu=0$ spectra, and the
node splitting of the inter-node singlet channel recast several
Weyl-semimetal pairing structures through blade signs~\cite{Bednik2015,Meng2012,Hosur2014,Cho2012,Li2018,Wei2014,Wu2022}.

(7) \textbf{Tomography.} The conditional linear-response reconstruction has
an error bound with $\kappa_2=(1+\varepsilon)/(1-\varepsilon)$ in the toy
probe model, together with rotor-optimized design examples.

(8) \textbf{Dimensional extension.} Center drift, the blade-square formula,
the causal four-dimensional signature, and the $(3,3)$ grade-2 split as the
Killing form of $\mathfrak{so}(1,3)$.

(9) \textbf{External fields and diagnosis.} External fields act through
the adjoint (rotor) action, which is trivial on the singlet kernel plane
(the center of $Cl_{3,0}$): no field mixes a pure singlet kernel,
singlet--triplet mixing is mediated by antisymmetric spin-orbit coupling
through the gap equation, vector/axial-vector fields redistribute the
$m$-components within the triplet sector, a pseudoscalar drive is
rotor-trivial at the kernel level (light-induced chirality is a property
of the driven multiorbital normal state), and stress acts through the
shape factor; Table~\ref{tab:diagnosis} collects the qualitative
diagnostics.

The scalar-projection method transfers across domains once the symmetry
stage and scalar observable are chosen: the remaining spaces enter as
coefficients after projection. Superconductivity, magnetism,
electron-phonon systems, and Dirac-fermion pairing can then share one
syntax; the spin$\otimes$isospin, SU($N$), and spin$\otimes$color
extensions follow the same pattern and are left to future work.

Outlook: tomographic reconstruction from measured data; the DMFT interface,
mapping two-particle vertices onto $Cl_{3,0}$ with extraction of the kernel
from the dynamical susceptibility; the reduction of the integer ($\mathbb Z$)
node charge of Dirac-semimetal pairing blocks to blade and chirality
factors; and experimental assessment of the kagome
quantum-geometry criterion.

\begin{acknowledgments}
The author thanks Professors Hai-Bing Xia, Xin-Bing Huang, and especially Professor Yang-Sheng Xu for their help.
\end{acknowledgments}

\begingroup
\setcounter{section}{0}
\renewcommand{\thesection}{S\arabic{section}}
\setcounter{equation}{0}
\renewcommand{\theequation}{S\arabic{equation}}

\begin{center}
{\large\bfseries Supplemental Material}
\end{center}
\medskip

\section{Blade-by-blade action of the Fierz map}
\label{sec:s1}
The Fierz map $\Phi:G\mapsto Ge_{31}$ with $e_{31}=e_3\wedge e_1$,
$e_{31}^2=-1$, acts on the basis blades as follows:
\begin{equation}
\begin{split}
1\mapsto e_{31},\quad I\mapsto-e_2,\quad e_1\mapsto-e_3,\quad e_2\mapsto I,\quad\\
e_3\mapsto e_1,\quad e_{12}\mapsto e_{23},\quad e_{23}\mapsto-e_{12},\quad
e_{31}\mapsto-1 .
\end{split}
\end{equation}
All entries follow from $e_ie_j+e_je_i=2\delta_{ij}$, $I^2=-1$,
$Ie_i=e_{jk}$ (cyclic), and $e_{31}e_{31}=-1$. In particular
$\Phi(\mathrm{span}\{1,I\})=\mathrm{span}\{e_{31},e_2\}$ and
$\Phi(\{e_i,e_{ij}\})=\{1,e_1,e_3,e_{12},e_{23},I\}$ up to signs, which is
the content of the kernel-level channel dictionary.

\section{Structure constants of the algebraic RG}
\label{sec:s2}
The fermion-loop trace of bilinears reduces to the cyclic scalar projection:
under $e_i\leftrightarrow\sigma_i$, $I\leftrightarrow i$ one has
$\frac12\mathrm{Tr}\,A=\langle A\rangle_0+i\langle A\rangle_3$ for
$A\in Cl_{3,0}$, and the contractions used below pair blades within a
single grade sector, whose products carry no grade-3 part, so
$\frac12\mathrm{Tr}_{\mathrm{loop}}[\Gamma_1\Gamma_2]
=\langle\Gamma_1\Gamma_2\rangle_0$ for all of them. The products required
for the feeding constants are $e_ie_j=\delta_{ij}+e_{ij}$ (for $i\neq j$),
$\langle e_{ij}e_{kl}\rangle_0=\delta_{il}\delta_{jk}-\delta_{ik}\delta_{jl}$,
and $I^2=-1$. Hence
$C^0_{00}=\langle1\rangle_0=1$;
$C^0_{11}=\sum_i\langle e_ie_i\rangle_0=3$;
$C^0_{22}=\sum_{i<j}\langle e_{ij}e_{ij}\rangle_0=-3$;
$C^0_{33}=\langle I\cdot I\rangle_0=-1$. The overall sign of the beta
function is fixed by the convention $\ell\to$ infrared with $V>0$
denoting attraction, Eq.~\eqref{eq:rg-v0} of the main text. All other two-point contractions
vanish by grade orthogonality (Lemma~\ref{lem:ortho} of the main text). The three-point
constants that would complete the beta system (for example the self-contraction
of the triplet texture component) require a careful definition of the
three-point channel projectors and are left to future work; the robust
conclusions of the main text use only the four constants above.

\section{Derivation of the inter-node constraint and blocks}
\label{sec:s3}
Starting from
$H_{\mathrm{pair}}=\sum_{\mathbf k}\sum_{ij}\Delta_{ij,\alpha\beta}(\mathbf k)
c^\dagger_{i\mathbf k\alpha}c^\dagger_{j,-\mathbf k\beta}+\mathrm{h.c.}$,
relabel $(i,\mathbf k,\alpha)\leftrightarrow(j,-\mathbf k,\beta)$ and use
$c^\dagger_{i\alpha}c^\dagger_{j\beta}=-c^\dagger_{j\beta}c^\dagger_{i\alpha}$:
\begin{equation}
\begin{split}
\sum_{\mathbf k}\Delta_{ij,\alpha\beta}(\mathbf k)c^\dagger_{i\mathbf k\alpha}
c^\dagger_{j,-\mathbf k\beta}
=-\sum_{\mathbf k}\Delta_{ij,\beta\alpha}(-\mathbf k)c^\dagger_{j\mathbf k\beta}
c^\dagger_{i,-\mathbf k\alpha}\\
=-\sum_{\mathbf k}\Delta_{ji,\beta\alpha}(-\mathbf k)c^\dagger_{i\mathbf k\alpha}
c^\dagger_{j,-\mathbf k\beta}.
\end{split}
\end{equation}
Therefore
$\hat\Delta_{ij}(\mathbf k)=-\hat\Delta_{ji}{}^{\mathsf T}(-\mathbf k)$, and
the $+-$ block is unconstrained at fixed $\mathbf k$. For a node
pseudo-singlet, $\hat\Delta_{-+}(\mathbf k)=-\hat\Delta_{+-}(\mathbf k)$,
combined with the constraint, gives
$\hat\Delta_{+-}(\mathbf k)=\hat\Delta_{+-}{}^{\mathsf T}(-\mathbf k)$: even
spin-triplet, odd spin-singlet. The pseudo-triplet cases follow with the
opposite sign. The explicit blocks $E_{1..4}$ and $O_{1..4}$ listed in the
main text span the even- and odd-parity sectors compatible with these
assignments. The even sector decomposes directly as $J=0\oplus J=1$; the
odd matrices span the $l=1$, $S=1$ sector and generally require
Clebsch--Gordan recoupling into definite $J=0,1,2$ channels. The $^{5}D$
sector follows at $l=2$.

\section{Proof of the tomography structure theorem}
\label{sec:s4}
For $P=P_0+E$ with $P_0$ diagonal (grade-resolved probes), Weyl's inequality
gives $|\sigma_{\min}(P)-\sigma_{\min}(P_0)|\leq\|E\|_2$, hence
$\sigma_{\min}(P)\geq\sigma_{\min}(P_0)-\|E\|_2$: reconstruction is feasible
whenever the leakage norm stays below the weakest resolved coupling. For the
toy model, $P=\mathbf1_8+\varepsilon\sigma_x\otimes\mathbf1_4$ (in sector space) has
eigenvalues $1\pm\varepsilon$, each fourfold, so
$\kappa_2=(1+\varepsilon)/(1-\varepsilon)$; solving
$\kappa_2\leq\eta\,\mathrm{SNR}$ gives the stated feasibility region.

\section{Proofs of Lemma 3 and Theorem 4 of the main text}
\label{sec:s5}
\paragraph{Lemma 3 (antisymmetry identity).}
$B^\dagger(V)=\sum_{\mathbf k}\sum_{\alpha\beta}V_{\alpha\beta}
c^\dagger_{\mathbf k\alpha}c^\dagger_{-\mathbf k\beta}$.
Relabel $\mathbf k\to-\mathbf k$, then anticommute:
\begin{equation}
\begin{split}
B^\dagger(V)=-\sum_{\mathbf k}\sum_{\alpha\beta}V_{\alpha\beta}
c^\dagger_{\mathbf k\beta}c^\dagger_{-\mathbf k\alpha}
=-\sum_{\mathbf k}\sum_{\alpha\beta}V_{\beta\alpha}\\
c^\dagger_{\mathbf k\alpha}c^\dagger_{-\mathbf k\beta}
=-B^\dagger(V^{\mathsf T}).
\end{split}
\end{equation}

\paragraph{Theorem 4 (complex-channel theorem).}
Survival requires $V=Ge_{31}$ antisymmetric (Lemma~\ref{lem:anti} of the main text). The
antisymmetric $2\times2$ complex matrices form the complex line
$\mathbb C(i\sigma_2)$, i.e.\ the real plane $\mathrm{span}\{e_2,e_{31}\}$.
From the blade table (S1), $G=-Ve_{31}$ for $V=ae_2+be_{31}$ gives
$G=-aI+b\in\mathrm{span}\{1,I\}$; conversely every $a+bI$ maps into the
plane. Each grade-1/2 blade maps outside:
$e_1e_{31}=-e_3$, $e_2e_{31}=I$, $e_3e_{31}=e_1$, $e_{12}e_{31}=e_{23}$,
$e_{23}e_{31}=-e_{12}$, $e_{31}e_{31}=-1$, all symmetric-matrix directions,
and no linear combination of them returns to the antisymmetric plane. Hence
the singlet eigenspace is $\mathrm{span}\{1,I\}$ and the triplet eigenspace
is its orthogonal complement, the six blades $e_i,e_{ij}$. For (iii), $I$ is
central, so $e^{I\varphi}$ acts within $\mathrm{span}\{1,I\}$ as a rotation
of $(g_0,g_3)$, and within the triplet sector it pairs $e_i$ with $e_{jk}$
by the Hodge map.

\section{Derivation of the spectra and the anomalous
correlator}
\label{sec:s6}
\subsection{Master formula}
With $h_\pm=vq_x\sigma_x+vq_y\sigma_y\mp v_zq_z\sigma_z$ and
$D\equiv\hat\Delta_{+-}(\mathbf q)$,
\begin{equation}
\begin{split}
H_A=\begin{pmatrix}h_+ & D\\ D^\dagger & h_-^{\mathsf T}\end{pmatrix},
H_A^2=\begin{pmatrix}\varepsilon^2\mathbf 1+DD^\dagger & M\\
M^\dagger & \varepsilon^2\mathbf 1+D^\dagger D\end{pmatrix},\\
M=h_+D+Dh_-^{\mathsf T}.
\end{split}
\end{equation}
The plus sign in $M$ follows from $H_A^2{}_{eh}=h_+D+DB$ with
$B=h_-^{\mathsf T}$; note that $D^\dagger$ (not $-D$) sits below the
diagonal, and that for the singlet block $D=\Delta_0i\sigma_2$ is
anti-Hermitian, $(i\sigma_2)^\dagger=-i\sigma_2$. For all four channels
below, $DD^\dagger$ and $D^\dagger D$ are scalar and equal, so
$E^2=\varepsilon^2+\lambda\pm m_M$ with $MM^\dagger=m_M^2\mathbf 1$.

\subsection{The four gap matrices}
We use $\sigma_xi\sigma_2=-\sigma_3$, $\sigma_yi\sigma_2=i\mathbf 1$,
$\sigma_zi\sigma_2=\sigma_1$ and
$(i\sigma_2)\sigma_x=\sigma_3$, $(i\sigma_2)\sigma_y=i\mathbf 1$,
$(i\sigma_2)\sigma_z=-\sigma_1$.

\paragraph{$E_1$: $D=\Delta_0i\sigma_2$.}
\begin{equation}
\begin{split}
h_+i\sigma_2=-vq_x\sigma_3+ivq_y\mathbf 1-v_zq_z\sigma_1,\\
i\sigma_2h_-^{\mathsf T}=vq_x\sigma_3-ivq_y\mathbf 1-v_zq_z\sigma_1,
\end{split}
\end{equation}
so $M=-2\Delta_0v_zq_z\sigma_1$, $MM^\dagger=4\Delta_0^2v_z^2q_z^2\mathbf 1$,
and $E^2=\varepsilon^2+\Delta_0^2\pm2\Delta_0v_z|q_z|
=v^2q_\perp^2+(v_z|q_z|\pm\Delta_0)^2$.

\paragraph{$O_1$: $D=\Delta_1q_zi\sigma_2$.}
$M=-2\Delta_1v_zq_z^2\sigma_1$ and
$E^2=v^2q_\perp^2+(v_z\pm\Delta_1)^2q_z^2$.

\paragraph{$O_3$: $D=\Delta(q_x\sigma_x+q_y\sigma_y)$ (Hermitian).}
Direct multiplication gives
\begin{equation}
M=2\Delta\big[vq_x^2\mathbf 1+i\big(v_zq_zq_y\sigma_x-v_zq_zq_x\sigma_y
-vq_xq_y\sigma_z\big)\big],
\end{equation}
\begin{equation}
\begin{split}
MM^\dagger=4\Delta^2\big(v^2q_x^4+v_z^2q_z^2q_\perp^2+v^2q_x^2q_y^2\big)\mathbf 1\\
=4\Delta^2q_\perp^2\big(v^2q_x^2+v_z^2q_z^2\big)\mathbf 1,
\end{split}
\end{equation}
so $E^2=\varepsilon^2+\Delta^2q_\perp^2\pm2\Delta q_\perp
\sqrt{v^2q_x^2+v_z^2q_z^2}$. The lower branch vanishes only if $q_y=0$
and $v_z^2q_z^2=(\Delta^2-v^2)q_x^2$; for $\Delta>v$ the zeros form two
nodal lines in the $x$--$z$ plane.

\paragraph{$O_2,O_4$: $D=\Delta q_\pm\mathbf 1$.}
Since $h_++h_-^{\mathsf T}=2vq_x\sigma_x$,
\begin{equation}
\begin{split}
M=2\Delta vq_\pm q_x\sigma_x,\qquad\\
MM^\dagger=4\Delta^2v^2q_\perp^2q_x^2\mathbf 1,
\end{split}
\end{equation}
and $E^2=v^2q_\perp^2+v_z^2q_z^2+\Delta^2q_\perp^2\pm2\Delta vq_\perp|q_x|$.
A zero away from the origin requires $(v-\Delta)^2q_\perp^2\le0$, so line
nodes appear only at the fine-tuned $\Delta=v$ along $\hat x$.

\subsection{Direct-diagonalization checks}
At $\mathbf q\parallel\hat x$ each block reduces to
$\big(\begin{smallmatrix}vq\sigma_i & \Delta q\sigma_j\\
\Delta q\sigma_j & vq\sigma_i\end{smallmatrix}\big)$-type structures whose
eigenvalues $vq\pm\Delta q$ reproduce the table; at $\mathbf q\parallel\hat z$,
$M$ is proportional to $\sigma_1$ or vanishes, again reproducing the table.

\subsection{Anomalous correlator of the $E_1$ block}
The gap equation contracts the anomalous block with $(i\sigma_2)^\dagger$.
Writing $(H_A^2)^{1/2}=\alpha\mathbf 1+\beta N$ with
$N^2=4\Delta_0^2v_z^2q_z^2\mathbf 1_4$, $2\alpha\beta=1$, and
$\alpha,\beta>0$ (the positive-energy branch),
\begin{equation}
\mathrm{Tr}\big[(i\sigma_2)^\dagger(\mathrm{sign}\,H_A)_{eh}\big]
=\frac{2\Delta_0\big(\alpha^2-v_z^2q_z^2\big)}
{\alpha\big(\alpha^2-4\beta^2\Delta_0^2v_z^2q_z^2\big)},
\end{equation}
which is even in $\mathbf q$ and nonzero on the generic Fermi surface, so
the Brillouin-zone sum in the gap equation is nonzero and a self-consistent
$\Delta_0$ exists.

\endgroup

\end{document}